\documentclass[aps,pra,reprint,notitlepage,superscriptaddress,nofootinbib]{revtex4-2}

\usepackage{graphicx}                            
\usepackage{amsmath}
\usepackage{amssymb}
\usepackage{amsfonts}
\usepackage{amsthm}
\usepackage{float}
\usepackage{mathtools}
\usepackage{enumerate}
\usepackage{hyperref}
\usepackage[normalem]{ulem}
\usepackage{tikz}
\usetikzlibrary{decorations.pathmorphing}

\usepackage{pagecolor,lipsum}

\hypersetup{colorlinks=true,urlcolor=[rgb]{0,0,0.5},citecolor=[rgb]{0.5,0,0},linkcolor=[rgb]{0,0,0.4}}

\newtheorem{theorem}{Theorem}
\newtheorem{conjecture}{Conjecture}
\newtheorem{corollary}{Corollary}
\newtheorem{lemma}{Lemma}
\newtheorem{definition}{Definition}
\newtheorem{proposition}{Proposition}

\def\autorefapp#1{\hyperref[#1]{Appendix~\ref{#1}}}

\newcommand{\bra}[1]{\left\langle #1 \right|}
\newcommand{\ket}[1]{\left|#1\right\rangle}

\def\ketbra#1{ |{#1}\rangle\!\langle{#1}| }
\def\vev#1{\langle{#1}\rangle}

\def\ni{\noindent}
\def\BE{\mathbb{E}}
\def\op{{\cal O}}
\def\ep{\varepsilon}
\DeclareMathOperator*{\EV}{\mathbb{E}}
\DeclareMathOperator*{\var}{\text{var}}
\DeclareMathOperator*{\Tr}{\text{Tr}}

\definecolor{myTan}{rgb}{0.941,0.934,0.879}

\newcommand\hl[1]{{\color{black}{#1}}}

\begin{document}

\title{Random quantum circuits anti-concentrate in log depth}

\author{Alexander M. Dalzell}
\affiliation{Institute for Quantum Information and Matter, Caltech, Pasadena, CA 91125}
\author{Nicholas Hunter-Jones}
\affiliation{Perimeter Institute for Theoretical Physics, Waterloo, ON N2L 2Y5}
\author{Fernando G. S. L. Brand{\~a}o}
\affiliation{Institute for Quantum Information and Matter, Caltech, Pasadena, CA 91125}
\affiliation{AWS Center for Quantum Computing, Pasadena, CA 91125}

\begin{abstract}
We consider quantum circuits consisting of randomly chosen two-local gates and study the number of gates needed for the distribution over measurement outcomes for typical circuit instances to be anti-concentrated, roughly meaning that the probability mass is not too concentrated on a small number of measurement outcomes. Understanding the conditions for anti-concentration is important for determining which quantum circuits are difficult to simulate classically, as anti-concentration has been in some cases an ingredient of mathematical arguments that simulation is hard and in other cases a necessary condition for easy simulation. Our definition of anti-concentration is that the expected \textit{collision probability}, that is, the probability that two independently drawn outcomes will agree, is only a constant factor larger than if the distribution were uniform. We show that when the two-local gates are each drawn from the Haar measure (or any two-design), at least $\Omega(n \log(n))$ gates (and thus $\Omega(\log(n))$ circuit depth) are needed for this condition to be met on an $n$ qudit circuit. In both the case where the gates are nearest-neighbor on a 1D ring and the case where gates are long-range, we show that $O(n \log(n))$ gates are also sufficient, and we precisely compute the optimal constant prefactor for the $n \log(n)$. The technique we employ relies upon a mapping from the expected collision probability to the partition function of an Ising-like classical statistical mechanical model, which we manage to bound using stochastic and combinatorial techniques. 
\end{abstract}

\maketitle

\section{Introduction}

Random quantum circuits (RQCs) are a crucial model for understanding a diverse set of phenomena in both quantum information and quantum many-body physics. They have been used to study the onset of quantum chaos and dynamical spread of entanglement in strongly interacting quantum systems \cite{NahumRuhmanVijayHaah2017EntanglementGrowth,vonKeyserlingk2018OperatorHydrodynamics,NahumVijayHaah2018OperatorSpreading}, including information processing in black holes \cite{HaydenPreskill}.
They also form the basis for recent experiments aiming to demonstrate exponential quantum advantage \cite{Boixo2018,Arute2019GoogleQuantumSupremacy,Wu2021StrongCompAdvantage}. 

The utility of RQCs in these situations derives from a myriad of quantitative properties they have been shown to possess. For example, RQCs quickly generate entanglement \cite{DahlstenOliveiraPlenio2007TypicalEntanglement,OliveiraDahlstenPlenio2007GenericEntanglement,NahumRuhmanVijayHaah2017EntanglementGrowth}, lead to fast scrambling and decoupling of quantum information \cite{BrownFawzi2012RQCScrambling, BrownFawzi2015RQCDecoupling}, and act as efficient encoding circuits for good quantum error-correcting codes \cite{BrownFawzi2013RandomCodes}. When the circuits are geometrically local, they lead to ballistic spreading of local operators \cite{vonKeyserlingk2018OperatorHydrodynamics,NahumVijayHaah2018OperatorSpreading}. Furthermore, they form approximate unitary designs, that is, despite being composed of local gates, they efficiently approximate a global random unitary transformation up to any polynomial number of moments \cite{HarrowLow2009RQC2design,BrownViola2010tdesign,BHH2016RQCtdesign,HarrowMehraban2018tdesign,Hunter-Jones2019StatMechDesign}. Meanwhile, computing transition amplitudes of RQCs has been shown to be just as difficult as for arbitrary quantum circuits \cite{Bouland2019RCSComplexity,Movassagh2018RCSAverageCase,Movassagh2019RCSAverageCaseRobust,Bouland2021NoiseFrontier,Kondo2021ImprovedRobustness}, suggesting that classical simulation of RQCs should require exponential time.  

In this work, we focus on another property of random quantum circuits called \textit{anti-concentration}. Roughly speaking, when we measure the output state of the circuit in the computational basis, anti-concentration is the property that the distribution over measurement outcomes is fairly well spread across all possible outcomes and not too concentrated onto just one or only a small portion of those outcomes. Quantitatively, our definition of anti-concentration depends on the \textit{collision probability}, the probability that measurement outcomes from two independent copies of the circuit agree. An RQC architecture is said to be anti-concentrated if the collision probability is at most a constant factor larger than its minimal value. Understanding when this is the case is particularly important for knowing when RQCs are hard to classically simulate. On the one hand, anti-concentration is a necessary ingredient in most formal hardness arguments for RQC simulation \cite{Aaronson2011BosonSampling,Bremner2016AverageCaseIQP,Morimae2017DQC1HardnessTVD,Bouland2018CCC,Hangleiter2018anticoncentration,Bouland2019RCSComplexity,Dalzell2020HowManyQubits,Morimae2019FGAdditive}. On the other hand, certain classical algorithms for simulating RQCs require anti-concentration in order to be efficient, for example, the algorithms discussed in Refs.~\cite{Bremner2017SparseNoisyIQP,GaoDuan2018NoisySimulation} for noisy circuit simulation and the algorithm in Ref.~\cite{Barak2020Spoofing} that spoofs the linear cross-entropy benchmarking metric introduced in Ref.~\cite{Arute2019GoogleQuantumSupremacy}.

In most previous work where RQC anti-concentration is needed, it has been asserted as an implication of the 2-design property (see, e.g., Refs.~\cite{Hangleiter2018anticoncentration,haferkamp2019closinggaps}).
However, the 2-design property is much stronger than what is required for anti-concentration. It was shown that $n$-qubit RQCs on a fully connected architecture form approximate 2-designs after roughly $O(n)$ depth \cite{HarrowLow2009RQC2design}, and this was later shown to also apply to geometrically local RQCs in 1D and improved to $O(n^{1/D})$ in $D$ spatial dimensions \cite{HarrowMehraban2018tdesign}. However, recent work by Barak, Chou, and Gao \cite{Barak2020Spoofing} --- using a similar method to the one presented here --- showed that for 1D RQCs the collision probability converges in depth $O(\log(n))$, much faster than the 2-design depth of $O(n)$. They also conjectured that 2D RQCs anti-concentrate in depth $O(\sqrt{\log(n)})$.

In this work, we prove sharp bounds on the number of gates needed for anti-concentration in two RQC architectures. For 1D RQCs, we confirm the $O(\log(n))$ upper bound on the anti-concentration depth in Ref.~\cite{Barak2020Spoofing}, and add a lower bound that matches the upper bound even up to the constant prefactor of the $\log(n)$. We also show that an $\Omega(\log(n))$ lower bound on the depth needed for anti-concentration holds regardless of which RQC architecture we use, which refutes the conjecture from Ref.~\cite{Barak2020Spoofing} that 2D RQCs anti-concentrate in $O(\sqrt{\log(n)})$ depth. We then consider a fully connected (i.e.~not geometrically local) RQC architecture, where each gate acts on a pair of qudits chosen randomly among all $n(n-1)/2$ possible such pairs. 
We show that, for qubits (local dimension $q=2$), $5n\log(n)/6$ gates are necessary and sufficient (up to subleading corrections) for anti-concentration to be achieved, which settles a conjecture in Ref.~\cite{HarrowMehraban2018tdesign}.  

Our method employs a technique for analyzing RQCs that converts the collision probability into a weighted sum over bit assignments to each location in the circuit diagram; this weighted sum can be viewed as a partition function for an Ising-like statistical mechanical model. The bit assignments can also be interpreted as a Markov chain, and the number of gates needed for anti-concentration ultimately translates into the time needed for certain expectation values to converge under the dynamics of the Markov chain. This method not only yields sharp quantitative bounds, it also produces an appealing qualitative explanation on how and why the collision probability reaches its limiting value, which allows for effective heuristic reasoning even in architectures that we have not explicitly considered here.

The main takeaways from our work are twofold. First, we show that anti-concentration is generally achieved much faster than the 2-design property. The fact that anti-concentration occurs in $\Theta(n\log(n))$ circuit size both in 1D and for the fully connected architecture --- these being two opposite extremes of geometric locality --- suggests that anti-concentration may require only $\Theta(n\log(n))$ size for \textit{any} reasonably well-connected architecture. This comes in sharp contrast to the situation for unitary designs, where the scaling of the size needed with $n$ is highly dependent on the architecture. Second, the fact that we can prove tight upper and lower bounds suggests a broader utility for our method based on the correspondence between RQCs and statistical mechanical partition functions. 

\section{Anti-concentration and the collision probability}

\begin{figure*}
\centering
\includegraphics[width=0.85\linewidth]{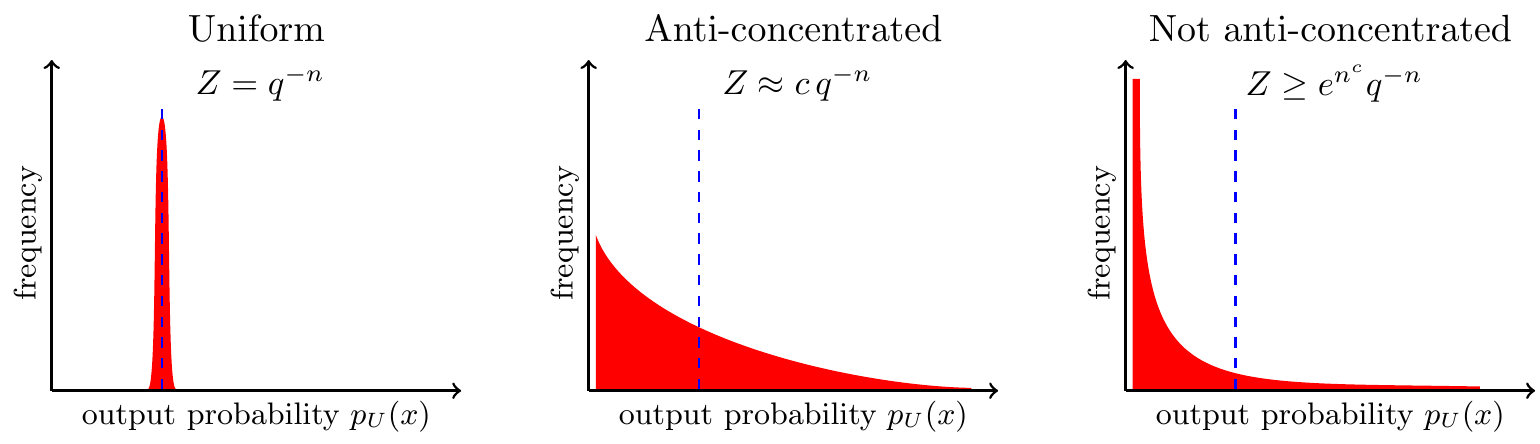}
\caption{\label{fig:ACcartoon}
A caricature of anti-concentration.  \hl{In the three examples, the fraction (frequency) of bit strings $x$ for which $p_U(x) = p$ is plotted against $p$. (This could be either for a fixed choice of $U$ or averaged over random choice of $U$.) Since there are $q^n$ bit strings, the mean of this distribution is $q^{-n}$ (dotted blue line).}
For the uniform distribution, which is completely anti-concentrated, all $q^n$ outcomes are allocated probability mass $q^{-n}$ and the collision probability is $Z=q^{-n}$. For globally Haar-random unitaries, the output probabilities are on average $q^{-n}$ but have some non-zero variance, and the collision probability is $Z\approx 2q^{-n}$. Whenever $Z \approx c q^{-n}$ for some $c$ independent of $n$, we call the distribution anti-concentrated. For low-depth RQCs, the mean output probability is $q^{-n}$, but the variance is much larger, and the collision probability is much larger than $q^{-n}$. Most of the probability mass is concentrated onto a few measurement outcomes, \hl{while the vast majority} of the outcomes \hl{are} assigned a very small amount of mass, leading to \hl{the divergence in the frequency that $p_U(x)$ is close to $0$} depicted in the plot.}
\end{figure*}

A RQC architecture is an instruction set on how to draw a circuit diagram given the number of qudits $n$ (each with local Hilbert space dimension $q$) and the size $s$ of the circuit. The two architectures we consider specifically are the \textit{1D architecture} (with periodic boundary conditions) and the \textit{complete-graph} architecture. The associated RQC ensemble for an RQC architecture is formed by following this instruction set and then choosing the value of each gate in the diagram independently and uniformly at random from the Haar measure. If we fix an instance $U$ from this ensemble, there is an associated output probability distribution $p_U$ over $q^n$ possible computational basis measurement outcomes $x \in [q]^n$, (where $[q] = \{1,2,\ldots,q\}$). Anti-concentration tries to capture the notion that the probability mass is well spread out over all the outcomes. The uniform distribution, where each output is allocated $q^{-n}$ fraction of the total probability mass, is the ultimate anti-concentrated distribution because the mass is exactly equally spread, but we say a distribution is still anti-concentrated as long as the average fluctuations from uniform are no larger than $O(q^{-n})$. This definition is captured precisely by the \textit{collision probability}, which is $\sum_{x}p_U(x)^2$. The collision probability gives the probability that measurement outcomes from two independent copies of the circuit are identical. It is also proportional to the second moment (and thus is related to the variance) of the output probability of a randomly chosen bit string. If $p_U$ is the uniform distribution, then the collision probability is $q^{-n}$, its minimal possible value. For a RQC architecture at a specified qubit number $n$ and circuit size $s$, we consider the collision probability averaged over the randomly chosen circuit instances $U$.
\begin{equation}
    Z := \EV_U\left[\sum_{x \in [q]^n} p_U(x)^2\right] = q^n \EV_U\left[p_U(1^n)^2\right]
\end{equation}
where the second equality holds because by symmetry each of the $q^n$ terms in the sum yields the same number under expectation as long as at least one Haar-random gate acts on each qudit. 

We say a RQC architecture with $n$ qudits and $s$ gates is anti-concentrated if there is a constant $\alpha$ (independent of $n$) with $0 < \alpha \leq 1$ for which $Z \leq \alpha^{-1} q^{-n}$, i.e.~that the collision probability is only a constant factor larger than its minimal value. In particular, our theorem statements roughly correspond to the choice $\alpha=1/4$, but other choices of $\alpha$ would yield the same results up to leading order. If desired, Markov's inequality can then be used to bound the fraction of the randomly chosen $U$ whose collision probability is larger than some constant multiple of $Z$. Moving forward, for convenience, when we say collision probability we will mean the average collision probability $Z$. 

Very shallow circuit architectures are not anti-concentrated: there are expected to be some output probabilities $x$ for which $p_U(x)$ is exponentially larger than the mean of $q^{-n}$. As the circuit gets deeper, we expect the probability distribution to become closer to uniform, but even at infinite depth, when the circuit unitary $U$ becomes a globally Haar-random $q^n \times q^n$ unitary, the output distribution still does not become completely uniform. In this case, the output distribution will typically follow a Porter-Thomas distribution\footnote{In the Porter-Thomas distribution, the frequency at which $p_U(x) = p$ is proportional to $\exp({-p/q^{-n}})$, illustrated roughly in the middle diagram of \autoref{fig:ACcartoon}.} and $Z$ can be exactly computed as
\begin{equation}
    \lim_{s \rightarrow \infty} Z = Z_H := \frac{2}{q^n+1},
\end{equation}
roughly twice as large as the minimal value of $q^{-n}$ associated with the uniform distribution. This statement is proved using the techniques described later. For a graphical illustration of these cases, see \autoref{fig:ACcartoon}.

While one could capture the notion of anti-concentration with a different definition, the definition we choose is useful and relevant because it has concrete ramifications in all of the previously mentioned applications of anti-concentration. For example, one implication of our definition (by application of the Paley-Zygmund inequality) is that if $Z \leq \alpha^{-1}q^{-n}$, then for any $0 \leq \beta \leq 1$
\begin{equation}\label{eq:prnotclosetozeroMAIN}
    \Pr_U[p_U(x) \geq \beta q^{-n}] \geq \hl{\alpha}(1-\beta)^2
\end{equation}
meaning that for at least a constant fraction of the circuit instances the probability of a given measurement outcome $x$ is at least a constant \hl{multiple} $\beta$ of the mean measurement probability $q^{-n}$. This sort of inequality is the relevant one for turning good additive approximations into good multiplicative approximations (with reasonable probability), employed in e.g.~\cite{Aaronson2011BosonSampling,Bremner2016AverageCaseIQP,Morimae2017DQC1HardnessTVD,Bouland2018CCC,Hangleiter2018anticoncentration,Bouland2019RCSComplexity,Dalzell2020HowManyQubits,Morimae2019FGAdditive} to argue\footnote{\hl{Note that, while Eq.~\eqref{eq:prnotclosetozeroMAIN} is an ingredient in these arguments, it is not alone sufficient to imply hardness of simulation as the arguments generally rely on additional unproven conjectures.}} that it is hard to classically sample output distributions for a large fraction of instances up to small total variation distance error (for more details, see \autoref{sec:ACinhardness}). In fact, equations like Eq.~\eqref{eq:prnotclosetozeroMAIN} are sometimes taken to be the definition of anti-concentration \cite{Hangleiter2018anticoncentration}, which is a weaker definition than ours since, in principle, \hl{Eq.~\eqref{eq:prnotclosetozeroMAIN} can hold even in cases where $Z$ exceeds any constant multiple of $q^{-n}$.} 

\section{Our results}

\begin{table*}[ht]
\centering
\normalsize
{\renewcommand{\arraystretch}{1.5}
\begin{tabular}{|c|c|c|}
\hline
Architecture & $s_{AC}$ upper bound & $s_{AC}$ lower bound\\
\hline\hline 
general & $O(n^2)$ & $ \Omega(n\log(n))$\\
\hline
\rule{0pt}{18pt} 1D & ~$\left(2\log\left(\frac{q^2+1}{2q}\right)\right)^{-1}n\log(n) + O(n)$~ & ~$\left(2\log\left(\frac{q^2+1}{2q}\right)\right)^{-1}n\log(n) - O(n)$~ \\[4pt]
\hline
\rule{0pt}{14pt} ~complete-graph~ & $\frac{q^2+1}{2(q^2-1)}\, n\log(n) + O(n)$ & $\frac{q^2+1}{2(q^2-1)}\, n\log(n) -O(n)$ \\[3pt]
\hline
\end{tabular}
}
\caption{\label{tab:results} Summary of results: upper and lower bounds on the circuit size $s_{AC}$ at which anti-concentration is achieved for different random circuit architectures.}
\end{table*}


We show that the collision probability is given by a discrete sum, which we interpret as the expectation value of a certain stochastic process. The correspondence between the collision probability and the discrete sum is described in \autoref{sec:overviewofmethod}, and a complete derivation is provided in \autorefapp{sec:analysisframework}. 

Analyzing our expression, we derive rigorous upper and lower bounds on the collision probability generally and for two specific architectures. These bounds are stated here and the proofs are provided in the appendices. These bounds are then used to form upper and lower bounds quoted in \autoref{tab:results} on the anti-concentration size $s_{AC}$, defined as the minimum circuit size required such that $Z \leq 2Z_H$.  The constant 2 in \hl{the definition of $s_{AC}$} is arbitrary but, \hl{since we will show that $Z$ approaches $Z_H$ as $Z_H(1+e^{-\Omega(s/n)})$}, a different choice of constant would only lead to linear-in-$n$ changes to $s_{AC}$, which would be subleading and would not affect any of the statements in \autoref{tab:results}. All logarithms in this paper are natural logarithms.

\subsubsection{Collision probability upper bounds}

Our upper bounds take the following form:
\begin{equation}\label{eq:upperboundform}
    Z \leq Z_H\left(1+e^{-\frac{2a}{n}(s-s^*)}\right)\,,
\end{equation}
where the constant $a$ is independent of $n$ and depends on the circuit architecture and $s^*$ is a function of $n$ that also depends on the architecture. Thus, if the anti-concentration size $s_{AC}$ is defined to be the minimum size $s$ such that $Z \leq 2Z_H$, then we have $s_{AC} \leq s^*$. Specifically, we have the following results, which are restated here as theorems, and proved rigorously in the Appendices.

First, we consider the \textit{1D architecture} with periodic boundary conditions, where the qudits are arranged on a ring and alternating layers of $n/2$ nearest-neighbor Haar-random gates are applied. 

\begin{theorem}\label{thm:1Dupperboundsummary}
    For the 1D architecture, Eq.~\eqref{eq:upperboundform} holds with
    \begin{align}
        a &= \log\left(\frac{q^2+1}{2q}\right)\label{eq:a1D}\\
        s^* &= \frac{1}{2a}n\log(n) + n\left(\frac{1}{2a}\log(e-1)+\frac{1}{2} \right)
        \label{eq:s*1D}
    \end{align}
    whenever $s \geq s^*$. 
\end{theorem}
Since this depth of the 1D architecture is given by $d:=2s/n$, we can define $d^*:= 2s^*/n = a^{-1}\log(n) + O(1)$ for 1D and conclude that the ``anti-concentration depth'' $d_{AC}$ satisfies $d_{AC} \leq d^*= O(\log(n))$.

Similarly, we show an upper bound for the \textit{complete-graph architecture}, where each gate acts on a random pair of qudits without regard for their spatial proximity. 

\begin{theorem}\label{thm:completegraphupperboundsummary}
    For the complete-graph architecture, Eq.~\eqref{eq:upperboundform} holds with
    \begin{align}
        a &= \frac{(q-1)^2}{2(q^2+1)} \\
        s^* &= \frac{q^2+1}{2(q^2-1)}n \log(n) + cn
    \end{align}
    whenever $s\geq s^*$, for a constant $c$ that is independent of $n$.
\end{theorem}
A size-$s$ circuit diagram chosen randomly from the complete-graph architecture will have depth at most $O(s\log(n)/n)$ with high probability \cite{BrownFawzi2015RQCDecoupling}, meaning that $O(\log(n)^2)$ depth is typically sufficient for anti-concentration in the complete-graph architecture.

We also consider general architectures. We define a property called \textit{regularly connected} \hl{(\autoref{def:regularlyconnected} in \autorefapp{sec:generalbounds})}, which applies to a RQC architecture when for any partition of qubits into two sets, there will be a gate in the circuit that couples the two sets at least once every $O(n)$ gates. Nearly all natural architectures have this property, including standard architectures in $D$ spatial dimensions for any $D$. 
\begin{theorem}\label{thm:generalupperboundsummary}
    If an architecture is regularly connected, then Eq.~\eqref{eq:upperboundform} holds with $a = \Theta(1)$ and $s^* = \Theta(n^2)$
\end{theorem}
This corresponds to $\Theta(n)$ gates per qudit. This result is weaker than our specific result for the 1D and complete-graph architecture, and we conjecture that much better is possible.
\begin{conjecture}\label{con:conjecture}
    Theorem \ref{thm:generalupperboundsummary} can be improved to $s^* = \Theta(n\log(n))$. 
\end{conjecture}

\subsubsection{Collision probability lower bounds}

Our lower bounds on the collision probability take the form
\begin{equation}
    Z \geq \frac{Z_H}{2}\exp\left(A e^{\log(n) - Bs/n}\right)\,,
\end{equation}
for constants $A$ and $B$ that are independent of $n$. (The lower bound for the complete-graph architecture takes a different but very similar form.) This form implies that if $s$ grows with $n$ like $s \approx fn\log(n)/B$ for some $f <1$, then we have $Z/Z_H \geq \frac{1}{2}e^{A n^{1-f}}$, which becomes arbitrarily large as $n \rightarrow \infty$, meaning that the architecture is not anti-concentrated. This puts a lower bound on the ``anti-concentration'' size $s_{AC}$ of $s_{AC} \geq n\log(n)/B - O(n)$. 

Specifically, we show a general lower bound, as well as specific lower bounds for the 1D and complete-graph architectures. 

\begin{theorem}\label{thm:generallowerboundsummary}
For any RQC architecture with $s$ two-qudit gates, the following holds. 
    \begin{equation}
        Z \geq \frac{Z_H}{2} \exp\left(\frac{\log(q)}{q+1}e^{\log(n)-2\log\left(q^2+1\right)s/n}\right)\,.
    \end{equation}
\end{theorem}
    This has the consequence that if $s_{AC}$ and $d_{AC}$ are defined as the minimum size and minimum depth for which $Z \leq 2Z_H$, then
    \begin{align}
        s_{AC}  &\geq (2\log(q^2+1))^{-1} n\log(n) - O(n) \\
        d_{AC} &\geq (\log(q^2+1))^{-1}\log(n) - O(1)\,.
    \end{align}

We improve on the general lower bound for the two specific architectures that we consider.

\begin{theorem}\label{thm:1Dlowerboundsummary}
    For the 1D architecture, there exists a constant $A$ such that
    \begin{equation}
            Z \geq \frac{Z_H}{2} \exp\left(Ae^{\log(n)-2as/n} \right)\,,
    \end{equation}
    where $a=\log((q^2+1)/(2q))$ is the same as for the upper bound in Eq.~\eqref{eq:a1D}.
\end{theorem}
This implies that in 1D,
    \begin{align}
        s_{AC} &\geq (2a)^{-1}n\log(n) - O(n) \\ 
        d_{AC} &\geq a^{-1}\log(n)-O(1)\,,
    \end{align}
which is tight with the upper bound up to subleading corrections.

\begin{theorem}\label{thm:completegraphlowerboundsummary}
    For the complete-graph architecture, 
    \begin{equation}
        Z \geq \frac{Z_H}{2} \exp\left(\frac{\log(q)}{q+1}e^{\log(n)+\log\left(1-\frac{2(q^2-1)}{n(q^2+1)}\right)s}\right)\,.
    \end{equation}
\end{theorem}
    Although a slightly different form than the other lower bounds, this still yields the conclusion
    \begin{equation}
        s_{AC} \geq \frac{q^2+1}{2(q^2-1)}n\log(n) - O(n)\,,
    \end{equation}
    which is tight with the upper bound up to subleading corrections. When $q=2$ (qubits), the prefactor of the $n\log(n)$ is $5/6$, settling a conjecture proposed in Ref.~\cite{HarrowMehraban2018tdesign}.

The upper and lower bounds together allow us to conclude that $s_{AC} = \Theta(n\log(n))$ for both the 1D architecture and the complete-graph architecture, and, in fact, we have matching upper and lower bounds on the constant prefactor of the $n\log(n)$. 

We note that, for $q\geq 5$, our results have the counter-intuitive implication that the 1D architecture anti-concentrates faster than the complete-graph architecture, even though it is geometrically local. We argue that this is an artifact of the definition of the models, and can be explained by the fact that the qudit pairs acted upon by the gates in the complete-graph architecture are chosen randomly, while the qudit pairs in the 1D architecture are not random; in fact, in the latter case they are optimally packed into layers of $n/2$ non-overlapping gates. As $q$ increases, anti-concentration becomes arbitrarily fast for the 1D architecture (the coefficient of the $n\log(n)$ decreases like $1/\log(q)$). Meanwhile, for the complete-graph architecture, no matter how large $q$ is, there will always be some minimum number of gates --- roughly $n \log(n)/2$ --- needed simply to guarantee that all the qudits have been involved in the circuit with high probability. We suspect that a parallelized version of the complete-graph architecture would anti-concentrate with a slightly better constant than the 1D architecture.

\section{Related work and Implications}

Here we highlight a few relevant previous works and emphasize how our results fit in.
\begin{list}{$\bullet$}{\leftmargin=1.2em \itemindent=0em} 
    \item Harrow and Mehraban \cite{HarrowMehraban2018tdesign} studied how quickly RQCs form approximate unitary $t$-designs and anti-concentrate for various architectures. For geometrically local circuits, they showed that the approximate $t$-design property is achieved after only $O(n^{1/D})$ depth in $D$ spatial dimensions, the first work to break the $O(n)$ barrier for designs. Since anti-concentration follows from the approximate 2-design property, their work implies an $O(n^{1/D})$ upper bound on the anti-concentration depth. We show that for $D=1$, the anti-concentration depth is actually $\Theta(\log(n))$ and we conjecture that this is also the case for $D \geq 2$, but we do not prove this, so the $O(n^{1/D})$ bound remains the best known for $D\geq 2$. 
    
    They also considered the question of anti-concentration in the complete-graph architecture and showed an upper bound on the anti-concentration size of $O(n\log(n)^2)$ and a lower bound of $\Omega(n\log(n))$. They used heuristic reasoning to conjecture that (for $q=2$) the anti-concentration size should be $5n\log(n)/3$, up to leading order.  We are able to show that to leading order the anti-concentration size for the complete-graph architecture is $5n\log(n)/6$. This is off by a factor of 2 from the conjecture stated in their paper, which we suspect is due to a minor error in their heuristic reasoning.
    
\item Barak, Chou, and Gao \cite{Barak2020Spoofing} developed a classical algorithm for shallow RQCs that achieves a non-negligible score on the Linear Cross-Entropy Benchmarking (XEB) metric despite not performing a full simulation of the RQCs. The Linear XEB metric was used by Google to verify its 2019 quantum computational supremacy experiment \cite{Arute2019GoogleQuantumSupremacy}. Barak, Chou, and Gao show that if a depth-$d$ RQC architecture in $D$ spatial dimensions has collision probability $Z$, their algorithm achieves a score of $\epsilon$ with high probability after a total runtime $(2^n Z) \cdot \exp(\epsilon \,15^{-d}) \cdot \text{poly}(n,2^{d^D})$ (here $q=2$). They prove that $Z = O(2^{-n})$ after $d = \Omega(\log(n))$ for 1D RQCs, which is equivalent to our \autoref{thm:1Dupperboundsummary}. This shows that their algorithm achieves a $\epsilon \geq 1/\text{poly}(n)$ score in polynomial time for logarithmic depth 1D RQCs. For 2D RQCs, they conjecture that $Z = O(2^{-n})$ after depth $d = O(\sqrt{\log(n)})$, which would imply their algorithm achieves $\epsilon \geq 1/\text{poly}(n)$ score in polynomial time at that depth. Our \autoref{thm:generallowerboundsummary} contradicts their conjecture by showing generally that $2^nZ \geq \exp(n^{1-o(1)})$ when $d$ is sublogarithmic. 

\item Our method performs expectations over individual gates in the RQC using formulas for Haar integration, a strategy that has also been used on similar problems in the past. Many works have used this strategy to form a random walk over Pauli strings with wide-ranging applications \cite{DahlstenOliveiraPlenio2007TypicalEntanglement,HarrowLow2009RQC2design,BrownFawzi2012RQCScrambling, BrownFawzi2015RQCDecoupling, BrownFawzi2013RandomCodes,onorati2017mixing, HarrowMehraban2018tdesign,Gharibyan2018RandomMatrix,nhj2018opgrowth}. Our analysis applies this strategy in a distinct way that more closely resembles a series of works that interpret the resulting expression as the partition function of classical statistical mechanical models \cite{Hayden2016HolographicDualityTN,NahumVijayHaah2018OperatorSpreading,vonKeyserlingk2018OperatorHydrodynamics,ZhouNahum2019EmergentStatMech,Hunter-Jones2019StatMechDesign,BaoChoiAltman2020TheoryPhaseTransition,JianYouVasseur2020MeasurementInduced,Napp2019SEBD,Lopez-Piqueres2020MeanField}. Here, we analyze those partition functions using a Markov chain analysis, but our Markov chain has different transition rules compared to the Pauli string Markov chain. 
\end{list}

\subsection{Connection to 2-design}

Anti-concentration for random quantum circuits (as well as some Hamiltonian models) is often established as a consequence of the convergence to approximate unitary 2-designs, where approximately reproducing the first two moments of the Haar measure allows one to bound the RQC collision probability. For both 1D and complete-graph RQCs, size $O(n^2)$ circuits (of linear depth) form approximate 2-designs and therefore anti-concentrate. 
There are a number of definitions of approximate unitary designs utilizing different norms, we briefly comment on the definitions and requirements for anti-concentration in this architecture. 

As we review in \autorefapp{app:2design}, defined in terms of the diamond norm, $\ep$-approximate 2-designs have a collision probability upper bounded by $Z_H$ up to additive error. In order to achieve anti-concentration, $\ep$ must be taken to be exponentially small (i.e.\ we require $\ep=1/q^{2n}$). Ref.~\cite{BHH2016RQCtdesign} introduced a stronger notion of approximate design in terms of the complete positivity of the difference in channels. Under this strong definition, 2-designs bound the collision probability up to relative error with respect to the Haar value and thus anti-concentrate. A much weaker definition of approximate design is the operator norm of the moment operators, often called the tensor product expander (TPE) condition. Interestingly, TPEs also bound the collision probability up to additive error, but again the error needs to be exponentially small to achieve anti-concentration. 

Random quantum circuits on the 1D architecture form $\ep$-approximate 2-designs, in both diamond norm and the stronger definition, when the circuit size is $O(n(n+\log(1/\ep)))$. 
Moreover, 1D random circuits actually form $\ep$-approximate TPEs in constant depth, when the circuit size is $O(n\log(1/\ep))$. But again, anti-concentration requires that $\ep$ be taken to be $\ep=1/q^{2n}$, thus mandating linear depth. So to establish that the collision probability is bounded up to a relative error, as in the definition of anti-concentration, using unitary 2-designs or a general bound on the moments necessitates linear depth. 
For non-local RQCs defined on a complete-graph, the best known upper bounds on the approximate 2-design depth are the same as for the 1D architecture. But it has been conjectured that this may be improved for non-local RQCs, which would close the gap between the 2-design time and the depth required for anti-concentration.

To further emphasize the distinction between anti-concentration and unitary 2-designs, we note that anti-concentration can be achieved for specific short-depth circuits without generating entanglement across the system (indeed, a circuit consisting of a single layer of single-qubit Hadamards suffices). Moreover, for an ensemble of random quantum circuits, anti-concentration can be equivalently phrased as the statement that certain matrix elements of second moment operator $\BE_U [U^{\otimes 2}\otimes U^*{}^{\otimes 2}]$ reach the Haar value of $2/q^{2n}$ after some depth. Whereas the approximate 2-design condition gives that $\BE_U[\vev{\psi|U^{\otimes 2}\otimes U^*{}^{\otimes 2}|\psi}]$ is small for all states $\ket\psi$, even those that are entangled across the tensor copies. As we show in \autorefapp{app:2design}, there are necessarily some states which require linear depth to equilibrate to the minimal Haar value, at least for RQCs on the 1D architecture.

\subsection{Implications for arguments on hardness of simulation}\label{sec:ACinhardness}

Anti-concentration is a key ingredient in hardness-of-simulation arguments \cite{Aaronson2011BosonSampling,Bremner2016AverageCaseIQP,Morimae2017DQC1HardnessTVD,Bouland2018CCC,Hangleiter2018anticoncentration,Bouland2019RCSComplexity,Dalzell2020HowManyQubits,Morimae2019FGAdditive, Movassagh2018RCSAverageCase,Movassagh2019RCSAverageCaseRobust,Bouland2021NoiseFrontier,Kondo2021ImprovedRobustness} that underlie quantum computational supremacy proposals. In this section we roughly explain its role in those arguments and the implications our results have in this context. 

The starting point for these hardness arguments is the long-known observation that the answer to a hard classical problem can be encoded into the output probability $p_U(x)$ of a quantum circuit $U$.\footnote{For example, given an $n$-bit efficiently computable Boolean function $f$ consider the following circuit $U$. First, perform a layer of Hadamard gates on every qubit, then perform the $2^n \times 2^n$ diagonal unitary operation $\sum_x(-1)^{f(x)}\ketbra{x}$, then perform another layer of Hadamard gates. It is straightforward to show that $p_U(0^n)$ is proportional to $\text{gap}(f)^2$ where $\text{gap}(f) = |\{x: f(x)=0\}| - |\{x: f(x) = 1\}|$, which is an extremely difficult quantity to compute classically; it is expected that there exist functions $f$ where the best classical algorithm is essentially a brute-force enumeration over all $2^n$ inputs $x$.} Thus, exactly computing $p_U(x)$ for arbitrary $U$ and $x$ should not be possible in classical polynomial time. This remains true even if one only needs to compute $p_U(x)$ up to some constant relative error. The ultimate goal in the context of quantum computational supremacy is to show that there is no polynomial-time classical algorithm that approximately simulates random circuits (or at least to give extremely convincing evidence in favor of this conclusion). More precisely, the approximate simulation task is to produce samples from a distribution $p'_U$ for which
\begin{equation}\label{eq:TVDerror}
    \| p_U-p'_U \|_1 := \sum_{x} | p_U(x) - p'_U(x)| = \varepsilon
\end{equation}
for some small $\varepsilon = O(1)$, and to do this for a large fraction of $U$ drawn randomly from some random ensemble. Turning the starting point into the ultimate goal requires a few steps (some of which rely on conjecture). Anti-concentration is one of these steps.

The primary role anti-concentration plays is to turn a small \textit{additive} difference $|p_U(x)-p'_U(x)|$ for most $x$ into a small \textit{relative} difference $r(x)$ for most $x$, where
\begin{equation}
    r(x) :=\frac{|p_U(x)-p'_U(x)|}{p_U(x)} \,.
\end{equation}
If Eq.~\eqref{eq:TVDerror} is obeyed then the value of $|p_U(x)-p'_U(x)|$ is on the order of $\varepsilon/q^n$ for most $x$. Meanwhile, the mean value of $p_U(x)$ for random $x$ is exactly $1/q^n$. If $p_U(x)$ is anti-concentrated, then for most $x$, $p_U(x)$ will be within a constant factor of the mean, as shown in the middle diagram of \autoref{fig:ACcartoon}, and $r(x) =O(\epsilon)$ will hold for most $x$. However, if $p_U(x)$ is not anti-concentrated, then $p_U(x)$ will be much smaller than the mean for most $x$, as depicted in the right diagram of \autoref{fig:ACcartoon}. This means that without anti-concentration, $r(x) \gg \epsilon$ for most $x$, which is problematic because the hard classical problems encoded into $p_U(x)$ are no longer hard when the relative error is extremely large, so anti-concentration appears to be necessary if there is any hope of completing the hardness argument using existing techniques.

Even if anti-concentration holds, more is needed to show hardness of approximately simulating RQCs. One must turn hardness of computing $p_U(x)$ into hardness of \textit{sampling} from $p_U$ and also turn hardness for arbitrary $U$ into hardness for a \textit{random} $U$. There are techniques that work for each of these steps individually, but currently they do not work together simultaneously, and thus an additional conjecture must be made.

Our work puts sharp bounds on the number of gates needed for anti-concentration to hold in multiple RQC architectures, which constrains when these hardness arguments have the potential to work. Our finding that the number of gates per qudit needed for anti-concentration grows only like $O(\log(n))$ in the 1D and complete-graph architectures implies that perhaps RQC-based quantum computational supremacy could be achieved at a shallower circuit depth than previously believed. For example, Google's 2019 quantum computational supremacy experiment was based on  2D RQC's of depth exceeding the $\sqrt{n}$ diameter of the qubit array \cite{Boixo2018,Arute2019GoogleQuantumSupremacy}. The fact that 1D circuits anti-concentrate in $\Theta(\log(n))$ depth is evidence that 2D circuits should have the same scaling (if anything, anti-concentration should happen faster in 2D). Thus a similar quantum computational supremacy experiment might be equally defensible at $\Theta(\log(n))$ depth instead of $\Omega(\sqrt{n})$ depth. We note, however, that there are other reasons to want to go to larger depth (e.g.~classical simulation via tensor network methods becomes harder at larger depths). 

Without anti-concentration, the hardness-of-simulation arguments appear to break down, but this does not generally imply that simulation is easy. On this topic, a subset of these authors and others described an algorithm for solving the approximate simulation problem for 2D RQCs \cite{Napp2019SEBD}. The algorithm is proved to be efficient for a certain constant-depth (and thus not anti-concentrated) 2D RQC architecture, but it is conjectured to become inefficient once the depth exceeds a larger constant threshold. Thus, the complexity of the algorithm transitions to inefficient before the circuits become anti-concentrated, suggesting that in 2D there could be a regime where the RQCs are too shallow to be anti-concentrated but classical simulation is still hard.

\section{Collision probability as a sum over bit string trajectories}\label{sec:overviewofmethod}

The main technical contribution of our work is to derive a correspondence between the collision probability and a discrete sum (which can be interpreted as the partition function of a classical statistical mechanical model or as the expectation of a Markov chain) and then to derive rigorous upper and lower bounds on the sum. Here we describe the correspondence along with a brief example for a simple random quantum circuit in \autoref{fig:trajectoryexamples}. We also explain why this correspondence leads us to expect anti-concentration to be achieved after $\Theta(n\log(n))$ gates in most architectures. In \autorefapp{sec:analysisframework}, we give a more careful derivation of this correspondence, and \hl{in} the other appendices, we use it to rigorously prove the upper and lower bounds quoted in \autoref{tab:results}.

Recall that we wish to compute the collision probability 
\begin{align}
    Z :={}& q^n\EV_{U}\left[p_U(1^n)^2\right] \\
    ={}& q^n \EV_{U}\left[\bra{1^n}^{\otimes 2} U^{\otimes 2} \ketbra{1^n}^{\otimes 2} {U^{\dagger}}^{\otimes 2} \ket{1^n}^{\otimes 2}\right]
\end{align}
where $U$ is the unitary enacted by the random quantum circuit. The Haar measure uniformly covers the unitary group so, intuitively speaking, taking the expectation over application of a Haar-random gate removes much of the bias in the quantum state; we use a technique that allows us to effectively keep track of only $n$ bits of information about the $n$-qudit state after the application of (two copies of) each Haar-random gate. Instead of 0 or 1, our bits take values $I$ or $S$, because they are associated with the identity and swap operations on two qudit copies. 

In particular, if $V$ is a $q \times q$ Haar-random matrix and $\sigma$ is an operator on two copies of a $q$-dimensional Hilbert space, then the quantity $\EV_V[V^{\otimes 2} \sigma {V^{\dagger}}^{\otimes 2}]$ is equal to a linear combination of the identity operation $I$ on the two copies of the Hilbert space, and the swap operation $S$ on the two copies of the Hilbert space. Specifically, it is given by
\begin{align}\label{eq:secondmomentformula}
   \left(\frac{\Tr(\sigma) - q^{-1}\Tr(\sigma S)}{q^2-1}\right)I + \left(\frac{\Tr(\sigma S) - q^{-1}\Tr(\sigma)}{q^2-1}\right)S\,.
\end{align}
This well-known formula is derived in \autorefapp{sec:analysisframework}.

By applying the formula to each of the two-qudit Haar-random gates sequentially, the state (which begins in $\ketbra{1^n}^{\otimes 2}$) evolves as a sum over $n$-fold tensor products of identity and swap operations. Each of these $n$-fold tensor products is labeled by an $n$-bit vector that we call a \textit{configuration} $\vec{\nu} \in \{I,S\}^n$. For a circuit with $s$ two-qudit gates, each term in the resulting sum is then associated with a length-$(s+1)$ sequence of configurations $\gamma = (\vec{\gamma}^{(0)},\ldots,\vec{\gamma}^{(s)})$, which we call a \textit{trajectory}. Each trajectory $\gamma$ has a certain non-negative coefficient in the sum, allowing us to write
\begin{equation}\label{eq:introZweight}
    Z = \frac{1}{(q+1)^n}\sum_\gamma \text{weight}(\gamma)
\end{equation}
for a fairly simple weighting function, described as follows and derived more carefully in \autorefapp{sec:analysisframework}. 

First of all, the weight for most trajectories is simply 0. In order for a trajectory to have positive weight it must obey the following rules. If the gate at time step $t$ acts on qudits $a$ and $b$, then the configuration values $\gamma^{(t)}_a,\gamma^{(t)}_b \in \{I,S\}$ at positions $a$ and $b$ must be equal, either both $I$ or both $S$.  Thus if the values disagreed at the previous time step, i.e.~$\gamma^{(t-1)}_{a} \neq \gamma^{(t-1)}_b$, one of the bits must be flipped during the transition from $\vec{\gamma}^{(t-1)}$ to $\vec{\gamma}^{(t)}$. If the values at positions $a$ and $b$ already agreed at time step $t-1$, they must remain unchanged from time step $t-1$ to time step $t$. Moreover, the bit values at the other $n-2$ positions must also remain unchanged from time step $t-1$ to time step $t$.

For trajectories that obey these rules, the weight is reduced by a factor of $2/5$ for qubits, or $q/(q^2+1)$ for general local dimensions. Thus, the most significant terms in the weighted sum are the terms with the fewest bit flips along the trajectory. 
The expression for $Z$ as a weighted sum can alternatively be interpreted as a partition function for an Ising-like classical statistical mechanical model since it is a weighted sum over ``spin'' configurations for spins with two possible values, or it can be interpreted as the expectation of a certain quantity over a simple Markov chain that generates the sequence $(\vec{\gamma}^{(0)},\ldots,\vec{\gamma}^{(s)})$. We take the latter approach in our application of the method to prove our upper and lower bounds. See \autoref{fig:trajectoryexamples} for an example of two trajectories for a simple RQC, along with a calculation of their weight.

\begin{figure*}
\centering
\includegraphics[width=0.9\linewidth]{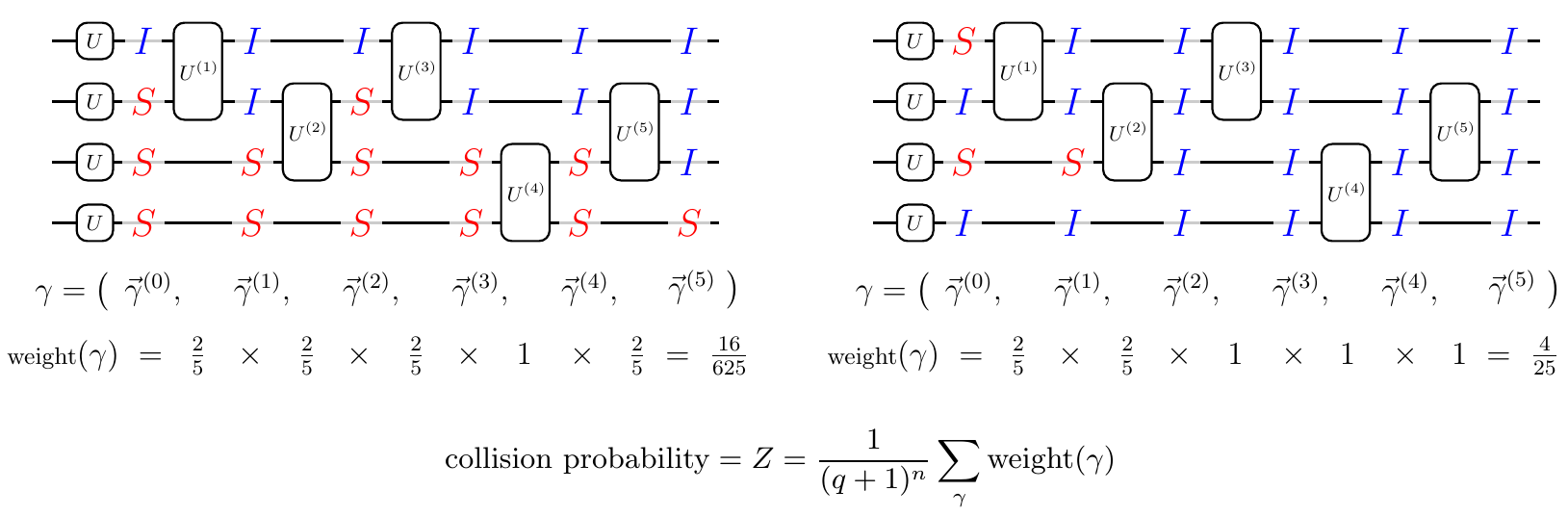}\\[-8pt]
\caption{\label{fig:trajectoryexamples} Two example trajectories for a quantum circuit diagram with $n=4$ qubits and $s=5$ gates. Each gate displayed is chosen randomly from the Haar measure over single or two-qubit unitaries. The collision probability $Z$ is expressed as a weighted sum over trajectories $\gamma = (\vec{\gamma}^{(0)},\ldots,\vec{\gamma}^{(s)})$, which are length-$(s+1)$ sequences of assignments (``configurations") of $I$ or $S$ to each of the $n$ qudits. When the input bits to a gate are assigned opposite values, one must be switched at the next configuration in the sequence. These bit flips happen at gates 1, 2, 3, and 5 in the first example, and at gates 1 and 2 in the second example. Each bit flip results in a reduction of the weight by a factor $2/5$ (when $q=2$). In the second example, the trajectory reaches one of the fixed point configurations where all $n$ values agree; this is not the case in the first example. Trajectories that quickly reach a fixed point generally have larger weights and make up most of the contribution to the collision probability. }
\end{figure*}

The correspondence given in Eq.~\eqref{eq:introZweight} is powerful because we have a good sense of what to expect from the weighted sum over trajectories, and we can draw conclusions that were not obvious from the definition of the collision probability itself. For example, we can straightforwardly analyze the infinite circuit-size limit. In this limit, each positive-weight trajectory $\gamma$ will be forced to keep flipping bits (each time a two-qudit gate acts on a disagreeing pair of bits) until it reaches a fixed point, either $I^n$ or $S^n$, in which case bits can no longer be flipped since all the bits agree.  Let $Q(x)$ be the total weight of all trajectories that begin at a configuration with $x$ $S$ assignments and $n-x$ $I$ assignments. At some point in the circuit, a disagreeing pair of bits will be acted upon by a gate, and one of the bits must flip, sending the number of $S$ assignments either to $x-1$ or $x+1$ and reducing the weight by $q/(q^2+1)$. Since there are an infinite number of gates, the following recursion relation must be obeyed
\begin{equation}
    Q(x) = \frac{q}{q^2+1}\left(Q(x-1) + Q(x+1)\right)\,,
\end{equation}
which, by imposing the boundary conditions $Q(0) = Q(n)=1$, has the unique solution
\begin{equation}\label{eq:Q(x)recursion}
    Q(x) = \frac{q^x + q^{-x}}{q^n+1} \,.
\end{equation}
Moreover, for each $x$, there are $\binom{n}{x}$ configurations each contributing weight $Q(x)$, so
\begin{equation}
    \lim_{s\rightarrow \infty} Z = \frac{\sum_{x=0}^n \binom{n}{x}(q^x + q^{-x})}{(q+1)^n(q^n+1)} = \frac{2}{q^n+1} = Z_H
\end{equation}
reproducing the value $Z_H$ that would be obtained if the random quantum circuit were one large $q^n \times q^n$ Haar-random transformation instead of a series of $q^2 \times q^2$ two-qudit gates. (The fact that a $q^n \times q^n$ Haar-random transformation yields $Z_H$ is a direct consequence of Eq.~\eqref{eq:secondmomentformula} with the substitution $q \rightarrow q^n$.) This conclusion makes sense since a random circuit with an infinite number of $2$-local Haar-random gates should enact a global Haar-random transformation. 

When the circuit size is a finite number $s$, we have $Z > Z_H$, corresponding to the fact that many trajectories have not yet reached a fixed point and are overweighted compared to their contribution to $Z_H$. As the circuit size increases, more of the trajectories get closer to the fixed point and $Z$ approaches $Z_H$. The point at which anti-concentration is achieved is intimately connected with the point at which most of the weight can be accounted for by trajectories that have reached a fixed point. A depiction of this process at $n=60$ is given in \autoref{fig:typicaltrajectories}.

Our quantitative challenge is to understand, for a certain RQC architecture, how quickly these trajectories approach the fixed points, and consequently how quickly $Z$ approaches $Z_H$, as the circuit size increases. Recall that we define the \textit{anti-concentration size} $s_{AC}$ to be the circuit size (as a function of the number of qudits $n$) needed for $Z$ to be only a constant factor larger than $Z_H$. Perhaps surprisingly, we find in multiple architectures that $s_{AC} = \Theta(n\log(n))$, corresponding to only $\Theta(\log(n))$ gates per qudit. We can explain this observation heuristically by generating trajectories $\gamma$ at random with probability proportional to $\text{weight}(\gamma)$ (in the statistical mechanical interpretation, this corresponds to drawing samples from the thermal distribution). For typical trajectories generated in this fashion, each additional layer of $\Theta(n)$ gates will cause the trajectory to move a constant fraction of the way closer to terminating at a fixed point. Since trajectories typically begin on the order of $n$ bit flips away from the fixed point (i.e.~the initial configuration typically has $\Theta(n)$ $I$ assignments and $\Theta(n)$ $S$ assignments), $\Omega(\log(n))$ layers are necessary and sufficient for typical trajectories to get within a constant distance from the fixed point. 

This heuristic statement is perhaps confirmed most clearly in the complete-graph architecture, where qudit pairs are chosen uniformly at random. Here let $x \ll n$, and suppose the current configuration at time step $t$ has value $S$ at $x$ of the $n$ positions and value $I$ at the other $n-x$ positions. If we perform gates on $n/2$ random pairs of qudits, we will expect roughly $x$ of those pairs to couple an $I$ value with an $S$ value. Each time this happens, a bit must be flipped and there is an opportunity for the trajectory to move closer to the fixed point $I^n$. Thus, we expect the number of $S$ values in the configuration at time step $t+n/2$ to have decreased by an amount proportional to $x$. After $\Theta(n\log(n))$ gates, we expect the trajectory to be at (or very close to) the fixed point $I^n$ with high probability. Fewer gates would leave most trajectories too far from the fixed point for anti-concentration to have been reached. In \autoref{fig:typicaltrajectories}, we illustrate the convergence of typical trajectories and the correspondent convergence of $Z$ for the complete-graph architecture at $n=60$.  

\begin{figure}[h]
\centering
\includegraphics[width=0.48\textwidth]{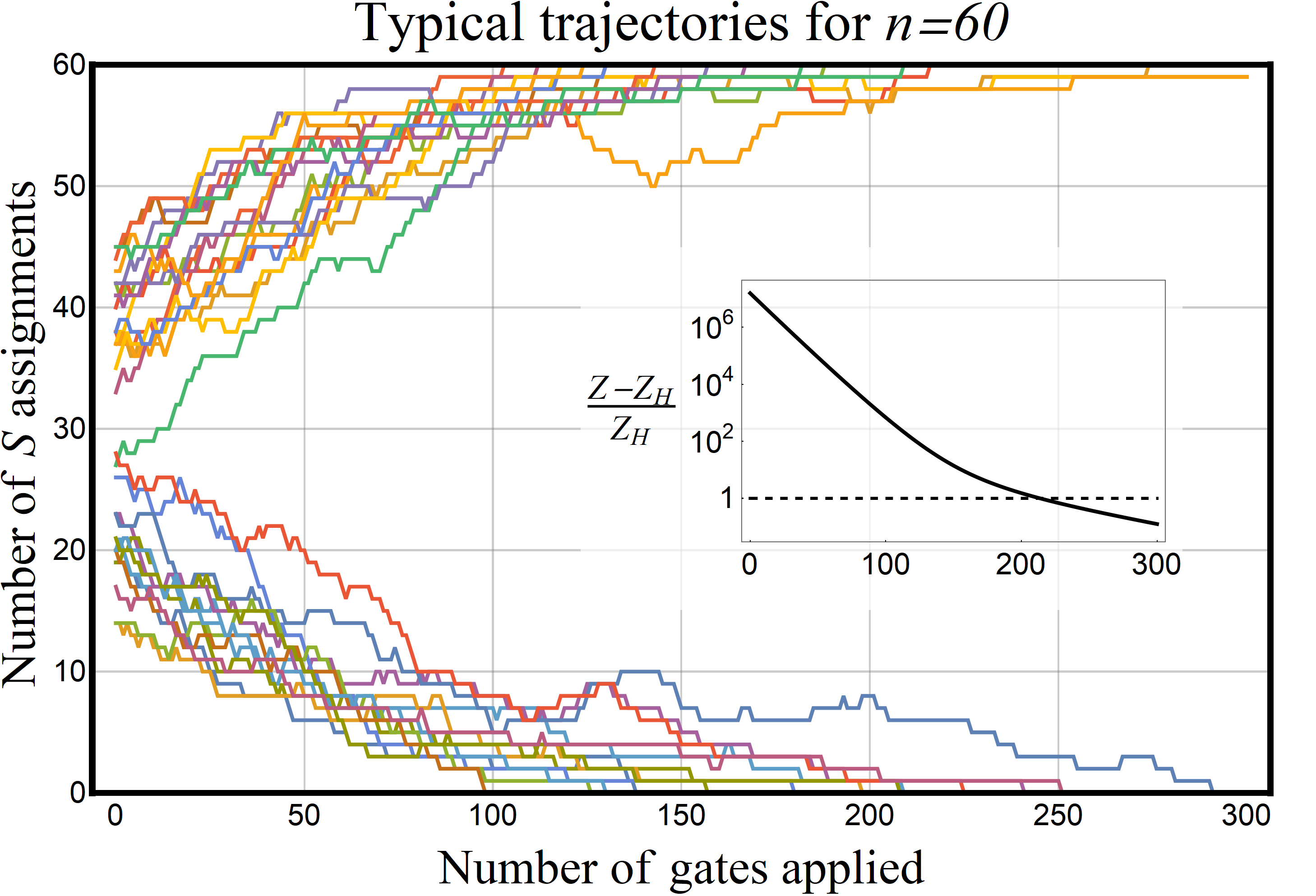}
\caption{\label{fig:typicaltrajectories} Thirty trajectories generated randomly for the complete-graph architecture at $n=60$. A trajectory $\gamma$ is chosen with probability proportional to $\text{weight}(\gamma)$ in the $s \rightarrow \infty$ limit, and then the number of $S$ assignments (out of 60) are plotted for the first 300 time steps. The trajectories rapidly approach either the fixed point $I^n$ with 0 $S$ assignments, or the fixed point $S^n$ with 60 $S$ assignments, but not all have reached the fixed point within 300 time steps. The distance of a typical trajectory from the nearest fixed point decays exponentially with time, with characteristic time scale $\Theta(n)$. Thus, it takes $\Theta(n\log(n))$ gates for most typical trajectories to have reached the fixed point. Inset: As trajectories approach the fixed points, the collision probability $Z$ (which can be efficiently numerically calculated for the complete-graph architecture) approaches $Z_H$. Anti-concentration is defined as the point where it falls beneath $2Z_H$ (dashed line), which occurs at $s=214$ for $n=60$.  }
\end{figure}

We prove that a similar situation occurs even if the gates are arranged in a 1D fashion, and we fully expect that this situation applies for nearly all natural\footnote{One can construct contrived architectures that do not quickly anti-concentrate by partitioning the qudits into many subsets and only rarely choosing a gate that couples qudits from different subsets. We define a property we call \textit{regularly connected} \autoref{def:regularlyconnected} to rule out this kind of situation. We prove that it implies anti-concentration in $O(n^2)$ gates and conjecture this can be improved to $O(n\log(n))$.} architectures, including circuits on $D$-dimensional lattices for $D>1$. We formalize this in \autoref{con:conjecture}. We believe \autoref{con:conjecture} firstly because anti-concentration should intuitively only be faster when the circuit becomes more connected, and the 1D architecture is perhaps the least connected a natural architecture can be, as it takes $\Omega(n^2)$ gates for information to travel across the diameter of the qudit array. Secondly, the above intuitive argument about the convergence of typical trajectories to a fixed point in $O(n\log(n))$ gates should apply to any natural architecture. Specifically, if you choose a configuration with $x$ $S$ assignments at random, and you apply a layer of $\Theta(n)$ two-qudit gates, with high probability you will have formed $\Theta(x)$ disagreeing pairs and moved the trajectory a constant fraction of the way to the nearest fixed point. The difficulty in proving \autoref{con:conjecture} lies in characterizing what happens in the low-probability event that this is not the case. 

Indeed, our rigorous proofs for the 1D and complete-graph architectures have to deal with the fact that it is not sufficient to examine only typical trajectory behavior. In particular, the collective contribution of trajectories at the tails of what is allowed are tricky to bound.  Nonetheless, heuristic reasoning about typical trajectory behavior ultimately gives accurate predictions about the collision probability in these cases.  

The rigorous bounds are provided in the appendices. \hl{In the 1D case, the proof associates each trajectory with a configuration of domain walls on a 2D lattice (of size $n \times (d+1)$, where $d$ is the circuit depth) and bounds their total contribution combinatorially, similar to the method employed in Refs.~\cite{NahumVijayHaah2018OperatorSpreading,vonKeyserlingk2018OperatorHydrodynamics,Hunter-Jones2019StatMechDesign,Barak2020Spoofing}. The main intuition is that trajectories that have not reached a fixed point must have domain walls that penetrate through the entire depth of the circuit and thus receive weight that decreases exponentially with the depth as $(q/(q^2+1))^d$. Accounting for the total number of possible domain walls of this type, which scales like $n2^d$, one finds that $d= O(\log(n))$ is sufficient for the overall contribution to be small. We use a similar domain wall counting method to produce a tight lower bound. 

In the complete-graph case, we present a much different and novel approach. For each trajectory $\gamma$, we define $R[\gamma]$ as the ``reduced'' trajectory that results from removing consecutive duplicates of the same configuration. Note that the weight of $\gamma$ depends only on the length of $R[\gamma]$ (the number of bit flips). Long subsequences of consecutive duplicates occur when the randomly chosen gates repeatedly act on pairs of qudits that are already assigned the same value by the configuration, an outcome that becomes more likely as the trajectory approaches a fixed point. For each $\psi$, we can condition on $R[\gamma] = \psi$, and examine the probability distribution over the length of $\gamma$ (i.e.~the number of consecutive duplicates plus the number of bit flips). We use a Chernoff bound to upper bound the probability that the length of $\gamma$ is greater than a certain quantity. We then use a generalization of the recursive method that yielded Eq.~\eqref{eq:Q(x)recursion} to perform the weighted sum over all reduced trajectories $\psi$. In the appendix, we provide a more detailed proof summary prior to the full proof.}

\section{Outlook}

In a quantum computer, quantum information is ultimately accessed by making measurements of the output state and obtaining samples from the associated output distribution over measurement outcomes. In many applications, it is desirable to choose our quantum computation completely at random, the only constraint being the arrangement of the different gates, and thus it is important to characterize the output distribution over measurement outcomes in random quantum circuits, and how it depends on the underlying circuit architecture. 

One feature of the output distribution is that, for very shallow circuits,  there are a relatively small number of very ``heavy'' measurement outcomes that are exponentially more likely than average to be obtained, a fact that inhibits the design of certain classical simulation algorithms, but also in other cases prevents potential proofs that no good simulation algorithms exist. As the circuit gets deeper, the probability mass gradually anti-concentrates and eventually becomes fairly well spread out over all possible measurement outcomes. We have developed a framework to quantitatively understand this situation; we map the anti-concentration process to the equilibration of a simple stochastic process (an alternative interpretation of the stochastic process is the partition function of a statistical mechanical model). The stochastic process allows for effective qualitative reasoning, but also produces sharp quantitative anti-concentration upper and lower bounds. 

Both sides of our bounds have meaningful and surprising takeaways. On the one hand, the fact that only $O(n\log(n))$ gates are needed to achieve anti-concentration in geometrically local and non-local architectures contradicts the intuition that anti-concentration should not occur until information has had time to spread across the entire system. In fact, up to a constant factor, the anti-concentration time does not appear to be sensitive to exact connectivity structure of the circuit. While we only rigorously consider two architectures, our work gives strong evidence that any natural architecture anti-concentrates in $O(n\log(n))$ gates (which typically corresponds to $O(\log(n))$ depth). In cases where anti-concentration is a desirable property, our work gives explicit bounds on how many gates are needed, and the fact that this number is relatively small will come as welcome news in practical situations where the gates are noisy or otherwise costly to implement. 

On the other hand, by showing that $\Omega(n\log(n))$ gates are necessary for anti-concentration (and computing the optimal constant pre-factor in our two specific scenarios), we have cleared up some confusion about very shallow circuits. Increasing the depth causes the anti-concentration process to begin, but our lower bound implies that the phenomenon of very heavy measurement outcomes will remain for any architecture of constant depth. Even the 2D circuits of depth $O(\sqrt{\log(n)})$ (for which the lightcone volume is $O(\log(n))$) considered in Ref.~\cite{Barak2020Spoofing} cannot be anti-concentrated, as had been speculated in that work. 

We conclude with some other specific open problems inspired by our work. 
\begin{list}{$\bullet$}{\leftmargin=1.2em \itemindent=0em} 
    \item We have proved that the anti-concentration size is $\Theta(n\log(n))$ for the 1D and complete-graph architectures. We believe this is true for most other natural architectures and formally conjecture in \autoref{con:conjecture} that this follows from our ``regularly connected'' definition.  
    
    \item A sharp anti-concentration analysis for 2D and higher dimensional geometrically local architectures would be particularly valuable since, unlike in 1D, $\Theta(\log(n))$-depth 2D circuits can perform universal quantum computation (indeed, $\Omega(1)$-depth is sufficient \cite{TerhalDiVincenzo2004ConstantDepthHardness}), and 2D circuits form the basis for Google's 2019 quantum computational supremacy experiment \cite{Arute2019GoogleQuantumSupremacy}. 
    
    \item We suspect the constant prefactor of $(2\log(q^2+1))^{-1}$ in the general lower bound in \autoref{thm:generallowerboundsummary} could be improved. What is its optimal value? That is, can we show an improved general lower bound and then find an RQC architecture $\mathcal{A}_{\text{fast}}$ that has a matching upper bound. This would show that $\mathcal{A}_{\text{fast}}$ is the fastest anti-concentrator. A candidate for $\mathcal{A}_{\text{fast}}$ is the architecture where each layer of $n/2$ gates is formed by choosing a random partition of the $n$ qudits into $n/2$ pairs. 
    
    \item Are there other problems involving second moment calculations over RQCs where our techniques would produce sharp upper and lower bounds? One such problem could be the $2$-design time for RQCs in various architectures.
\end{list}

\subsection*{Acknowledgments}
We thank Hrant Gharibyan, Jonas Haferkamp, Aram Harrow, Richard Kueng, Saeed Mehraban, John Napp, Sepehr Nezami, and John Preskill for discussions and helpful comments on the draft.
AD and FB acknowledge funding provided by the Institute for Quantum Information and Matter, an NSF Physics Frontiers Center (NSF Grant {PHY}-{1733907}). This material is also based upon work supported by the NSF Graduate Research Fellowship under Grant No.~DGE‐1745301. 
NHJ would like to thank the IQIM at Caltech for its hospitality during the completion of part of this work. Research at Perimeter Institute is supported by the Government of Canada through the Department of Innovation, Science and Industry Canada and by the Province of Ontario through the Ministry of Colleges and Universities.

\pagebreak

\onecolumngrid\par
\appendix

\section{Definitions: Random quantum circuit architectures, anti-concentration, and the collision probability}\label{sec:preliminaries}

\subsection{Random quantum circuits (RQCs)}

Here we establish some precise definitions for the terms in this paper. Throughout, we consider systems of $n$ qudits of local Hilbert space dimension $q$, with basis states $\{\ket{1},\ket{2},\ldots,\ket{q}\}$. Loosely speaking, a quantum circuit is a sequence of unitary transformations called gates, which each typically involve only a few of the $n$ qudits, acting on the initial state $\ket{1}^{\otimes n} \equiv \ket{1^n}$. Formally, we let a \textit{quantum circuit diagram} of circuit size $s$ be specified by a length-$s$ sequence $(A^{(1)},\ldots, A^{(s)})$ of non-empty subsets of $[n]:= \{1,2,\ldots,n\}$, indicating for each gate which qudits participate in that gate.  Since we consider circuits consisting only of two-qudit gates, we require $\lvert A^{(t)} \rvert = 2$ for all $t$. We also make the assumption that the circuit begins with a single-qudit gate on each of the $n$ qudits at the beginning of the circuit, without counting these $n$ gates toward the circuit size. This sequence can be turned into a diagram as in \autoref{fig:trajectoryexamples} (ignoring the overlaid $I$ and $S$), where the gate sequence is $(\{1,2\},\{2,3\},\{1,2\},\{3,4\},\{2,3\})$. Note that the single-qudit unitaries are each displayed with the symbol $U$ but will not necessarily be the same unitary. The circuit \textit{depth} $d$ of a circuit diagram is the minimum number of layers of non-overlapping gates needed to implement all $s$ gates in the circuit, or formally, the smallest integer such that there exists a sequence $0=s_0 < s_1 < s_2 <\ldots < s_d=s$ where $A^{(t)}\cap A^{(t')} = \emptyset$ whenever $s_j < t< t' \leq s_{j+1}$. 

Once a circuit diagram has been chosen, a \textit{quantum circuit instance} is generated by additionally specifying a length-$s+n$ sequence of unitary matrices $(U^{(-n)},\ldots,U^{(-1)}, U^{(1)},\ldots, U^{(s)})$ where $U^{(-j)}$ is a $q \times q$ (single-qudit) matrix for each $j=1,\ldots, n$ and $U^{(t)}$ is a $q^2\times q^2$ (two-qudit) matrix for each $t=1,\ldots, s$. We denote the global $q^n \times q^n$ unitary operator implemented by the circuit by $U$, where
\begin{equation}\label{eq:U}
    U := U_{A^{(s)}}^{(s)} U_{A^{(s-1)}}^{(s-1)} \ldots U_{A^{(2)}}^{(2)} U_{A^{(1)}}^{(1)}U_{\{1\}}^{(-1)}\ldots U_{\{n\}}^{(-n)}\,,
    \end{equation}
with $V_X$ indicating the action of the $q^{\lvert X \rvert} \times q^{\lvert X \rvert}$ unitary $V$ on the qudits in subregion $X \subset [n]$ tensored with the identity operation on the qudits in the complement of $X$. 

In this work, we will always assume that projective computational basis measurements are performed on all $n$ qudits at the end of the circuit. Thus, a quantum circuit instance $U$ has a corresponding classical probability distribution $p_U$ over possible measurement outcomes $x \in [q]^n$, as follows:
\begin{equation}
    p_U(x) = \lvert \bra{x} U \ket{1^n}\rvert^2\,.
\end{equation}

Random quantum circuits will refer to situations when, once a circuit diagram has been fixed, the actual unitary gates $U^{(t)}$ that determine the circuit instance are each randomly chosen independently from some distribution over the unitary group. In this paper, we always take this distribution to be the Haar measure, but since our techniques rely on calculating expectations over quantities with only two copies of each $U^{(t)}$, our results also apply when the gates are drawn from any 2-design, such as the Clifford group. Note that Google's quantum computational supremacy experiment \cite{Arute2019GoogleQuantumSupremacy} drew gates from another distribution that is not a 2-design. Heuristically speaking, as long as the distribution lacks any bias or symmetries, we expect properties like anti-concentration to be the same as in the Haar-random case.

\subsection{Random quantum circuit architectures}

An \textit{architecture} for random quantum circuits is simply a procedure for choosing a circuit diagram. Formally, we define it to be a (possibly randomized) classical algorithm that, given parameters $n$ and $s$, computes a circuit diagram of size $s$ on $n$ qudits. Given an architecture and parameters $n$ and $s$, we let the expectation of some quantity $Q$, denoted $\EV_{U}[Q]$, refer to the expectation over the process of first choosing a circuit diagram according to the architecture, and then choosing a circuit instance by randomly generating each gate in the circuit diagram independently from the Haar measure. Next, we define the two architectures that we consider.

\begin{definition}[Complete-graph architecture]\label{def:CGarchitecture}
    Circuit diagrams of size $s$ on $n$ qudits are generated by choosing $s$ gates each uniformly at random from the set of all two-qudit gates, i.e.~$A^{(t)}$ is chosen uniformly from $\{\{a,b\}: a,b\in [n], a \neq b\}$.
\end{definition}

Note that if it could be guaranteed that every qudit would eventually participate in at least one gate, the distribution over circuit instances would be equivalent if we omitted the first layer of $n$ single-qudit gates (defined to be part of every architecture), a fact that follows from the invariance of the Haar measure; the single-qudit gates could be absorbed into the two-qudit Haar-random gates that act directly before or after without changing the distribution over the two-qudit gates. However, in the complete-graph architecture there is a chance that a qudit does not participate in any two-qudit gates, although for sufficiently large circuit size the probability of this vanishes.  

\begin{definition}[1D architecture]\label{def:1Darchitecture}
    Assume $n$ is even and $d:=2s/n$ is an integer. The circuit diagram of size $s$ on $n$ qudits is generated by alternating between the two types of layers of $n/2$ non-overlapping nearest-neighbor two-qudit gates on a ring. That is, for each $t = 1,\ldots, n/2$, if $j$ is even, then $A^{(t+jn/2)} = \{2t-1,2t\}$, and if $j$ is odd, then $A^{(t+jn/2)} = \{2t,2t+1\}$, where index $n$ is identified with index $0$ to enforce periodic boundary conditions. 
\end{definition}

\subsection{Collision probability and anti-concentration}

Anti-concentration is a concept that describes a classical probability distribution for which the probability mass is not too concentrated onto a small number of outcomes of the random variable. The uniform distribution is the ultimate anti-concentrated distribution, as the probability mass is allocated evenly over every possible outcome, but we would still like the term anti-concentrated to apply to some non-uniform distributions if the probability mass is fairly well spread over many of the outcomes. There are multiple ways to make this quantitative. For the purposes of this paper, we choose one way --- the collision probability --- that mirrors previous work on anti-concentration of quantum circuit outputs and suffices for the applications we discuss in the introduction.

Let $X$ be a discrete random variable and let $M$ be the set of outcomes of $X$. We can form another random variable $p$, where $p$ is equal to $\Pr[X=x]$ for an $x$ chosen uniformly at random from $M$. Since $\sum_{x}\Pr[X=x] = 1$, we have $\EV[p] = 1/\lvert M \rvert$ no matter how $X$ is distributed. We define the collision probability for $X$ to be
\begin{align}
    Z :={}& \sum_{x \in M}\Pr[X=x]^2 \\
    ={}&  \EV[p^2] \cdot \lvert M \rvert \\
    ={}& \var(p) \lvert M \rvert+ \lvert M \rvert^{-1}\,,
\end{align}
which is the probability that two identical independent copies of $X$ will be equal to each other --- hence \textit{collision} probability. If the distribution over $X$ is the uniform distribution, then the distribution over $p$ is the point distribution on the value $\lvert M \rvert^{-1}$, the collision probability takes its minimal value $Z = \lvert M \rvert^{-1}$, and $\var(p) = 0$. If $X$ is non-uniform but still somewhat anti-concentrated, then $p$ won't always be $\lvert M \rvert^{-1}$ but it will usually be close, and this will be reflected by a collision probability that is greater, but not too much greater than $\lvert M \rvert^{-1}$. Formally, we make the following definition.

\begin{definition}[Anti-concentrated]
We say that a random variable $X$ over a set $M$ of outcomes is \textit{$\alpha$-anti-concentrated} for $0 < \alpha \leq 1$ if
\begin{equation}
Z:= \sum_{x \in M}\Pr[X=x]^2 \leq \frac{1}{\lvert M \rvert \alpha}\,.
\end{equation}
\end{definition}
Thus a distribution is $1$-anti-concentrated if and only if it is the uniform distribution. 

In our setting, the random variable $X$ is the measurement outcome of a random quantum circuit instance, which is distributed according to the distribution $p_U$ over the outcome set $[q]^n$. Example distributions of $p_U$ for RQC outputs in the uniform, the non-uniform but still anti-concentrated, and the not anti-concentrated case are shown in the caricature in \autoref{fig:ACcartoon}.  A random quantum circuit architecture for specified $n$ and $s$ is understood as an ensemble over many different $U$, only some of which will have output distributions $p_U$ that are $\alpha$-anti-concentrated for a certain choice of $\alpha$. We would like to say that the architecture as a whole is anti-concentrated if typical circuit instances drawn from the architecture are anti-concentrated, acknowledging that not every instance will be. We also require this to hold for the same constant $\alpha$ as $n$ increases, with $s$ increasing like some function $s(n)$. Formally, we accomplish this by averaging the collision probability over the random circuit instance, as follows. 

\begin{definition}[Anti-concentrated RQC architecture]\label{def:anticoncentratedArch}
    We say that a random quantum circuit architecture is \textit{$\alpha$-anti-concentrated} for $0 < \alpha\leq 1$ at circuit size $s(n)$ if there exists $n_0$ such that whenever $n\geq n_0$
    \begin{equation}\label{eq:anticoncARCH}
        Z:=\EV_U\left[\sum_{x\in[q]^n} p_U(x)^2 \right] \leq (\alpha q^n)^{-1}\,,
    \end{equation}
    where $\EV_U$ denotes drawing circuit instances according to the architecture over $n$ qudits with circuit size $s(n)$. Generally, we say that the architecture is \textit{anti-concentrated} at size $s(n)$ if there exists a constant $\alpha > 0$ independent of $n$ for which it is $\alpha$-anti-concentrated at that size.
\end{definition}

RQC architectures for which every qudit experiences at least one gate, which includes all the architectures introduced above, will have a symmetry over the $q^n$ measurement outcomes in the sense that the quantity $p_U(x)$ is distributed identically (over circuit instances) for every $x$. In this case each term in the sum in Eq.~\eqref{eq:anticoncARCH} will have the same contribution and we can write simply 
\begin{equation}
        Z = q^n\EV_U\left[ p_U(1^n)^2 \right]\,.
\end{equation}

The anti-concentration of an architecture implies that most of the instances drawn from that architecture have good anti-concentration properties: Given an architecture at a certain size and a bound on its collision probability $Z \leq \alpha^{-1} q^{-n}$, we can use Markov's inequality to assert that at least a $1-\beta$ fraction of instances have collision probability at most $q^{-n}\left(1+(\alpha^{-1}-1)\beta^{-1}\right)$. In practice, we expect the collision probability of individual instances to be even more clustered near the mean collision probability than this analysis indicates, but proving that this is the case would seem to require computing higher moments like $\EV_U[p_U(1^n)^{k}]$ for $k > 2$. 

As discussed in the main text, an important implication of an $\alpha$-anti-concentrated architecture is that for any $\beta$ with $0 \leq \beta \leq 1$ and sufficiently large $n$ 
\begin{equation}\label{eq:prnotclosetozero}
    \Pr_{U}[p_U(x) \geq \beta q^{-n}] \geq (1-\beta)^2\hl{\alpha}\,,
\end{equation}
which follows directly from the Paley-Zygmund inequality. This inequality indicates that whenever an architecture is anti-concentrated, at least a constant fraction of the outcomes will be allocated an amount of mass that is within a constant factor $\beta$ of the mean mass; it cannot be the case that all but a vanishing fraction of the outcomes are allocated a vanishing fraction of the mean mass. 

\section{Framework for analysis: Random quantum circuits as a stochastic process}\label{sec:analysisframework}

This appendix gives more details on the correspondence discussed in \autoref{sec:overviewofmethod} from the main text. The key idea in our analysis of the collision probability of RQCs is to perform the Haar expectation over each local unitary individually. This is possible due to explicit formulas for expectations under action by a Haar-random unitary. We use these formulas to re-express the collision probability, originally an integral over many continuously varying unitary matrices drawn from the Haar measure, as a weighted discrete sum, which is then analyzed using combinatorial and stochastic methods. This weighted sum can also be interpreted as the partition function of a classical statistical mechanical Ising-like model or as the expectation value of a simple stochastic process. \autoref{fig:statmech} depicts these equivalent representations of the problem. In this appendix, we explain this method and derive the important formulas that will apply generally for any RQC architecture, which are then used in later sections to prove our main results.

\begin{figure}
\centering
\includegraphics[width=0.98\linewidth]{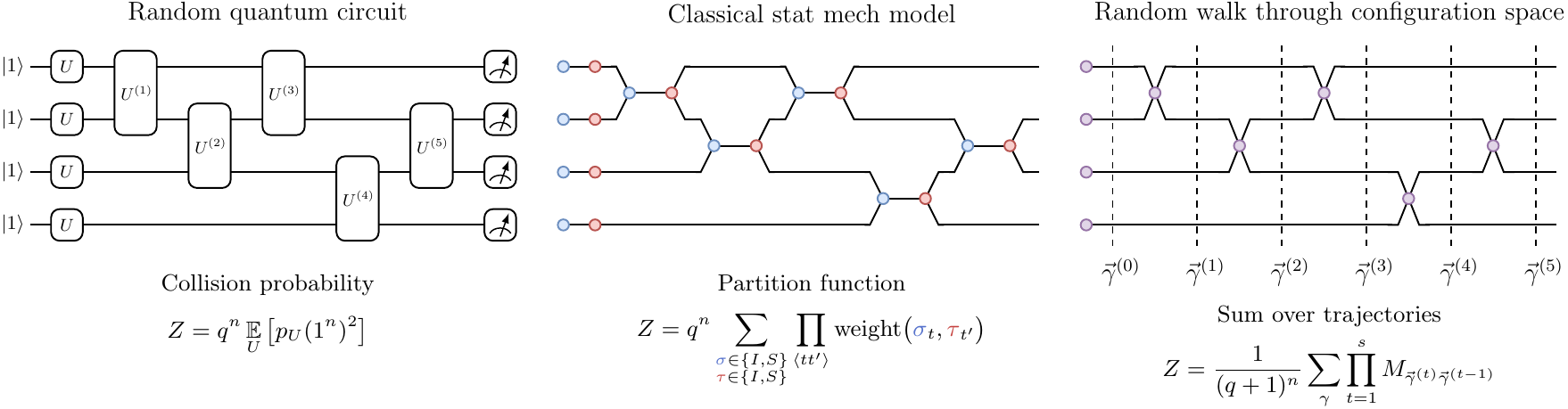}
\caption{A diagram depicting the equivalent ways to interpret the expected value of the collision probability for random quantum circuits. Left: a random quantum circuit of size five. Middle: the reinterpretation as the partition function of a classical statistical mechanics model with local Ising-like spins. Right: another interpretation as a stochastic process of evolving configurations.}
\label{fig:statmech}
\end{figure}

\subsection{Averaging individual unitaries over the Haar measure}

The quantity of interest for anti-concentration is the expected collision probability, which is proportional to a second moment over choice of unitary operator $U$, as illustrated in the following equation, where we recall that $\ketbra{1^n}^{\otimes 2}$ is two copies of the circuit input state
\begin{align}
    Z:={}&q^n\EV_U\left[\left(\bra{1^n} U \ket{1^n} \bra{1^n} U^\dagger \ket{1^n}\right)^2\right] \\
    ={}& q^n \Tr\left[\EV_U \left[ U^{\otimes 2} \ketbra{1^n}^{\otimes 2} U^{\dagger}{}^{\otimes 2} \right]\ketbra{1^n}^{\otimes 2}  \right]\label{eq:collisionprobcopies}\,.
\end{align}
Moreover, for a fixed quantum circuit diagram, the unitary $U$ is given by Eq.~\eqref{eq:U} as a product of single-qudit unitaries $U^{(-j)}$ acting on qudit $j$ for $j=1,\ldots,n$ and two-qudit unitaries $U^{(t)}$ acting on some pair of qudits $A^{(t)}\subset [n]$ for $t=1,\ldots, s$. Each unitary is independently chosen according to the Haar measure, and its expectation can be evaluated separately. Let
\begin{equation}
M^{(t)}[\rho] := \EV_{U^{(t)}}\left[{U^{(t)}_{A^{(t)}}}^{\otimes 2} \rho\, {{U^{(t)}_{A^{(t)}}}^\dagger}^{\otimes 2}\right]\,.
\end{equation}
Then we can write
\begin{align}\label{eq:explocalunitaries}
    \EV_{U}\left[U^{\otimes 2} \ketbra{1^n}^{\otimes 2} {U^\dagger}^{\otimes 2} \right] &= M^{(s)}\circ M^{(s-1)} \circ \cdots \circ M^{(1)} \circ M^{(-1)} \circ \ldots \circ M^{(-n)}\left[\ketbra{1^n}^{\otimes 2}\right]\,.
\end{align}
When an architecture is itself a mixture over randomly chosen circuit diagrams, such as the complete-graph architecture, the overall quantity $\EV_U[U^{\otimes 2} \ketbra{1^n}^{\otimes 2} {U^\dagger}^{\otimes 2}]$ is a mixture over terms of the above form. 

The remainder of this subsection illustrates how the action of $M^{(t)}$ can be evaluated, ultimately allowing us to arrive at the expression for $Z$ given in Eq.~\eqref{eq:partitionfunctiontrajectories}. In the other subsections of this section, we explain how that equation can be interpreted as a partition function of a classical statistical mechanical model or as the expectation over simple stochastic process. 

When the local unitaries are drawn from the Haar measure (or any exact 2-design), the expression $M^{(t)}[\rho]$ can be evaluated in a simple way. Generally, for $\sigma$ a $q^2 \times q^2$ Hermitian operator, and with $V$ chosen from the Haar measure over the set of $q \times q$ unitaries, we define
\begin{equation}
    M[\sigma] := \EV_{V}\left[V^{\otimes 2} \sigma {V^{\dagger}}^{\otimes 2}\right] \\
\end{equation}
and observe that, for any unitary $W$ and any $\sigma$,
\begin{equation}
    M[\sigma]W^{\otimes 2} = \EV_{V}\left[V^{\otimes 2} \sigma {{(W^{\dagger}V)}^{\dagger}}^{\otimes 2}\right] = \EV_{V}\left[{(WV)}^{\otimes 2} \sigma {{V}^{\dagger}}^{\otimes 2}\right] = W^{\otimes 2} M[\sigma]\,,
\end{equation}
where the second equality follows from the invariance of the Haar measure under the substitution $V \rightarrow 
WV$.
A mathematical fact from Schur-Weyl duality (see Ref.~\cite{GoodmanWallach2000Representations}) is that any operator on $k$ copies of a system that commutes with $W^{\otimes k}$ for any unitary $W$ must be a linear combination of permutation operators over the $k$ systems. Here, we have $k=2$ and thus the only permutation operators are the identity operation $I$ and the swap operation $S$, which can be defined as the operator satisfying $S\ket{\psi} \otimes \ket{\phi} = \ket{\phi} \otimes \ket{\psi}$ for any $\ket{\psi}, \ket{\phi}$. Letting $M[\sigma] = \alpha I + \beta S$, we make the following calculations:
\begin{align}
    \Tr[M[\sigma]] &= \Tr[\sigma] = \alpha q^2 + \beta q \\
    \Tr[M[\sigma]S] &= \Tr[\sigma S] = \alpha q + \beta q^2\,,
\end{align}
which determine $\alpha$ and $\beta$ and allow us to write
\begin{equation}\label{eq:Midentityswap}
    M[\sigma] = \frac{\Tr(\sigma) - q^{-1}\Tr(\sigma S)}{q^2-1}I + \frac{\Tr(\sigma S) - q^{-1}\Tr(\sigma)}{q^2-1}S\,.
\end{equation}

The unitaries $U^{(-j)}$ are $q \times q$ (single-qudit) that act on qudit $j$. Two copies of the input state on qudit $j$ is $\ketbra{1}_{\{j\}}^{\otimes 2}$. Denote two copies of the input state on the other $n-1$ qudits by $\rho_{[n]\setminus\{j\}}$. Using Eq.~\eqref{eq:Midentityswap}, we then find that
\begin{equation}
    M^{(-j)}\left[\rho_{[n]\setminus{j}} \otimes \ketbra{1}_{\{j\}}^{\otimes 2}\right] = \rho_{[n]\setminus\{j\}} \otimes \frac{1}{q(q+1)}\left(I + S\right)_{\{j\}}\,,
\end{equation}
meaning that $M^{(-j)}$ simply replaces the state on qudit $j$ as a uniform sum over operators $I$ and $S$. Hence
\begin{equation}
    M^{(-1)}\circ \cdots \circ M^{(-n)}\left[\ketbra{1^n}^{\otimes 2}\right] = \bigotimes_{j=1}^n \left(\frac{1}{q(q+1)}\left(I + S\right)_{\{j\}}\right) =  \frac{1}{q^n(q+1)^n}\sum_{\vec{\gamma}\in \{I,S\}^{n}} \bigotimes_{j=1}^n \gamma_j\,.
\end{equation}
We call each $\vec{\gamma} \in \{I,S\}^n$ a \textit{configuration}. The above equation states that the expected value of two copies of the state after application of all the single-qudit unitaries is precisely a uniform sum over all identity/swap configurations of the $n$ sites.

Now, we need to examine the action of $M^{(t)}$ for $t>0$. In this case, the unitaries are $q^2 \times q^2$ and act on the qudit pair $A^{(t)}$. We can use Eq.~\eqref{eq:Midentityswap} by replacing $q \rightarrow q^2$ and sending $I\rightarrow I\otimes I$, the identity operation on two copies of two qudits, and $S \rightarrow S\otimes S$, the swap operation on two copies of two qudits. We assume that the input state is a product state $\rho_{[n]\setminus A^{(t)}} \otimes \rho_{A^{(t)}}$ and see that
\begin{align}
    M^{(t)}[&\rho_{[n]\setminus A^{(t)}} \otimes \rho_{A^{(t)}}] \\
    ={}& \rho_{[n]\setminus A^{(t)}} \otimes \left(\frac{\Tr(\rho_{A^{(t)}})-q^{-2}\Tr\left(\rho_{A^{(t)}} (S \otimes S)\right)}{q^{4}-1}(I \otimes I)_{A^{(t)}} + \frac{\Tr\left(\rho_{A^{(t)}} (S \otimes S)\right)-q^{-2}\Tr(\rho_{A^{(t)}})}{q^{4}-1}(S\otimes S)_{A^{(t)}}\right)\,.
\end{align}

Since the two qudit gates act after the single-qudit gates, the input state to $M^{(t)}$ will always be a sum of tensor products of $I$ and $S$, so we only need to evaluate the above expression when $\rho_{A^{(t)}}$ is either $I\otimes I$, $I \otimes S$, $S\otimes I$, or $S\otimes S$. Doing so, we arrive at
\begin{align}\label{eq:Mtevaluated1}
    M^{(t)}\left[\rho_{[n]
\setminus A^{(t)}} \otimes (I\otimes I)_{A^{(t)}}\right] &= \rho_{[n]
\setminus A^{(t)}} \otimes (I\otimes I)_{A^{(t)}} \\
    M^{(t)}\left[\rho_{[n]
\setminus A^{(t)}} \otimes (S\otimes S)_{A^{(t)}}\right] &=\rho_{[n]
\setminus A^{(t)}} \otimes (S\otimes S)_{A^{(t)}}\label{eq:Mtevaluated2} \\
    M^{(t)}\left[\rho_{[n]
\setminus A^{(t)}}\otimes (I\otimes S)_{A^{(t)}}\right]= M^{(t)}\left[\rho_{[n]
\setminus A^{(t)}} \otimes (S\otimes I)_{A^{(t)}}\right]&= \rho_{[n]
\setminus A^{(t)}} \otimes  \left(\frac{q}{q^2+1}(I \otimes I)_{A^{(t)}} +\frac{q}{q^2+1}(S \otimes S)_{A^{(t)}}\right) \label{eq:Mtevaluated3}\,.
\end{align}

Thus, if $\rho$ is a linear combination of configurations in $\{I,S\}^n$, $M^{(t)}[\rho]$ will also be a linear combination of configurations, with coefficients that transform linearly under application of $M^{(t)}$.  For configurations $\vec{\gamma}, \vec{\nu} \in \{I,S\}^n$, we let $M^{(t)}_{\vec{\nu}\vec{\gamma}}$ be the matrix element of this linear transformation defined such that
\begin{equation}
    M^{(t)}\left[\bigotimes_{j=1}^n \gamma_j\right] = \sum_{\vec{\nu} \in \{I,S\}^n} M^{(t)}_{\vec{\nu}\vec{\gamma}} \bigotimes_{j=1}^n \nu_j\,.
\end{equation}

Suppose that $U^{(t)}$ acts on qudits $A^{(t)} =\{a,b\} \subset [n]$. Then from Eqs.~\eqref{eq:Mtevaluated1}, \eqref{eq:Mtevaluated2}, \eqref{eq:Mtevaluated3}, we have
\begin{equation}
    M^{(t)}_{\vec{\nu}\vec{\gamma}}=\begin{cases}
    1 & \text{if } \gamma_a = \gamma_b \text{ and } \vec{\gamma} = \vec{\nu} \\
    \frac{q}{q^2+1} & \text{if } \gamma_a \neq \gamma_b \text{ and } \nu_a = \nu_b \text{ and } \gamma_c = \nu_c \;\forall c \in [n]\setminus\{a,b\} \\
    0 & \text{otherwise}
    \end{cases}
\end{equation}
Importantly, $M^{(t)}_{\vec{\nu}\vec{\gamma}}$ is always non-negative. The way to think about the above equation is to notice three things. First, the input configuration $\vec{\gamma}$ and the output configuration $\vec{\nu}$ must agree on all indices that are not involved in the gate, i.e.,~for all indices $c \not\in \{a,b\}$; otherwise the matrix element is 0. Second, if the two input values involved in the gate agree, i.e.,~if $\gamma_a = \gamma_b$ then $\nu_a = \nu_b = \gamma_a = \gamma_b$ must hold (in which case the matrix element is 1); otherwise, it is 0. Third, if the two input values disagree, then one of them must be flipped so that the two output values agree (in which case the matrix element is reduced to $q/(q^2+1)$); otherwise, it is 0.

Note also that 
\begin{equation}
    \Tr\left[\left(\bigotimes_{j=1}^n \gamma_j\right) \ketbra{1^n}^{\otimes 2}\right] = 1
\end{equation}
for all $\vec{\gamma} \in \{I,S\}^n$. Thus, from Eq.~\eqref{eq:collisionprobcopies}, we find
\begin{align}\label{eq:partitionfunctiontrajectories}
    Z &= \frac{1}{(q+1)^n}\sum_{\gamma \in \{I,S\}^{n \times (s+1)}} \;\;\prod_{t=1}^{s}M^{(t)}_{\vec{\gamma}^{(t)}\vec{\gamma}^{(t-1)}} \\
    &=: \frac{1}{(q+1)^n}\sum_{\gamma}\;\;\text{weight}(\gamma) \,,
\end{align}
which is the expression quoted in Eq.~\eqref{eq:introZweight} from the main text. In the above equation, the sum is over length-$(s+1)$ sequences of configurations, which we call a \textit{trajectory} $\gamma = (\vec{\gamma}^{(0)},\ldots,\vec{\gamma}^{(s)})$ and the weight of each term is given by the product of the matrix elements for each step in the trajectory. This final equation is depicted graphically in the right-hand part of \autoref{fig:statmech}.

\subsection{Collision probability as statistical mechanical partition function}

The expression for the collision probability in Eq.~\eqref{eq:partitionfunctiontrajectories} can be interpreted as a partition function for a classical statistical mechanical model by thinking of each $\gamma^{(t)}_j$ as an Ising spin variable with the association $\{I,S\} \leftrightarrow \{+1,-1\}$.
A trajectory $\gamma$ is then a configuration of the Ising spins, and $Z$ is a weighted sum over all the spin configurations. Moreover, the weight is always non-negative and is given by a product of factors $M^{(t)}_{\vec{\gamma}^{(t)}\vec{\gamma}^{(t-1)}}$ that can be determined by examining a small number of the spin values. This means that the energy functional over spin configurations of the classical Ising-like model is always real and can be broken up into local terms that depend on the local dimension $q$ and which qudits are acted upon at each step in the circuit.

The statistical mechanics interpretation has been a useful one for similar problems in the past, where certain RQC moment quantities can be exactly rewritten as the partition sum over spin configurations of a lattice model, as depicted in the central diagram in \autoref{fig:statmech}. We can arrive at the formulation as in Eq.~\eqref{eq:partitionfunctiontrajectories} from the lattice model by summing over a subset of the spins and reinterpreting the resulting nodes as 4-body interaction vertices.

This exact rewriting of RQC moment quantities has been used to compute, for instance, correlation functions \cite{NahumVijayHaah2018OperatorSpreading}, R{\'e}nyi entropies \cite{ZhouNahum2019EmergentStatMech}, and the distance to forming an approximate design \cite{Hunter-Jones2019StatMechDesign}.
Moreover, thermal phase transitions in the classical model can be related to phase transitions of entanglement-entropy-like quantities for the output state of the RQC \cite{JianYouVasseur2020MeasurementInduced,BaoChoiAltman2020TheoryPhaseTransition,Napp2019SEBD}. The interpretation is particularly intriguing when considering analogous quantities to $Z$ for higher moments. The collision probability is a second moment quantity, and the resulting stat mech model has Ising-like variables with two possible values. Quantities related to the $k$th moment will map to classical stat mech models that have $k!$ possible values, one for each element of the symmetric group $\mathcal{S}_k$. However, one challenge of computing higher-moment quantities is that the weights in the partition function can be negative (corresponding to non-real values of the energy for certain spin configurations), complicating many strategies for bounding its behavior, including the strategies employed in the rest of this paper. 

\subsection{Unbiased random walk}

We can build from the formula for $Z$ in Eq.~\eqref{eq:partitionfunctiontrajectories} and re-express it in terms of a length-$s$ unbiased random walk through configuration space $\{I,S\}^n$, which we denote $P_u$. At time step 0, a configuration $\vec{\gamma}^{(0)}$ is chosen uniformly at random, i.e.,~the initial distribution is the uniform distribution in configuration space, denoted $\Lambda_u$.
Then configuration $\vec{\gamma}^{(t+1)}$ at time step $t+1$ is generated from the configuration $\vec{\gamma}^{(t)}$ at time step $t$ as follows: letting $A^{(t)} = \{a,b\}$, if the $a$th and $b$th bits of $\vec{\gamma}^{(t)}$ agree, then the configuration is left unchanged at time step $t+1$; if they disagree, either the value at $a$ or the value at $b$ is flipped each with probability $1/2$ to form $\vec{\gamma}^{(t+1)}$. The weight is reduced each time a bit is flipped. Thus we can write
\begin{equation}\label{eq:ZunbiasedEV}
    Z = \frac{2^n}{(q+1)^n} \EV_{P_u,\Lambda_u}\left[ \left(\frac{2q}{q^2+1} \right)^{\left(\text{\# of bit flips during walk}\right)}\right]\,,
\end{equation}
where $\EV_{P_u,\Lambda_u}$ indicates the expectation over the choosing a length-$s$ walk as described above, where the initial distribution is $\Lambda_u$. This is seen to be equivalent to Eq.~\eqref{eq:partitionfunctiontrajectories} since the probability of a certain trajectory occurring is given by $q^{-n}(1/2)^{\text{\# of bit flips}}$ and thus each trajectory contributes exactly the same amount toward $Z$, once the probability of observing the trajectory is accounted for. 

\subsection{Biased random walk}

A potential problem with the unbiased random walk picture is that the weight of a particular walk is related to the number of bit flips that occur during that walk; it depends not only where the walk begins and ends but also on how it got there. To fix this issue, we can form an equivalent \textit{biased} random walk denoted $P_b$. In this case, the initial distribution $\Lambda_b$ is not uniform over $\{I,S\}^n$, rather, the probability of choosing $\vec{\gamma}^{(0)} = \vec{\nu}$ is proportional to $q^{-\lvert \vec{\nu} \rvert}$, where $\lvert \vec{\nu} \rvert$ is the Hamming weight of $\vec{\nu}$ (number of $S$ entries). Specifically, we have
\begin{equation}
    \Lambda_b(\vec{\nu}) = \frac{q^n}{(q+1)^n}q^{-\lvert \vec{\nu} \rvert}\,.
\end{equation}
The dynamics of $P_b$ are the same as $P_u$, except that when the two bits involved in a gate disagree, it chooses to flip the $S$ to $I$ with probability $q^2/(q^2+1)$ and $I$ to $S$ with probability $1/(q^2+1)$. Thus, it is biased in the $I$ direction. Then, we can express
\begin{equation}\label{eq:ZbiasedEV}
    Z = \frac{1}{q^{n}} \EV_{P_b,\Lambda_b}\left[ q^{\lvert \vec{\gamma}^{(s)} \rvert}\right]\,.
\end{equation}
Note that the quantity being averaged is exponentially large in the Hamming weight of its final ending point, making the quantity sensitive to the probability that the biased walk stays far from the all $I$ configuration. The biased walk is observed to be equivalent to the unbiased walk simply by noting that once the probability of observing a certain trajectory is included, every trajectory contributes the same amount to $Z$ for both walks. The exponential weighting underneath the expectation in the biased walk exactly cancels the bias in the probability of observing a certain walk.

\subsection{Computing sums over trajectories}

Throughout our analysis, we will need to compute weighted sums over various trajectories, or, relatedly, compute probabilities that the biased and unbiased walks end in a certain place. We use the following lemma. The key takeaway is that (perhaps surprisingly), in the limit of infinite size, the contribution of all trajectories originating from a certain initial configuration \textit{depends only on the Hamming weight} of that initial configuration, and not the configuration itself. Moreover, this contribution can be calculated. This lemma is a more precise and generalized version of the recursive calculation of $Q(x)$ in \autoref{sec:overviewofmethod} in the main text. 

\begin{lemma}\label{lem:sumovertrajectories}
    Fix an infinite-size circuit diagram, that is, an infinite sequence of qudit pairs $A=(A^{(1)},A^{(2)},\ldots)$. Also fix integers $0\leq x,y,m \leq n$ such that $y\leq x<y+m$, as well as an initial configuration $\vec{\gamma}^{(0)}$ such that $\lvert \vec{\gamma}^{(0)}\rvert = x$. For each $s\geq 0$, let $\mathcal{T}_s$ be the set of length-$s$ trajectories that
    \begin{enumerate}[(1)]
    \item begin at configuration $\vec{\gamma}^{(0)}$
    \item have a non-zero contribution to $Z$ for the circuit diagram $(A^{(1)},\ldots,A^{(s)})$ formed by truncating $A$ to length $s$
    \item end at any configuration $\vec{\gamma}^{(s)}$ for which $\lvert \vec{\gamma}^{(s)} \rvert = y$, and
    \item satisfy $y < \lvert \vec{\gamma}^{(t)} \rvert < y+m$ for all $t=0,1,2,\ldots,s-1$. 
    \end{enumerate}
    Let $\mathcal{T}=\bigcup_{s=0}^\infty \mathcal{T}_s$. Then
    \begin{equation}
        \sum_{\gamma \in \mathcal{T}} \left( \frac{q}{q^2+1} \right)^{(\text{\# of bit flips during } \gamma)} = \frac{1}{1-q^{-2m}}\left(q^{-(x-y)}-q^{-2m+x-y}\right)\,.
    \end{equation}
\end{lemma}
\begin{proof}
    First, we claim that the sum should depend only on $x$, $y$, and $m$, and not on $\vec{\gamma}^{(0)}$ (other than through its dependence on $x$). To see this, note that there is a one-to-one correspondence between trajectories in $\mathcal{T}$ and sequences of Hamming weights $(x,x_1,\ldots,x_{s'-1},y)$ with the property that either $x_t = x_{t+1}+1$ or $x_t = x_{t+1}-1$ for every $t$ (no consecutive duplicates). This is seen by (1) the fact that given a trajectory in $\mathcal{T}$, one can generate such a sequence by taking the Hamming weight of each configuration in the sequence and removing consecutive duplicates and (2) the fact that given such a Hamming weight sequence one can generate a unique trajectory by starting with $\vec{\gamma}^{(0)}$, evolving the trajectory according to the circuit diagram $A$, and always choosing whether to flip $I$ to $S$ or $S$ to $I$ so that the order of Hamming weights prescribed by the sequence is followed. Thus, the sum over trajectories in $\mathcal{T}$ may be replaced by a sum over Hamming weight sequences, which does not depend on $\vec{\gamma}^{(0)}$, except through its Hamming weight $x$. 
    
    For each $x$ in the interval $[y,y+m]$, let the expression on the left-hand-side of the lemma be given by $Q(x)$. Then for each $x$ in $[y+1,y+m-1]$, we have the recursion relation
    \begin{equation}
        Q(x) = \frac{q}{q^2+1}(Q(x-1)+Q(x+1))\,,
    \end{equation}
    since the first bit flip will either send $x$ to $x-1$ or to $x+1$ and in either case a factor of $q/(q^2+1)$ is incurred. The recursion relation gives rise to a general solution of the form
    \begin{equation}
        Q(x) = F q^{x} + G q^{-x}
    \end{equation}
    for some constants $F$ and $G$. This is a unique solution since all values can be generated once two consecutive values are specified, and the specification of two consecutive values also uniquely specifies $F$ and $G$. To find $F$ and $G$ in this case, we must also impose the boundary conditions $Q(y)=1$ and $Q(y+m)=0$, since if $x=y$ the only trajectory in $\mathcal{T}$ is the length-0 trajectory $(\vec{\gamma}^{(0)})$, and if $x=y+m$, $\mathcal{T}$ is the empty set. By specifying these boundary conditions we can solve for $F$ and $G$ and verify the statement of the lemma. 
\end{proof}

\begin{corollary}\label{cor:biasedprob}
    Fix non-negative integers $x,y,m$ such that $y \leq x < y+m$. For the biased walk, if the starting configuration has Hamming weight $x$, the probability that the walk reaches a configuration with Hamming weight $y$ before it reaches a configuration with Hamming weight $y+m$ is given by 
    \begin{equation}
        \frac{q^{x-y}}{1-q^{-2m}}\left(q^{-(x-y)}-q^{-2m+x-y}\right)\,.
    \end{equation}
\end{corollary}
\begin{proof}
    The transition rules of the biased walk prescribe that transitions upward in Hamming weight occur with probability $1/(q^2+1)$, and transitions downward in Hamming weight occur with probability $q^2/(q^2+1)$. Thus the probability of a series of transitions in which the initial Hamming weight is $x$, the final Hamming weight is $y$, and the number of times a bit flip occurs is $b$ is precisely $q^{x-y}(q/(q^2+1))^{b}$. The sum over all paths weighted by their probability is then precisely the sum in the left-hand-side of \autoref{lem:sumovertrajectories} scaled by $q^{x-y}$, yielding the corollary. 
\end{proof}
\begin{corollary}\label{cor:biasedPIPS}
    If we begin at a trajectory $\vec{\gamma}^{(0)}$ with $\lvert \vec{\gamma}^{(0)} \rvert = x $ and allow the biased walk to evolve until it ends at one of the fixed points $I^n$ or $S^n$, then the probability that the trajectory ends at $I^n$ is given by
    \begin{align}
        P_I(x) &= \frac{1}{1-q^{-2n}}\left(1-q^{-2n+2x}\right)
    \end{align}
    and the probability that it ends at $S^n$ is given by
    \begin{align}
        P_S(x) &=  \frac{q^{-2n+2x}}{1-q^{-2n}}\left(1-q^{-2x}\right)\,.
    \end{align}
\end{corollary}
\begin{proof}
    Termination at $I^n$ corresponds to the cases where Hamming weight 0 is hit before Hamming weight $n$. Thus the equation for $P_I(x)$ follows from \autoref{cor:biasedprob} with $y=0$ and $m=n$. We have $P_S(x) = 1-P_I(x)$, since the trajectory must terminate at one fixed point or the other. 
\end{proof}

\subsection{Sanity check: Infinite circuit size convergence to Haar value}\label{sec:sanitycheckZH}

The Markov chain has two stationary distributions, at configurations $I^n$ and $S^n$. In the infinite circuit size limit, the biased walk will converge to a mixture of these two fixed-point configurations, where the amount of mass at each fixed-point depends only on the Hamming weight of the initial configuration, as described in \autoref{cor:biasedPIPS}. Using the expressions for $P_I$ and $P_S$, we find that, in the infinite circuit size limit,
\begin{align}
    Z &= \frac{1}{q^n} \sum_{\vec{\gamma}^{(0)}}\Lambda_b(\vec{\gamma}^{(0)})\EV_{P_b,\vec{\gamma}^{(0)}}\left[q^{\lvert\vec{\gamma}^{(s)}\rvert}\right]\\
    &= \frac{1}{(q+1)^n} \sum_{\vec{\gamma}^{(0)}}q^{-\lvert \vec{\gamma}^{(0)} \rvert}(P_I(\lvert \vec{\gamma}^{(0)} \rvert) + q^n P_S(\lvert \vec{\gamma}^{(0)} \rvert)) \\
    &= \frac{1}{(q+1)^n(1-q^{-2n})}\sum_{x=0}^n \binom{n}{x}q^{-x}(1-q^{-2n+2x} + q^{-n+2x}-q^{-n}) \\
    &= \frac{1}{(q+1)^n(1-q^{-2n})}\left(\left(\frac{q+1}{q}\right)^n-q^{-2n}(q+1)^n+q^{-n}(q+1)^n-q^{-n}\left(\frac{q+1}{q}\right)^n \right) \\
    &= \frac{(2q^{-n}-2q^{-2n})(q+1)^n}{(q+1)^n(1-q^{-2n})} \\
    &= \frac{2}{q^n+1} = Z_H\,,
\end{align}
where $Z_H$ is the Haar value. This outcome is expected since in the infinite circuit size limit the distribution over random unitaries formed from Haar-random local components will approach the distribution over $n$-qudit Haar random unitaries.

\section{Bounds for general architectures}\label{sec:generalbounds}

\subsection{Upper bound on collision probability}

In order to have a meaningful upper bound, we need the architecture to satisfy basic connectivity requirements; for example, if the architecture performs gates on the same pair of qudits over and over again, $Z$ will never decrease and the output distribution never become anti-concentrated. We need to rule out this sort of architecture.

Recall that an RQC architecture is a (possibly randomized) procedure for choosing a length-$s$ sequence $(A^{(1)},\ldots, A^{(s)})$ of pairs of qudit indices on which to perform a Haar-random gate. 

\begin{definition}[Regularly connected]\label{def:regularlyconnected}
    We say an RQC architecture is $r$-\textit{regularly connected} if for any $n$, any $t$, any subsequence $A=(A^{(1)},\ldots,A^{(t)})$ and any proper subset $R \subset [n]$ of qudit indices, there is at least a $1/2$ probability that, conditioned on the first $t$ gates in the gate sequence being $A$, there exists some index $t'$ for which $t< t' \leq t+rn$, $A^{(t')} \cap R \neq \emptyset$, and $A^{(t')} \not\subset R$.
\end{definition}

The above definition requires that given any partition of the qudits into two sets, we should expect at least one gate to couple a qudit from one set with a qudit from the other set after only a linear number of gates. Note that both the 1D and the complete-graph architecture have this property. In 1D, it only takes two layers, or $n$ gates, to guarantee having performed a gate that crosses any partition one might choose. Similarly, in the complete-graph architecture, the probability that a randomly chosen gate crosses a partition is at least $1/n$ (which happens if the partition splits the indices into a set with one index and a set with the other $n-1$ indices), and the probability of having crossed the partition becomes large after $\Theta(n)$ gates. Most natural architectures we might consider have this property. One architecture that is not regularly connected is the hypercube architecture, where $n=2^{D}$ qudits lie at the vertices of a $D$-dimensional hypercube, and $D$ layers of gates are performed by cycling through each set of parallel edges. In this architecture, it would take $nD/2 = \Theta(n\log(n))$ gates to guarantee that any partition has been crossed.

Assuming the regularly connected property, we can show a weak upper bound on the collision probability. 

\begin{theorem}[\autoref{thm:generalupperboundsummary} from main text]\label{thm:generalUB}
    If an RQC architecture is $r$-regularly connected, then the collision probability satisfies 
    \begin{equation}
        Z \leq Z_H\left(1+e^{-\frac{2a}{n}(s-s^*)}\right)\,,
    \end{equation}
    where
    \begin{align}
    a&:= (2r)^{-1}\log\left(\frac{2(q^2+1)}{(q+1)^2}\right)\\
    s^*&:= (2a)^{-1}\log\left(\frac{2q}{q+1}\right)n^2 + O(n)\,.
\end{align}
\end{theorem}
\begin{proof}
We use the expression given to us by the unbiased walk in Eq.~\eqref{eq:ZunbiasedEV}
\begin{equation}
    Z = \frac{2^n}{(q+1)^n} \EV_{P_u,\Lambda_u}\left[ \left(\frac{2q}{q^2+1} \right)^{\left(\text{\# of bit flips during walk}\right)}\right]\,.
\end{equation}
Define $Z^{(t)}$ to be the value of the collision probability, given above via the biased walk, after $t$ time steps, so $Z = Z^{(s)}$ and $Z^{(0)} = 2^n/(q+1)^n$. 

Consider a given trajectory produced by the unbiased walk up to time step $t$, $\gamma = (\vec{\gamma}^{(0)},\ldots,\vec{\gamma}^{(t)})$. If $\vec{\gamma}^{(t)} = I^n$ or $\vec{\gamma}^{(t)} = S^n$ then the walk has reached a fixed point and will never again change. From the calculation in \autoref{sec:sanitycheckZH}, we know that the sum of all the weights of all walks \textit{of any length} that reach a fixed point is precisely $Z_H$. Since the weights are non-negative, this implies that the sum over walks that have reached it before time step $t$ is less than $Z_H$, and hence the combined weight of trajectories that have \textit{not} reached a fixed point by time step $t$ is at least $Z^{(t)}-Z_H$. Meanwhile, if $\vec{\gamma}^{(t)}$ is not at a fixed point, then we can consider the proper subset $R \subset [n]$ of sites with value $S$. By the 
$r$-regularly connected property, there is at least a 1/2 chance that one of the gates between time step $t+1$ and $t+rn$ matches an index in $R$ with one in the complement of $R$. When this happens, a bit must be flipped and the weight of that trajectory is reduced by factor $2q/(q^2+1)$. Thus, the following must hold
\begin{equation}
    Z^{(t+rn)}-Z_H \leq \left(\frac{1}{2} + \frac{1}{2}\frac{2q}{q^2+1}\right)\left(Z^{(t)}-Z_H\right)\,.
\end{equation}

Moreover, we know that $Z^{(0)} = 2^n/(q+1)^n$, so
\begin{align}
    Z^{(s)} &\leq Z_H + \left(\frac{(q+1)^2}{2(q^2+1)}\right)^{s/(rn)}\left(\frac{2^n}{(q+1)^n} - Z_H\right) \\
    &\leq Z_H + \left(\frac{(q+1)^2}{2(q^2+1)}\right)^{s/(rn)}\left(\frac{2^n}{(q+1)^n}\right) \\
    &= Z_H\left(1+\frac{2^{n}(q^n+1)}{2(q+1)^n}\left(\frac{(q+1)^2}{2(q^2+1)}\right)^{s/(rn)}\right) \\
    &\leq Z_H\left(1+\frac{2^nq^{n}}{(q+1)^n}\left(\frac{(q+1)^2}{2(q^2+1)}\right)^{s/(rn)}\right) \\
    &\leq Z_H(1+e^{-\frac{2a}{n}(s-s^*)})\,,
\end{align}
where
\begin{align}
    a&:= (2r)^{-1}\log\left(\frac{2(q^2+1)}{(q+1)^2}\right) = \Theta(1)\\
    s^*&:= (2a)^{-1}\log\left(\frac{2q}{q+1}\right)n^2 = \Theta(n^2)\,.
\end{align}
\end{proof}
Note that we have made no attempt to optimize the constant prefactor of the $\Theta(n^2)$ or the value of $a$. Indeed, we conjecture that \autoref{thm:generalUB} could be improved so that $s^* = \Theta(n\log(n))$, which would be a dramatic improvement that implies the fundamental scaling of the anti-concentration size is independent of the architecture's connectivity, so long as it satisfies the regularly connected property. 

\subsection{Lower bound on collision probability}

In this section, we prove an $\Omega(n\log(n))$ lower bound on the circuit size needed for anti-concentration in general circuit architectures. This also implies an $\Omega(\log(n))$ lower bound on the anti-concentration depth.

\begin{theorem}[\autoref{thm:generallowerboundsummary} from main text]\label{thm:generalLB}
    For any RQC architecture of size $s$ on $n$ qudits with local dimension $q$, the collision probability satisfies
    \begin{equation}\label{eq:generalLB}
        Z \geq \frac{Z_H}{2} \exp\left(\frac{\log(q)}{q+1}\exp\left(\log(n)-\frac{2s}{n}\log\left(q^2+1\right)\right)\right)\,.
    \end{equation}
\end{theorem}
\begin{corollary}
    For a given RQC architecture, let $s_{AC}$ be the minimum circuit size, as a function of $n$, such that $Z \leq 2Z_H$. Then, it must hold that
    \begin{equation}
        s_{AC} \geq \left(2\log\left(q^2+1\right)\right)^{-1}n\log(n) -  O(n)\,.
    \end{equation} 
\end{corollary}
\begin{proof}
    This statement follows directly from the bound in Eq.~\eqref{eq:generalLB}.
\end{proof}
\begin{corollary}
    For a given RQC architecture, let $d_{AC}$ be the minimum circuit depth, as a function of $n$, such that $Z \leq 2 Z_H$. Then, it must hold that
    \begin{equation}
        d_{AC} \geq \left(\log\left(q^2+1\right)\right)^{-1}\log(n) - O(1)\,.
    \end{equation} 
\end{corollary}
\begin{proof}
    Each layer can have at most $n/2$ gates, so it must hold that $d_{AC} \geq 2s_{AC}/n$. 
\end{proof}
\begin{proof}[Proof of \autoref{thm:generalLB}]
    We use the framework of the biased random walk, given by the expression for $Z$ in Eq.~\eqref{eq:ZbiasedEV}. For each of the $n$ sites, there is some initial probability that it starts with value $S$, and then each gate involving that site has some chance of flipping it to value $I$. However, there will always be some minimum probability that even after many gates, the value has not yet been flipped to $I$. This constitutes the idea behind our lower bound.
    
    Given an index $j \in [n]$, we compute a lower bound on the probability that $\gamma_j^{(t)}=S$ for all $t = 0,1,\ldots, s$, (i.e.~the $j$th bit begins with value $S$ and is never flipped to $I$), as a function of the number of gates $s_j$ that act on qudit $j$
    \begin{equation}
        \Pr_{P_b,\Lambda_b}[\gamma_j^{(t)}=S\;\; \forall t \in \{0,\ldots,s\}] \geq \frac{1}{q+1}\left(\frac{1}{q^2+1}\right)^{s_j}\,,
    \end{equation}
    since there is a $1/(q+1)$ chance that $\gamma_j^{(0)}=S$ when we draw $\vec{\gamma}^{(0)}$ from $\Lambda_b$, and the probability it does not flip after each gate is at least $1/(q^2+1)$. This holds for each $j$, and thus we have
    \begin{align}
        \EV_{P_b,\Lambda_b}[\lvert \vec{\gamma}^{(s)}\rvert] &= \sum_{j=1}^n \Pr_{P_b,\Lambda_b}[\gamma^{(s)}_j = S] \\
        &\geq \sum_{j=1}^n \Pr_{P_b,\Lambda_b}[\gamma_j^{(t)}=S \;\; \forall t \in \{0,\ldots,s\}] \\
        &\geq \frac{1}{q+1}\sum_{j=1}^n\left(\frac{1}{q^2+1}\right)^{s_j}\,.
    \end{align}
    Since each of the $s$ gates in the circuit diagram acts on two indices, it must hold that $\sum_j s_j = 2s$, and given this constraint, the minimum of the final expression above occurs when all the  $s_j$ are equal, and thus
    \begin{equation}
        \EV_{P_b,\Lambda_b}[\lvert \vec{\gamma}^{(s)}\rvert] \geq \frac{1}{q+1}n\left(\frac{1}{q^2+1}\right)^{2s/n}\,.
    \end{equation}
    By convexity of the exponential function, we have $\EV[q^x] \geq q^{\EV[x]}$, and hence
    \begin{align}
        Z &= \frac{1}{q^n}\EV_{P_b,\Lambda_b}[q^{\lvert \vec{\gamma}^{(s)}\rvert}] \\
        &\geq \frac{1}{q^n}\exp\left(\log(q)\frac{n}{q+1}\left(\frac{1}{q^2+1}\right)^{2s/n}\right) \\
        &\geq \frac{Z_H}{2} \exp\left(\frac{\log(q)}{q+1}\exp\left(\log(n)-\frac{2s}{n}\log\left(q^2+1\right)\right)\right)\,.
    \end{align}
\end{proof}

\section{Bounds for the 1D architecture}\label{sec:1DBounds}

We now focus specifically on the 1D architecture defined formally in \autoref{def:1Darchitecture}. We assume periodic boundary conditions, although it would be possible to consider open boundary conditions as well. In 1D, the qudits are arranged in a geometrically local fashion and it is fruitful to think of a configuration $\vec{\gamma} \in \{I,S\}^n$ as being composed of contiguous \textit{domains}, consecutive sites where all the values are $I$ or all the values are $S$. We then identify \textit{domain walls} as locations where one domain ends and another begins. Gates that couple qudits in different domains cause one of the values to flip, which moves the domain wall separating those domains one unit to the left or one unit to the right.  The notation for talking formally about this is discussed in the next subsection, and then the upper and lower bounds on $Z$ are proved. 

\subsection{Domain walls and notation}

In 1D, configurations $\vec{\gamma}\in \{I,S\}^n$ are associated with a set of domain wall locations. We let
\begin{equation}
    DW(\vec{\gamma}) := \{e \in \{0,1,2,\ldots,n-1\}: \gamma_e \neq \gamma_{e+1}\}
\end{equation}
be the set of domain wall positions for a configuration $\vec{\gamma}$, where $\gamma_0$ is identified with $\gamma_n$ when there are periodic boundary conditions. For each set of domain wall locations there are exactly two configurations that map to it, since choosing $\gamma_0 = I$ or $\gamma_0 = S$ determines the value of all other sites. 

A configuration trajectory $\gamma = (\vec{\gamma}^{(0)},\ldots,\vec{\gamma}^{(s)})$ is then associated with a sequence of sets of domain wall locations $G = (g^{(0)},\ldots,g^{(s)})$ where $g^{(t)} = DW(\vec{\gamma}^{(t)})$. We call $G$ a \textit{domain wall trajectory}. Domain wall trajectories with non-zero contribution to the collision probability $Z$ obey the following rules: when there is a domain wall at position $e$ and a gate acts on qudits $\{e,e+1\}$, the domain wall must move to position $e-1$ or $e+1$ (at the cost of a reduction in the weight) and may annihilate with another domain wall if there is already a domain wall at the new position. However, pairs of domain walls cannot be created; the number of domain walls that exist throughout the domain wall trajectory is non-increasing, and a particular domain wall can be uniquely tracked throughout each step of the trajectory (either until the final step or until its annihilation). Let $\mathcal{G}$ be the set of all domain wall trajectories that obey these rules. Any domain wall trajectory $G \in \mathcal{G}$ will have the property that when $t$ is odd, $e$ is even for all $e \in g^{(t)}$, and when $t$ is even (but non-zero), $e$ is odd for all $e \in g^{(t)}$. This is because odd (even) numbered layers couple qubits $\{2j-1,2j\}$ ($\{2j, 2j+1\}$) meaning domain walls must lie between qudit positions $2j$ and $2j+1$ (between qudit positions $2j-1$ and $2j$) for some $j$. 

By converting the sum over trajectories in Eq.~\eqref{eq:partitionfunctiontrajectories} to a sum over domain wall trajectories, we can express $Z$ by the equation
\begin{equation}\label{eq:collprobDWtraj}
    Z = \frac{2}{(q+1)^n}\sum_{G \in \mathcal{G}} \text{weight}(G)\,,
\end{equation}
where the weight is given as follows, recalling that $A^{(t)}$ is the pair of qudit indices involved in the $t$th gate, which in 1D is always $A^{(t)} = \{j,j+1\}$ for some $j$.
\begin{align}
    \text{weight}(G) &:= \prod_{t=1}^s M^{(t)}_{g^{(t-1)}g^{(t)}}\\
    M^{(t)}_{g^{(t-1)}g^{(t)}}&:= \begin{cases}
    \frac{q}{q^2+1} & \text{if } \min(A^{(t)}) \in g^{(t-1)} \\
    1 & \text{otherwise}
    \end{cases}
\end{align}
In other words, if the gate on qudits $\{j,j+1\}$ and there is a domain wall at position $j$, then the weight is reduced by a factor $q/(q^2+1)$ (and the domain wall must move to position $j-1$ or position $j+1$, possibly annihilating if a domain wall already exists at that position). 

Given two domain wall trajectories $G$ and $G'$, we will consider the combined domain wall trajectory
\begin{equation}
    G \sqcup G' := (g^{(0)} \sqcup {g'}^{(0)}, \ldots, g^{(s)} \sqcup {g'}^{(s)})\,,
\end{equation}
where $\sqcup$ is the disjoint union and is defined only under the assumption $g^{(t)} \cap {g'}^{(t)} = \emptyset$ for all $t$. 

The upshot of thinking about trajectories this way is that if $H = G \sqcup G'$ then
\begin{equation}
    \text{weight}(H) = \text{weight}(G)\; \text{weight}(G')\,.
\end{equation}

In particular, we will find it useful to decompose a domain wall trajectory $G$ into $G = G_U \sqcup G_0$ where $G_U$ is a domain wall trajectory with a conserved number of domain walls throughout the trajectory, and $G_0$ is a trajectory for which $|G_0^{(s)}|=0$, i.e.~all the domain walls have annihilated by the end of the trajectory. This decomposition is unique, and an example is shown in \autoref{fig:1Dpaths}. Let $\mathcal{G}_U$ and $\mathcal{G}_0$ be the subsets of $\mathcal{G}$ that have no annihilations and that have no surviving domain walls at the end of the circuit, respectively. Let $\mathcal{G}_{U,k}$ be the subset of $\mathcal{G}_U$ with $k$ domain walls. When the boundary conditions are periodic, $k$ must be even for $G_{U,k}$ to be non-empty. 

\begin{figure}
\centering
\includegraphics[width=0.72\linewidth]{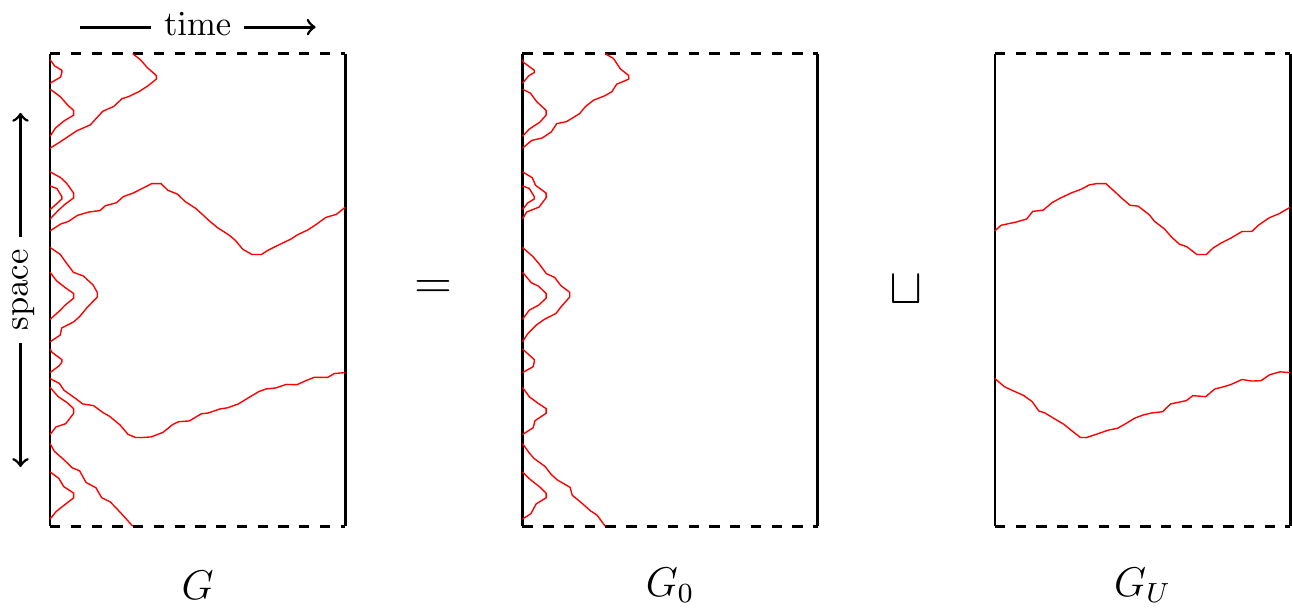}
\caption{Cartoon illustrating unique decomposition of a domain wall trajectory $G$ into a disjoint union of one part, $G_0$, where all domain walls annihilate prior to the end of the circuit, and another part, $G_U$, where no domain walls annihilate.}
\label{fig:1Dpaths}
\end{figure}

\subsection{Collision probability upper bound}

\begin{theorem}[\autoref{thm:1Dupperboundsummary} from main text]\label{thm:1DcollisionprobUB}
For the 1D architecture, let
\begin{align}
    a &:= \log\left(\frac{q^2+1}{2q}\right) \label{eq:1Da}\\
    s^* &:= \frac{1}{2a}n\log(n)+n\left(\frac{1}{2a}\log(e-1)+\frac{1}{2}\right) = (2a)^{-1}n\log(n) + O(n) \label{eq:1DdAC}\,.
\end{align}
Then,
\begin{equation}
    Z \leq Z_H(1+e^{-\frac{2a}{n}(s-s^*)})
\end{equation}
whenever $s \geq s^*$. The circuit depth $d$ is $d=2s/n$, so we may define $d^* = 2s^*/n$ and equivalently conclude
\begin{equation}
    Z \leq Z_H(1+e^{-a(d-d^*)})\,.
\end{equation}
\end{theorem}

Note that when $s < s^*$, an upper bound on  $Z$ can still be inferred from this method.  The essence of the proof of \autoref{thm:1DcollisionprobUB} is the same as the proof of the statement proved in \cite{Barak2020Spoofing}, although we have expressed it here within our notation and framework.

\begin{proof}
We use the formula in Eq.~\eqref{eq:collprobDWtraj}, which expresses $Z$ as a weighted sum over domain wall trajectories. Each domain wall trajectory $G = (g^{(0)},\ldots,g^{(s)})$ can be associated with an integer $k = |g^{(s)}|$, the number of domain walls that remain unannihilated at the end of the trajectory. Due to periodic boundary conditions, $k$ must be even, and let $k_0 = k/2$. Let $\mathcal{G}_k \subset \mathcal{G}$ be the associated set of length-$s$ domain wall trajectories, and let $\mathcal{G}_{U,k} \subset \mathcal{G}_k$ be the subset containing domain wall trajectories that have a conserved number of domain walls throughout. As discussed in the previous subsection, it is possible to uniquely decompose $H \in \mathcal{G}_k$ into $H = G \sqcup G'$ where $G \in \mathcal{G}_{U,k}$ and $G' \in \mathcal{G}_0$. 

Suppose we fix a domain wall configuration ${g}^{(0)}$ for the initial time step at the beginning of the circuit with $k$ domain walls. There are $\binom{n}{k}$ such configurations. The total weight of all the trajectories in $\mathcal{G}_{U,k}$ that begin at this configuration is at most $(2q/(q^2+1))^{k(d-1)}$ since each domain wall must move either left or right (introducing a factor of 2) during each of the  $d$ layers of gates, except for possibly the first layer (if the domain wall begins at an even position it does not move during the first layer), and each time one moves it incurs a weight reduction of $q/(q^2+1)$. This does not account for the rule that the $k$ domain walls cannot intersect, but it still yields an upper bound on the total weight. 

Meanwhile, the sum of the weights of all domain wall trajectories in $\mathcal{G}_0$ approaches $Z_H(q+1)^n/2$ from below as depth increases. This follows from the analysis in \autoref{sec:sanitycheckZH} where it was shown that the sum over all trajectories that eventually reach a fixed point is exactly $Z_H(q+1)^n$, but at a finite depth not every trajectory will have reached a fixed point so only a subset of the terms are included in the sum. Due to the fact that each domain wall configuration corresponds to 2 equal-weight trajectories through $\{I,S\}^n$ the sum of the weights of all the domain wall trajectories in $\mathcal{G}_{0}$ can be at most $Z_H(q+1)^n/2$.

Collecting these observations, and recalling $k = 2k_0$, we have
\begin{align}
    Z &= \frac{2}{(q+1)^n}\sum_{k_0=0}^{n/4}\sum_{G \in \mathcal{G}_{2k_0}} \text{weight}(G)\\
    &= \frac{2}{(q+1)^n}\sum_{k_0=0}^{n/4}\sum_{G \in \mathcal{G}_{U,2k_0}}\sum_{\substack{G' \in \mathcal{G}_0\\ G\cap G' = \emptyset}} \text{weight}(G) \cdot \text{weight}(G') \\
    &\leq \left(\sum_{k_0=0}^{n/4}\sum_{G \in \mathcal{G}_{U,2k_0}}\text{weight}(G)\right)\left(\frac{2}{(q+1)^n}\sum_{G' \in \mathcal{G}_0} \text{weight}(G')\right) \\
    &\leq \left( \sum_{k_0=0}^{n/4}\binom{n}{2k_0} \left(\frac{2q}{q^2+1}\right)^{2k_0(d-1)} \right) \left( Z_H\right)  \\
    &= Z_H \sum_{k_0=0}^{n/4}\binom{n}{2k_0} {(e^{-a})}^{2k_0(d-1)} \\
    &\leq Z_H(1+e^{-a(d-1)})^{n}\\
    &\leq Z_H(1+(e-1) ne^{-a(d-1)})\label{eq:1Dupperboundintermediate} \\
    &= Z_H\left(1+\exp\left(\log(n)-da+ \log\left(e-1\right)+a\right)\right) \\
    %
    %
    &\leq Z_H\left(1+\exp\left(-a(d-d^*)\right)\right)\,,
\end{align}
where Eq.~\eqref{eq:1Dupperboundintermediate} holds so long as $d \geq d^*$, based on the following small lemma.
\begin{lemma}
    If $b,c > 0$ and $cb \leq 1$ then 
    \begin{equation}
        (1+c)^b \leq 1+cb(e-1)\,.
    \end{equation}
\end{lemma}
\begin{proof}
    \begin{align}
        (1+c)^b &= \sum_{k=0}^b \binom{b}{k}c^k = 1+cb\sum_{k=1}^b \binom{b}{k}\frac{c^{k-1}}{b}\\ &\leq 1+cb\sum_{k=1}^b \binom{b}{k}b^{-k} \leq 1+ cb \left((1+b^{-1})^b-1 \right)\\
        &\leq 1+cb(e-1)\,.
    \end{align}
\end{proof}
\vspace*{-6pt}
\end{proof}

\subsection{Collision probability lower bound}

\begin{theorem}[\autoref{thm:1Dlowerboundsummary} from main text]\label{thm:1DcollisionprobLB}
    Consider the 1D architecture. There are constants $A$ and $A'$ such that as long as $s^*-s \geq A'n$, the collision probability satisfies
    \begin{equation}
            Z \geq \frac{Z_H}{2} \exp\left(Ae^{\log(n)-\frac{2a}{n}s} \right)\,.
    \end{equation}
    where $a$ and $s^*$ are the same as in \autoref{thm:1DcollisionprobUB}.
\end{theorem}
In our proof, the constant $A$ is explicit but very small, on the order of $e^{-10}$, and $A'\approx -\log(A)$. The value of $A$ could certainly be improved with some attempt at optimization.
\begin{corollary}
    For the 1D architecture, if we define $s_{AC}$ and $d_{AC}$ to be the smallest circuit size and circuit depth for which $Z \leq 2 Z_H$, then 
    \begin{align}
        \left \lvert s_{AC} - \left(2\log\left(\frac{q^2+1}{2q}\right)\right)^{-1}n\log(n)  \right\rvert  &\leq O(n) \\
        \left \lvert d_{AC} - \left(\log\left(\frac{q^2+1}{2q}\right)\right)^{-1}\log(n) \right \rvert &\leq O(1) \,.
    \end{align}
\end{corollary}
\begin{proof}
     \autoref{thm:1DcollisionprobUB} implies that 
     \begin{equation}
         s_{AC} \leq s^* = (2a)^{-1}n\log(n) + O(n)\,.
     \end{equation}
     Meanwhile, \autoref{thm:1DcollisionprobLB} implies that if 
     \begin{equation}
         s \leq (2a)^{-1}n\log(n)-\max\bigg((2a)^{-1}\log(\log(4)A^{-1}),A'\bigg)n  = (2a)^{-1}n\log(n) - O(n)
     \end{equation}
     then $Z\geq 2Z_H$. Hence $s_{AC} \geq (2a)^{-1}n\log(n)-O(n)$. Together these implies $|s_{AC}-(2a)^{-1}n\log(n)| = O(n)$. 
\end{proof}

\begin{proof}[Proof of \autoref{thm:1DcollisionprobLB}]
    Eq.~\eqref{eq:collprobDWtraj} expresses $Z$ as a weighted sum over domain wall trajectories. Heuristically, when $d<d^*$ we expect that the output distribution will \textit{not} be anti-concentrated and that domain wall trajectories drawn at random with probability proportional to its weight will usually have many domain walls that never annihilate. To lower bound $Z$, we will sum over the set of configurations with $k$ unannihilated domain walls for a particularly chosen value of $k$. 
    
    For a fixed value of the depth $d$, define
    \begin{equation}
        n_H := \frac{e^{(d-1)a}}{2(e-1)}\,.
    \end{equation}
    We chose $n_H$ to be exactly half the value of $n$ for which a depth-$d$ circuit would be anti-concentrated. Heuristically, we expect on the order of $n/2n_H$ unannihilated domain walls in typical configurations.
    
    Let $k$ be an even integer to be specified later. Let $\mathcal{H}_k \subset \mathcal{G}_{U,k}$ be the set that contains any domain wall trajectory $H = (h^{(0)},\ldots,h^{(s)})$ for which
    \begin{enumerate}[(1)]
        \item $H$ has $k$ domain walls at each time step (none annihilate)
        \item For each of the $k$ domain walls in the initial configuration $h^{(0)}$, the nearest domain wall in both directions is at most $n_H$ positions away.
    \end{enumerate}
    
    Now, temporarily fix some $H \in \mathcal{H}_k$. It has $k$ domain walls which move around throughout the trajectory. We let $e_{H,j,t}$ be the location of the $j$th domain wall at time step $t$ in the trajectory $H$. We then define the set $\mathcal{J}_{H,j} \subset \mathcal{G}_0$, for $j = 1,\ldots,k$ to be the set of domain wall trajectories for which (1) all of the domain walls annihilate before time step $s$ and (2) the position $e_t$ of any domain wall at time step $t$ satisfies
    \begin{equation}
        e_{H,j,t} < e_t < e_{H,j+1,t}\,.
    \end{equation}
    In other words, all of the domain walls fall between the $j$th and $(j+1)$th domain walls of $H$. This ensures that $H$ is disjoint from any $J_j \in \mathcal{J}_{H,j}$. 
    
    Specifying a trajectory $H \in \mathcal{H}_k$ as well as $J_j \in \mathcal{J}_{H,j}$ for each $j = 1,\ldots, k$, determines a unique trajectory $H' = H \sqcup J_{H,1}\sqcup \ldots \sqcup J_{H,k}$. This decomposition is illustrated in \autoref{fig:1Dlowerboundpaths}.
    \begin{figure}[h]
    \centering
    \includegraphics[width=0.88\linewidth]{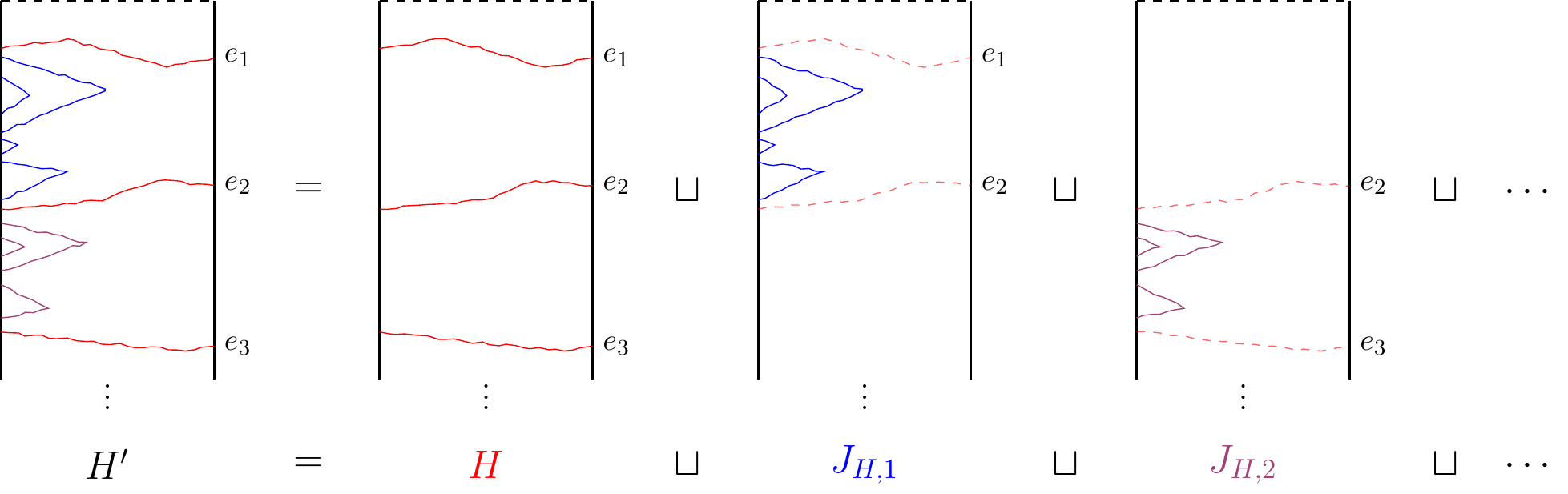}
    \caption{Outline of the main idea of the proof of \autoref{thm:1DcollisionprobLB}. We choose a domain wall trajectory $H$ which has $k$ domain walls that never annihilate and such that the distance between consecutive domain walls is always at most $n_H$. We then choose domain wall trajectories $J_{H,1},\ldots,J_{H,k}$ such that the domain walls of $J_{H,j}$ lie between the $j$th and $(j+1)$th domain walls of $H$ and all annihilate before the end of the circuit. The domain wall configuration $H'$ is the disjoint union of $H$ and $J_{H,j}$ for $j=1,\ldots,k$. We can lower bound the collision probability by lower bounding the weighted sum over the contribution from all $H'$ formed this way. }
        \label{fig:1Dlowerboundpaths}
    \end{figure}
    Thus, if we perform the weighted sum only over the set of $H'$ formed this way, we will arrive at a lower bound to $Z$, as follows:
    \begin{align}
        Z &= \frac{2}{(q+1)^n}\sum_{H \in \mathcal{G}} \text{weight}(H) \\
        &\geq \frac{2}{(q+1)^n}\left(\sum_{H \in \mathcal{H}_k}\text{weight}(H)\right)\left(\sum_{J_1 \in \mathcal{J}_{H,1}}\text{weight}(J_1)\right) \cdots \left(\sum_{J_k \in \mathcal{J}_{H,k}}\text{weight}(J_k)\right) \,.
    \end{align}

The quantities in parentheses can be bounded with the following two lemmas, whose proofs are delayed until after the proof of the Theorem. 
\begin{lemma}\label{lem:sumHk}
    If $4d \leq \lfloor n/k \rfloor$ and $n_H/2 \geq \lceil n/k \rceil$ hold, then the set $\mathcal{H}_k$ satisfies
    \begin{equation}
        \sum_{H \in \mathcal{H}_k}\text{weight}(H)  \geq \left(\frac{1}{2}\left\lfloor\frac{n}{k} \right\rfloor\right)^k\left(\frac{2q}{q^2+1}\right)^{dk}\,.
    \end{equation}
\end{lemma}

\begin{lemma}\label{lem:sumJHj}
    Fix a value of $H$ and $j$. Suppose the $j$th and $(j+1)$th domain walls of the initial configuration of $H$ lie at positions $e$ and $e+X-1$ (mod $n$), respectively, for some positive integer $X < n$. Then
    \begin{equation}
        \left(\sum_{J \in \mathcal{J}_{H,j}}{\rm weight}(J)\right) \geq  \frac{1}{c}\left(\frac{q+1}{q}\right)^X\,,
    \end{equation}
    where $c = 3e^{10}$.
\end{lemma}

The sum of the domain length $X$ for each of the domains is simply $n$. Thus the $((q+1)/q)^X$ factors cancel the $1/(q+1)^n$ prefactor for $Z$, and we have
\begin{align}
    Z \geq q^{-n}\left(\frac{1}{2}\left\lfloor\frac{n}{k}\right\rfloor\right)^kc^{-k}\left(\frac{2q}{q^2+1}\right)^{dk} = q^{-n}\left(\frac{1}{2}\left\lfloor\frac{n}{k}\right\rfloor\right)^kc^{-k}e^{-adk} \label{eq:Zexpk}
\end{align}
for any $k$ that satisfies $4d \leq \lfloor n/k \rfloor$ and $n_H/2 \leq \lceil n/k \rceil$.

Now we choose a value of  $k$ to maximize the right-hand-side of the above equation. In the limit of large $n$, the requirement that $k$ is an even integer will have negligible effect. In our analysis, we handle this requirement by defining $k'$ to be a real number and $k$ to be the smallest even integer larger than $k'$, and then we make a few rather crude bounds on the floor and ceiling of quantities like $ n / k $, which are not asymptotically tight but good enough for our purposes.  We choose
\begin{align}
    k'&:=  \frac{n\left(\frac{2q}{q^2+1}\right)^d}{8ce} = \frac{n e^{-da}}{8ce} =\frac{n}{n_H}\frac{e^{-a}}{16e(e-1)c}\\
    k &:= \text{smallest even integer greater than } k'
\end{align}

Note that $n/k'$ is at least $8ce$, which is very large, meaning $\lceil n / 2k' \rceil/2 \leq n/2k' \leq 2\lfloor n / 2k' \rfloor$ certainly holds.  For finite $n$, we can say that as long as $k'\geq 1$, then $k' \leq k\leq 2k'$ will hold. The requirement $k' \geq 1$ translates into
\begin{align}
    d &\leq a^{-1}(\log(n)-\log(8ce)) \,.
\end{align}
which, by recalling $s = nd/2$ and that $s^* \geq (2a)^{-1}n\log(n)$, can be re-expressed as 
\begin{equation}
    s^*-s \geq A'n
\end{equation}
with $A' := (2a)^{-1}\log(8ce(e-1))+2^{-1}$, which is assumed to hold in the Theorem statement. 
This implies that
\begin{equation}
    \left\lfloor \frac{n}{k} \right\rfloor \geq \left\lfloor \frac{n}{2k'} \right\rfloor \geq \frac{n}{4k'}
\end{equation}

Inspection of the formula for $k'$ reveals that the relation $4d \leq \lfloor n/k \rfloor$ holds for any $d$ and $n$. Moreover, we have $n_H/2 = (ne^{-a})/(k'32e(e-1)c)\leq \lceil n/k \rceil$ so the second relation holds as well. 

Recall that $Z_H = 2/(q^n+1) \leq 2q^{-n}$. Plugging in the above bound on $\lfloor n/k \rfloor$ into Eq.~\eqref{eq:Zexpk}, we find
\begin{align}
    Z &\geq \frac{Z_H}{2}\exp(k) \geq \frac{Z_H}{2}\exp(k') \\
    &\geq \frac{Z_H}{2}\exp\left(\frac{n e^{-da}}{8ce}\right) \\
    &= \frac{Z_H}{2}\exp\left(\frac{1}{8ce}e^{\log(n)-da}\right) \\
    &= \frac{Z_H}{2}\exp\left(\frac{1}{8ce}e^{\log(n)-\frac{2as}{n}}\right)\\
    &= \frac{Z_H}{2}\exp\left(Ae^{\log(n)-\frac{2as}{n}}\right)
\end{align}
for $A := 1/8ce$. Note that this value of $A$ is quite small (on the order of $e^{-10}$) but with some optimization could likely be made much larger. 

\end{proof}

Now we provide the delayed proofs of the two lemmas.

\begin{proof}[Proof of \autoref{lem:sumHk}]
Each term in the sum on the left-hand-side is non-negative, so we make a lower bound by summing over a subset of the terms. To do so, we can split the $n$ indices up into $k$ nearly equal-size segments of length at most $\lceil n/k \rceil$, which is less than $n_H/2$ by assumption. Then for each of these segments, we choose the location of a single domain wall that is at least distance $d$ from each edge of the segment. This will generate a unique initial domain wall configuration that satisfies criteria (2) of $\mathcal{H}_k$, since any pair of consecutive domain walls is closer than $n_H$ apart. The total number of choices is at least
\begin{equation}
    \left(\left\lfloor \frac{n}{k} \right\rfloor -2d\right)^k
\end{equation}
which, by the assumption $4d \leq \lfloor n/k \rfloor$, is at least $(\lfloor n/k \rfloor/2)^k$. 

Once the initial $k$ domain wall locations have been chosen, we examine how they can propagate through the circuit. Each layer of gates will force each of the $k$ domain walls to move in one of two directions, and the weight is reduced by a factor $(q/(q^2+1))^k$, except for the first layer, where some of the domain walls may not move if they begin at an even index. Since by construction, there are no instances where domain walls start within a distance of $2d$ of any other domain wall, there is no chance of domain walls crossing. Thus, we find that for each initial set of $k$ locations chosen in the manner outlined above, the combined weight of all possible trajectories is at least $(2q/(q^2+1))^{kd}$. This proves the lemma. 
\end{proof}

\begin{proof}[Proof of \autoref{lem:sumJHj}]
Consider an alternative 1D qudit system with periodic boundary conditions consisting of $X$ sites by identifying site $e+X$ with site $e$ and ignoring all other sites. Because $H \in \mathcal{H}_k$, we can be assured that $X \leq n_H$. Let $\mathcal{J}'_{H,j}$ be the set of all domain wall trajectories on the size-$X$ system. Let $\mathcal{J}'_{H,j,l}$ be the subset that have $l=2l_0$ domain walls on the last time step.  Because the collision probability, denoted $Z_X$, for this $X$-qudit system must satisfy $Z_X \geq Z_{H,X}$, and here $Z_{H,X} = 2/(q^X+1)$, it must be the case that
\begin{equation}
    Z_X := \frac{2}{(q+1)^X}\left(\sum_{J' \in \mathcal{J}'_{H,j}}\text{weight}(J)\right) = \frac{2}{(q+1)^X}\sum_{l_0=0}^{X/4}\left(\sum_{J' \in \mathcal{J}'_{H,j,2l_0}}\text{weight}(J)\right)\geq \left(\frac{2}{q^X+1}\right) =: Z_{H,X} \label{eq:1DLBint1}\,.
\end{equation}
We can upper bound the contribution of all the terms with $l_0 >0$ in the above expression by the method that yielded the upper bound in \autoref{thm:1DcollisionprobUB}. The sum of those terms is upper bounded by the second term in Eq.~\eqref{eq:1Dupperboundintermediate}, that is
\begin{align}
 \frac{2}{(q+1)^X}\sum_{l_0=1}^{X/4}\left(\sum_{J' \in \mathcal{J}'_{H,j,2l_0}}\text{weight}(J)\right)&\leq Z_{H,X} (e-1)Xe^{-a(d-1)} \\
 &= \left(\frac{2}{q^X+1}\right)\frac{X}{2n_H}  \leq \frac{1}{2}\left(\frac{2}{q^X+1}\right) \label{eq:1DLBint2}\,,
\end{align}
since $X \leq n_H$. Combining Eqs.~\eqref{eq:1DLBint1} and \eqref{eq:1DLBint2}, we find a lower bound on the $l_0 = 0$ term
\begin{equation}
 \left(\sum_{J' \in \mathcal{J}'_{H,j,0}}\text{weight}(J)\right) \geq \left(\frac{(q+1)^X}{2}\right) \left(\frac{2}{q^X+1}\right)\left(1-\frac{1}{2}\right) = \left(\frac{q+1}{q}\right)^X\left(\frac{1}{2}\frac{q^X}{q^X+1}\right) \geq \left(\frac{q+1}{q}\right)^X\left(\frac{1}{3}\right) \label{eq:1DLBint3}\,,
\end{equation}
where the last inequality follows since $q\geq 2$ and $X \geq 1$ must be true.
\begin{figure}[h]
\centering
\includegraphics[width=0.64\linewidth]{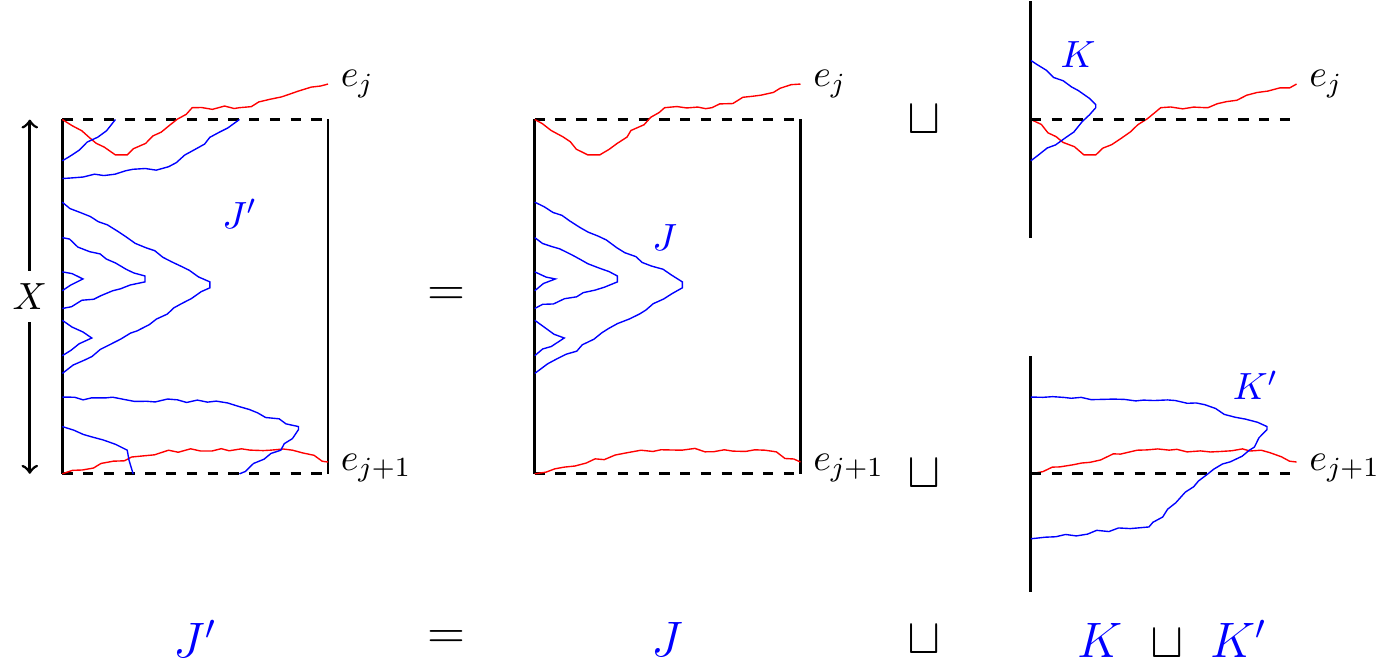}
    \caption{Outline of the argument in the proof of \autoref{lem:sumJHj} that the sum over domain wall trajectories in $\mathcal{J}_{H,j}$ is at least the sum in $\mathcal{J}'_{H,j,0}$ divided by some constant factor, expressed in Eq.~\eqref{eq:sumJJconstantfactor}. Every trajectory in $J' \in \mathcal{J}'_{H,j,0}$ can be  decomposed into a trajectory $J\in \mathcal{J}_{H,j}$ a trajectory $K \in \mathcal{B}_{H,j}$, and a trajectory $K' \in \mathcal{B}_{H,j+1}$, where each domain wall in $K$ intersects the $j$th domain wall of $H$, and each domain wall of $K'$ intersects the $(j+1)$th domain wall of $H$. Because the combined weight of all possible $K$ and $K'$ is only a constant factor, independent of $n$, the combined weight of all possible $J$ cannot be more than a constant factor smaller than the combined weight of all possible $J'$. Note that the system with $X$ sites has periodic boundary conditions in this figure.}
    \label{fig:1Dlowerboundlemmapaths}
\end{figure}
Now, every domain wall trajectory in $\mathcal{J}_{H,j}$ will also be in $\mathcal{J'}_{H,j,0}$, but the converse will not be true. Some trajectories in the latter set will have one or more domain walls that intersect with either the $j$th or the $(j+1)$th domain wall of $H$ at some time step, which is not allowed within the former set. Thus, the sum over the domain wall trajectories in $\mathcal{J}_{H,j}$ will be smaller than the sum over those in $\mathcal{J}'_{H,j,0}$, but we argue by at most some constant factor by the following argument, which is also described in \autoref{fig:1Dlowerboundlemmapaths}. Let $\mathcal{B}_{H,j}$ be the set of all trajectories in which every domain wall either intersects the $j$th domain wall of $H$ at some time step $t$, or it annihilates with a domain wall that previously intersected with the $j$th domain wall of $H$. Then any trajectory in $\mathcal{J'}_{H,j,0}$ can be formed as the disjoint union of a trajectory in $J \in \mathcal{J}_{H,j}$, a trajectory in $K \in \mathcal{B}_{H,j}$ and a trajectory in $K' \in \mathcal{B}_{H,j+1}$, to account for the parts that intersect the $j$th and $(j+1)$th domain walls. Given $J'$, the choice of $J$ for this decomposition is unique, but there may be multiple choices of $(K,K')$ for which it holds. Note also that a trajectory in $\mathcal{B}_{H,j}$ can be decomposed into individual domain wall pairs that coincide with the $j$th domain wall of $H$ at some time step $t$ and annihilate at some time step $t'$. The combined weight of all such pairs, given fixed coincidence point at $e_{H,j,t}$ is at most $(2q/(q^2+1))^{2t'}$. Summing over $t' \geq t$ we find the combined weight for all possible domain wall pairs coinciding at time step $t$ is at most
\begin{equation}
    \frac{\left(\frac{2q}{q^2+1} \right)^{2t}}{1-\left(\frac{2q}{q^2+1}\right)^2} = \frac{(q^2+1)^2}{(q^2-1)^2}\left(\frac{2q}{q^2+1} \right)^{2t} \,.
\end{equation}
There can be many domain wall pairs that intersect the $j$th domain wall of $H$, but for each value of $t$ there will either be no intersection (in which case the factor is 1) or one intersection (in which case the factor is at most the above quantity). Thus we can take the product over including or not including a domain wall at each value of $t$ and find
\begin{equation}
    \sum_{K \in \mathcal{B}_{H,j}}\text{weight}(K) \leq \prod_{t=1}^s\left(1+\frac{(q^2+1)^2}{(q^2-1)^2}\left(\frac{2q}{q^2+1} \right)^{2t}\right)\,.
\end{equation}

This implies
\begin{align}
    \left(\sum_{J' \in \mathcal{J}'_{H,j,0}}\text{weight}(J)\right) &\leq \left(\sum_{J \in \mathcal{J}_{H,j}}\text{weight}(J)\right)\left(\sum_{K \in \mathcal{B}_{H,j}}\text{weight}(K)\right) \left(\sum_{K' \in \mathcal{B}_{H,j+1}}\text{weight}(K')\right)  \\
    &\leq \left(\sum_{J \in \mathcal{J}_{H,j}}\text{weight}(J)\right) \cdot \prod_{t=1}^s \left(1+\frac{(q^2+1)^2}{(q^2-1)^2}\left(\frac{2q}{q^2+1}\right)^{2t}\right)^2 \\
    &\leq \left(\sum_{J \in \mathcal{J}_{H,j}}\text{weight}(J)\right) \cdot \exp\left( 2\sum_{t=1}^s \frac{(q^2+1)^2}{(q^2-1)^2}\left(\frac{2q}{q^2+1}\right)^{2t}\right) \\
    &\leq \left(\sum_{J \in \mathcal{J}_{H,j}}\text{weight}(J)\right) \cdot \exp\left(\frac{8q^2}{(q^2-1)^2}
    \frac{1}{1-\left(\frac{2q}{q^2+1}\right)^2}\right) \\
    &= \left(\sum_{J \in \mathcal{J}_{H,j}}\text{weight}(J)\right) \cdot \exp\left(\frac{8q^2(q^2+1)^2}{(q^2-1)^4}
    \right)\\
    &\leq \left(\sum_{J \in \mathcal{J}_{H,j}}\text{weight}(J)\right) \cdot e^{10}\,, \label{eq:sumJJconstantfactor}
\end{align}
where the last inequality follows since $q \geq 2$ and the function of $q$ inside the $\exp$ is monotonically decreasing. Combining the above with Eq.~\eqref{eq:1DLBint3}, we arrive at
\begin{equation}
    \left(\sum_{J \in \mathcal{J}_{H,j}}\text{weight}(J)\right) \geq \left(\frac{q+1}{q}\right)^X \left(\frac{1}{3e^{10}}\right) = \left(\frac{q+1}{q}\right)^X c^{-1}\,.
\end{equation}
\end{proof}

\section{Bounds for the complete-graph architecture}\label{sec:completegraphbounds}

\subsection{Proof intuition and guide}

In the following sections, we complete the proofs for upper and lower bounds of the complete-graph architecture, defined formally in \autoref{def:CGarchitecture}. The first insight about the complete-graph architecture is that all configurations with the same Hamming weight are equivalent, as there is a symmetry upon permutation of the qudits. Thus, trajectories through configuration space $\{I,S\}^n$ are reduced to trajectories through Hamming weight space $\{0,1,\ldots,n\}$.

Our upper bound will use the framework of the unbiased walk, and the lower bound will use the biased walk. Recall we can use the unbiased walk to express the collision probability $Z$ as a sum over all possible paths that the trajectory might take, working from Eq.~\eqref{eq:ZunbiasedEV}
\begin{align}
    Z &= \frac{1}{(q+1)^n} \sum_{\vec{\gamma}^{(0)}} \EV_{P_u,\vec{\gamma}^{(0)}}\left[ \left(\frac{2q}{q^2+1} \right)^{\left(\text{\# of bit flips during walk}\right)}\right] \\
    &= \frac{1}{(q+1)^n} \sum_{x=0}^n \binom{n}{x} \EV_{P_u,x}\left[ \left(\frac{2q}{q^2+1} \right)^{\left(\text{\# of bit flips during walk}\right)}\right]\,,
\end{align}
where the $\binom{n}{x}$ comes from the fact that this is the number of initial configurations with Hamming weight $x$. For the complete-graph case, $P_u$ takes on a simple form: if the current configuration is $x$, the chance that the configuration changes on the next time step is precisely the chance of finding mismatching values upon drawing a random pair of indices in $[n]$, which is given by $\frac{2x(n-x)}{n(n-1)}$, and if it does change, it is equally likely to become $x-1$ or to become $x+1$. For the biased walk $P_b$, everything is the same except that when the configuration changes, it is biased  to travel to $x-1$ with probability $q^2/(q^2+1)$. Also, in the biased case, the initial configuration is not a uniform choice over all configurations but instead distributed according to $\Lambda_b$, and the expectation in the above equation is replaced with $\EV[q^{|\vec{\gamma}^{(s)}|}]$. Thus larger Hamming weight configurations are exponentially more significant in their contribution to $Z$. 

To gain an intuition for what we expect, we first think about the biased walk, which is what we use for the lower bound. Here, the peak of the probability mass in the initial configuration $\Lambda_b$ starts around Hamming weight $x= n/(q+1)$. On average, the walk lingers for $n(n-1)/2x(n-x)$ time steps before moving, which is approximately equal to $n/2x$ when $x$ is close to 0. Due to the bias, and due to the time required to wait, the effective speed of the biased walk is 
\begin{equation}
    v_{\text{eff}}(x) = \frac{2x(n-x)}{n(n-1)}\left(\frac{q^2}{q^2+1} - \frac{1}{q^2+1}\right) \approx \frac{2x}{n}\frac{q^2-1}{q^2+1}
\end{equation}
in the direction of 0, since each time it moves, it has a $q^2/(q^2+1)$ chance of moving one unit closer to 0, but a $1/(q^2+1)$ chance of moving one unit farther away from 0. Thus, in expectation, the time it takes for the peak of the probability mass to reach value 0 is
\begin{align}
    \sum_{x=1}^{n/(q+1)}\frac{1}{ v_{\text{eff}}(x)} \approx \frac{q^2+1}{q^2-1}\frac{1}{2}n\log(n) \approx: s^*\,,
\end{align}
noting that $\sum_x\frac{1}{x} \approx \log(n)$.

This strongly suggests that $s^*$ time steps are \textit{necessary} for anti-concentration, as any less time would mean the peak of the distribution over Hamming weights at the end of the circuit will be located at some Hamming weight $y > 0$ and as a result, it will receive a significant amount of weight $q^y$ in its contribution to $Z$. This is the intuition for our lower bound.

The biased walk also gives intuition for why there is a matching upper bound. If the circuit size is a little bigger than the lower bound, we expect the peak of the distribution to have terminated at the fixed point at 0. It is still possible that the tail of the distribution, which will not yet have reached the fixed point, is too fat to for anti-concentration to have been achieved; each unit farther away from 0 results in a factor of $q$ larger contribution to $Z$, so we need the tail to be exponentially decaying if we want to be able to ignore it. This is essentially what we are able to show, albeit in a way where it might not be completely clear that this is what we have done. Intuitively, one reason we expect this exponentially decaying tail is because the effective speed slows down as you get closer to zero. This gives the tail of the distribution, which is sitting further away from 0, time to ``catch up,'' as its effective speed is faster.
 
To actually perform the upper bound, we turn back to the unbiased walk. To be clear and to match the progression in the full proof, we introduce the concept of a \textit{reduced} path 
(equivalently, ``reduced walk'') as a walk that never stands still at a certain configuration. If its Hamming weight at time step $t$ is $y$, then its Hamming weight at time step $t+1$ will move to $y-1$ or $y+1$. For each walk, we can form a corresponding reduced walk simply by removing consecutive duplicates from the sequence of configurations. Another way to look at it is that given a fixed reduced walk, the actual walk will linger at each location for a certain number of time steps before continuing. In the limit of large circuit size $s$, there is enough time for the actual walk to linger as long as it would like at each step, and any reduced walk will successfully be ``completed'' by the actual walk. In this limit, $Z = Z_H$ where $Z_H := 2/(q^n+1)$ is the Haar value. Away from this limit, there is some probability that some reduced paths will not be completed. 

We are able to express the difference between $Z$ and $Z_H$ as a sum over all reduced paths, including reduced paths that do not terminate at Hamming weight 0 or Hamming weight $n$, where the summand is proportional to the probability that the reduced path is not completed within $s$ time steps:
\begin{align}
    &(q+1)^n(Z-Z_H) \\
    ={}&  \sum_{\text{red path } \phi} \left(\frac{q}{q^2+1}\right)^{\text{length of } \phi}\frac{(q-1)^2}{2q} \Pr[\phi \text{ not completed in } s \text{ time steps}]\,.
\end{align}
The next key insight is to use a Chernoff bound to bound the probability of a reduced walk not being completed. If $L$ is the length of the walk, the Chernoff bound states (for any constant $a>0$) that
\begin{equation}
    \Pr[L > s] \leq \frac{\EV[e^{aL}]}{e^{as}}\,,
\end{equation}
but this is particularly useful because, for fixed $\phi$, $L$ is itself a sum of independent random variables $L_{\phi^{(i)}}$, the number of time steps the walk waits on step $i$. Thus
\begin{equation}
    \EV[e^{aL}] = \prod_i \EV[e^{aL_{\phi^{(i)}}}]
\end{equation}
and because each random variable $L_{\phi^{(i)}}$ is exponentially distributed, we can calculate $\EV[e^{aL_{\phi^{(i)}}}]$ exactly. For the purposes of the proof sketch, denote
\begin{equation}
    T_x := \EV[e^{aL_x}]\,,
\end{equation}
which will depend only on the Hamming weight $x$ of the configuration the walk is at. (The walk will wait longer when it is near 0 or $n$ than when it is near $n/2$.) This dependence appears to be a problem as it is unclear how to actually perform the following sum over all possible $\phi$. (The constant $a$ for each reduced walk $\phi$, denoted $a_\phi$, will be specified later.)
\begin{align}
    (q+1)^n(Z-Z_H)
    ={}& \frac{(q-1)^2}{2q}\sum_{\text{red path } \phi}e^{-a_{\phi}s}\prod_{i=1}^{\text{length of } \phi} \left(\frac{q}{q^2+1}T_{\phi^{(i)}}\right)\,.
\end{align}

To proceed, we break $\phi$ up into subpaths that inch closer and closer to 0 and $n$. We can write $\phi$ as the concatenation of $\phi_x$, $\phi_{x-1}$, \ldots, $\phi_{w}$, where $\phi_v$ begins at either $v$ or $n-v$ and only reaches $v-1$ or $n-v+1$ for the first time on the very last step. Then, $w$ is the minimum Hamming weight distance from one of the fixed points ($0$ or $n$) the reduced walk $\phi$ ever reaches. Because the walk $\phi_v$ spends all its time between $v$ and $n-v$, the expectation $T_y$ for all of the $y$ within one of these $\phi_v$ walks will be less than or equal to $T_v$ (the walk moves slower when its closer to $0$ or $n$), and we can write
\begin{equation}\begin{split}
    &(q+1)^n(Z-Z_H) \\
    \leq {}&  \frac{(q-1)^2}{2q}\sum_{x=0}^{n} \binom{n}{x}\sum_{w=0}^{\min(x,n-x)} e^{-a_w s}\left[\sum_{\phi_x} \left(\frac{q}{q^2+1}T_x\right)^{\text{length of } \phi_x}\right]\ldots \left[\sum_{\phi_w} \left(\frac{q}{q^2+1}T_w\right)^{\text{length of } \phi_w}\right]\,.
\end{split}\end{equation}
Here the $\binom{n}{x}$ comes in as the number of configurations with Hamming weight $x$, and $a_{\phi}$ has changed to $a_w$ because we will choose it so that it only depends on the end point $w$ of $\phi$. 

The above equation is huge progress because we already know how to perform the sums in brackets. Essentially, the factor of $T_x$ simply changes the \textit{effective} value of $q$; we may define $\bar{q}$ to satisfy
\begin{equation}
    \bar{q}/(\bar{q}^2+1) = qT_x/(q^2+1)\,.
\end{equation}
Then we can use the formulas for sums over paths that we have already developed in \autoref{lem:sumovertrajectories} to perform the sums. What we find is that, for the values of $a_w$ that we can choose, we  must allocate roughly $(q^2+1)n/2(q^2-1)x$ time steps for $s$ such that $e^{-a_ws}$ can cancel out the value of the sum in brackets for $\phi_x$. Note that this is precisely the inverse of the effective speed we defined before. Then, for all the sums to be canceled from $v=1$ to $v=n/2$, we must allocate
\begin{equation}
    \sum_{v=1}^{n/2} \frac{q^2+1}{2(q^2-1)}\frac{n}{x} \approx \frac{q^2+1}{2(q^2-1)}n\log(n)
\end{equation}
time steps. Fundamentally, the $\log(n)$ factor becomes necessary because the walk waits longer and longer as it gets closer and closer to the fixed points. In the full analysis, we find a term linear in $n$ is also necessary to fully anti-concentrate but our analysis of the linear term is not tight.

\subsection{Preliminaries}

\subsubsection{Trajectories}

For the complete-graph architecture, we may keep track of only the Hamming weight of a certain configuration. Thus, our random walks are over the set $\{0,1,\ldots,n\}$. A trajectory $\gamma$ is now a sequence of integers $(\gamma^{(0)},\ldots,\gamma^{(s)})$. Generally speaking, if $t > s$ for a sequence of length $s$, let $\gamma^{(t)}$ return $\gamma^{(s)}$. A sequence is valid if for every $t$, $\lvert \gamma^{(t)}-\gamma^{(t-1)}\rvert \leq 1$ and such that if $0$ or $n$ appears, it appears only once at the very end of the sequence. Let $\Gamma$ be the set of all valid trajectories.

For any valid trajectory $\gamma$ the unbiased random walk associates a non-zero probability:
\begin{equation}
    \Pr_{P_u}[\gamma] := \prod_{t=1}^s P_u[\gamma^{(t)}|\gamma^{(t-1)}]\,,
\end{equation}
where
\begin{equation}
    P_u[y|x] = 
    \begin{cases}
        \frac{x(n-x)}{n(n-1)} & \text{if } \lvert y-x\rvert = 1 \\
        1-\frac{2x(n-x)}{n(n-1)} & \text{if }  y=x \\
        0 & \text{otherwise}
    \end{cases}\,.
\end{equation}
We can make the same definition for the biased random walk by replacing $P_u$ with $P_b$ where
\begin{equation}
    P_b[y|x] = 
    \begin{cases}
        \frac{2q^{1+x-y}}{q^2+1}\frac{x(n-x)}{n(n-1)} & \text{if } \lvert y-x\rvert = 1 \\
        1-\frac{2x(n-x)}{n(n-1)} & \text{if }  y=x \\
        0 & \text{otherwise}
    \end{cases}\,.
\end{equation}
For $P \in \{P_u,P_b\}$ and any subset $\Upsilon \subseteq \Gamma$, we let $\Pr_{P}[\Upsilon] = \sum_{\gamma \in \Upsilon}\Pr_{P}[\gamma]$  be the total probability assigned to paths in $\Upsilon$. 

\subsubsection{Conditional probabilities and expectations}

For any $\gamma \in \Upsilon$ we may also define the conditional probability
\begin{equation}
    \Pr_{P}[\gamma|\Upsilon] := \frac{\Pr_{P}[\gamma]}{\Pr_{P}[\Upsilon]}\,,
\end{equation}
which indicates drawing from the subset $\Upsilon$ with probability proportional to that assigned by the (unbiased or biased) random walk. This also allows us to naturally define conditional expectation values for some quantity $Q$ computed from $\gamma$
\begin{equation}
    \EV_{P}[Q[\gamma]|\gamma \in \Upsilon] := \sum_{\gamma \in \Upsilon}\frac{\Pr_{P}[\gamma]}{\Pr_{P}[\Upsilon]}Q[\gamma]\,.
\end{equation}

\subsubsection{Trajectory concatenation and other operations }

For any trajectory $\gamma = (\gamma^{(0)},\ldots,\gamma^{(s)}) \in \Gamma$, let $L[\gamma]=s$ be the length of the trajectory, and let $\gamma^{(L)}$ be shorthand for $\gamma^{(L[\gamma])}$. The statement $w \in \gamma$ returns true if there exists some $t$ for which $\gamma^{(t)}=w$. Then, we let 
\begin{equation}
    S_w[\gamma] = 
    \begin{cases}
        \min\left(\{t:\gamma^{(t)}=w\}\right) & \text{if } w \in \gamma \\
        -1 & \text{if } w \not\in \gamma
    \end{cases}
\end{equation}
be the first time step along $\gamma$ for which the trajectory reaches $w$. We also let 
\begin{align}
    M[\gamma] &:= \max_{0\leq t \leq L[\gamma]} \gamma_t \\
    m[\gamma] &:= \min_{0\leq t \leq L[\gamma]} \gamma_t
\end{align}
be the maximum and minimum Hamming weight the trajectory passes through. 

We can naturally concatenate two trajectories $\gamma_1$ and $\gamma_2$ if $\gamma_1^{(L)} = \gamma_2^{(0)}$ to form a trajectory $\gamma = \gamma_1 \cdot \gamma_2$ of length $L[\gamma_1] + L[\gamma_2]$. We will say that $\gamma_A \subset\gamma_C$ if there exists some $\gamma_B$ for which $\gamma_A \cdot \gamma_B = \gamma_C$. 

For any trajectory $\gamma$ we let $\tilde{\gamma}$ be the flipped trajectory.
\begin{equation}
    \tilde{\gamma} := (n-\gamma^{(0)}, n-\gamma^{(1)}, \ldots, n-\gamma^{(s)})\,.
\end{equation}
In general, if $v$ is an integer with $0 \leq v \leq n$, then let $\tilde{v}:= \min(v,n-v)$. 

Similarly, let $\bar{\gamma}$ return the reversed trajectory.
\begin{equation}
    \bar{\gamma} := (\gamma^{(s)},\ldots, \gamma^{(0)})\,.
\end{equation}
Moreover, let $\gamma^{[t]}$ return the trajectory $\gamma$ truncated to length $t$, or simply return $\gamma$ if $t\geq L[\gamma]$, i.e.~
\begin{equation}
    \gamma^{[t]} :=
    \begin{cases}
        \gamma & \text{if } t \geq L[\gamma] \\
        (\gamma^{(0)},\ldots,\gamma^{(t)}) & \text{if } t < L[\gamma]
    \end{cases}\,.
\end{equation}
Let $\gamma^{[L-t]}$ be shorthand for $\gamma^{[L[\gamma]-t]}$.  More generally, let $\gamma^{[a,b]} = (\gamma^{(a)},\gamma^{(a+1)},\ldots, \gamma^{(b)})$.

\subsubsection{Important subsets of \ensuremath{\Gamma}}

We now define various subsets of $\Gamma$. Let
\begin{equation}
    \Gamma_x := \{\gamma \in \Gamma: \gamma^{(0)}= x\}
\end{equation}
be the subset of trajectories that begin at $x$, and let 
\begin{equation}
    \Gamma^w := \{\gamma \in \Gamma: S_w[\gamma] = L[\gamma]\}
\end{equation}
be the set of trajectories that reach $w$ for the first time and immediately terminate. We make the natural combination of these
\begin{equation}
    \Gamma_x^w := \{\gamma \in \Gamma: \gamma_0=x, S_w[\gamma] = L[\gamma]\}\,.
\end{equation}
Of particular importance are sets where $w=0$ or $w=n$, which include valid trajectories that terminate at one of the fixed points of the random walk. Define
\begin{equation}
    \Gamma_x^* = \Gamma_x^0 \cup \Gamma_x^n
\end{equation}
and note that $\Pr_{P_u}[\Gamma_x^*] = \Pr_{P_b}[\Gamma_x^*] = 1$, a statement that intuitively makes sense since walks will eventually reach either 0 or $n$ with probability 1.
Adding the superscript $w$ to any set $\Upsilon$ restricts to walks for which $L[\gamma] = S_w[\gamma]$. 

When any walk in $\Upsilon$ can be concatenated with any walk in $\Upsilon'$ we let
\begin{equation}
    \Upsilon \cdot \Upsilon' := \{ \gamma \cdot \gamma': \gamma \in \Upsilon, \gamma' \in \Upsilon'\}\,.
\end{equation}

Additionally, we let 
\begin{align}
    \tilde{\Upsilon} &:= \{\tilde{\gamma}: \gamma \in \Upsilon\} \\
    \bar{\Upsilon} &:= \{\bar{\gamma}: \gamma \in \Upsilon\} \,.
\end{align}

\subsubsection{Reduced trajectories}

We also introduce the concept of a \textit{reduced} trajectory, which we sometimes refer to synonymously as a reduced walk, which is a valid trajectory for which $\lvert \gamma^{(t)}-\gamma^{(t-1)}\rvert = 1$ for all $t$; that is, the reduced walk never stands still. We let the set of reduced walks be $\Psi$, and let all sub and superscripts restrict $\Psi$ in the same way they restricted $\Gamma$. For any $\gamma \in \Gamma$ we can associate a reduced walk $\psi \in \Psi$ by removing consecutive duplicates from $\gamma$. Under this definition, we let $R[\gamma] := \psi$. For any $\psi \in \Psi$ we let
\begin{equation}
    \Gamma_\psi = \{\gamma \in \Gamma: R[\gamma] = \psi, R[\psi^{[L-1]}] \neq \psi \} \,,
\end{equation}
where the second condition acts to include only trajectories $\gamma$ whose final configuration appears only once (i.e.~when the final configuration is removed, the reduced sequence changes).

Under dynamics by either the unbiased or biased walk, it is easy to calculate the probability associated with $\Gamma_\psi$:
\begin{align}
    \Pr_{P_u}[\Gamma_\psi]&:=\left(\frac{1}{2}\right)^{L[\psi]} \\
    \Pr_{P_b}[\Gamma_\psi]&:=q^{\psi^{(0)}-\psi^{(L)}}\left(\frac{q}{q^2+1}\right)^{L[\psi]}\,.
\end{align}

Finally, define the following subsets of $\Psi$:
\begin{align}
    \Lambda_x &= \Psi_{x|n-x+1}^{x-1}\cup \Psi_{x|x-1}^{n-x+1}  \label{eq:Lambdav}\\
    \Xi_w &= \{\psi \in \Psi_w: m(\psi) \geq w, M(\psi) \leq n-w\} \,,
    \label{eq:Xiw}
\end{align}
where the subset $\Psi^c_{a|b}$ is defined as follows.
\begin{align}\label{eq:Gammax|zw}
    \Psi^w_{x|z} &:=
    \begin{cases}
        \{\psi \in \Psi_x^w: M[\psi] < z\} & \text{if } w<x< z \\
        \{\psi \in \Psi_x^w: m[\psi] > z\} & \text{if } z<x<w \\
        \emptyset & \text{otherwise}
    \end{cases}\,.
\end{align}
In words, the set $\Psi_{x|z}^w$ includes reduced walks that begin at $x$ and end at $w$ without ever reaching $z$. Thus $\Lambda_x$ is the set of reduced walks that start at $x$ and end at $x-1$ without ever reaching $n-x+1$ or end at $n-x+1$ without ever reaching $x-1$. The set $\Xi_w$ is the set of reduced walks of any finite length that start at $w$ but never reach either $w-1$ or $n-w+1$.

\subsection{Upper bound proof}
\begin{theorem}[\autoref{thm:completegraphupperboundsummary} from main text]\label{thm:completegraphUB}
    For the complete-graph architecture with circuit size $s$ on $n$ qudits with local Hilbert space dimension $q$
    \begin{equation}
        Z \leq Z_H\left(1+e^{-\frac{2a}{n}(s-s^*)}\right)\,,
    \end{equation}
    as long as $s\geq s^*$, where
    \begin{align}
        s^* &= \frac{q^2+1}{2(q^2-1)}n \log(n) + O(n) \\
        a &= \frac{(q-1)^2}{2(q^2+1)} \,.
    \end{align}
\end{theorem}
\begin{proof}
    In this proof, we will be working with expressions for the collision probability $Z$. It will take several steps to manipulate the original expression into the form we need, so we will move back and forth between updating the expression and developing the tools needed to justify these updates. 
    
    We start by expressing
    \begin{align}
        Z &= \frac{1}{(q+1)^n}\sum_{x=0}^n \binom{n}{x} \EV_{P_u}\left[\left(\frac{2q}{q^2+1}\right)^{L\left[R\left[\gamma^{[s]}\right]\right]}  \; \Big| \; \gamma \in \Gamma_x^*\right]\,.
    \end{align}
    This is seen to be equivalent to Eq.~\eqref{eq:ZunbiasedEV} as follows. There are $\binom{n}{x}$ initial configurations with Hamming weight $x$, and generating a length-$s$ trajectory beginning at $x$ with the unbiased Markov chain is equivalent to randomly choosing a trajectory $\gamma$ from $\Gamma_x^*$, which begins at $x$ and ends at a fixed point ($0$ or $n$), with probability proportional to that assigned by the unbiased walk, and then truncating the walk to length $s$, denoted by $\gamma^{[s]}$. Then $R[\gamma^{[s]}]$ is the reduced trajectory, where consecutive duplicates are removed, and $L[R[\gamma^{[s]}]]$ is the length of that reduced trajectory, or in other words, the total number of bit flips that have occurred within the first $s$ time steps.
    
    Moving ahead, we observe that drawing $\gamma$ from $\Gamma_x^*$ is equivalent to first drawing a reduced trajectory $\psi$ from $\Psi_x^*$ and then drawing $\gamma$ from $\Gamma_\psi$, so we can rewrite
    \begin{align}
        Z &= \frac{1}{(q+1)^n}\sum_{x=0}^n \binom{n}{x} \sum_{\psi \in \Psi_x^*}\Pr_{P_u}[\Gamma_\psi] \EV_{P_u}\left[\left(\frac{2q}{q^2+1}\right)^{L\left[R\left[\gamma^{[s]}\right]\right]}  \; \Big| \; \gamma \in \Gamma_\psi \right]\,.
    \end{align}
    
    Now, note the following general statement about any integer-valued random variable $X$ such that $0\leq X \leq M$. For any function $f$ we have
    \begin{align}
        \EV[f(X)]&=\sum_{m=0}^M \Pr[X=m] f(m)=\sum_{m=0}^M\left(\Pr[X<m+1]-\Pr[X<m]\right)f(m) \\
        &= f(M)+ \sum_{m=1}^M\Pr[X<m] \left(f(m-1)-f(m)\right)\,.
    \end{align}
    Taking $X = L[R[\gamma^{[s]}]]$ for $\gamma$ drawn at random from $\Gamma_\psi$ and $f(X) = (2q/(q^2+1))^X$, we find
    \begin{align}
        \EV_{P_u}\left[\left(\frac{2q}{q^2+1}\right)^{L\left[R\left[\gamma^{[s]}\right]\right]}  \; \Big| \; \gamma \in \Gamma_\psi \right]
        &= \left(\frac{2q}{q^2+1}\right)^{L[\psi]}+\sum_{m=1}^{L[\psi]} \Pr_{P_u}\left[L[R[\gamma^{[s]}]] < m  \; \Big| \; \gamma \in \Gamma_\psi \right]\left(\frac{2q}{q^2+1}\right)^{m-1}\frac{(q-1)^2}{q^2+1} \\
        &= \left(\frac{2q}{q^2+1}\right)^{L[\psi]}+\sum_{m=1}^{L[\psi]} \Pr_{P_u}\left[L[\gamma] > s  \; \Big| \; \gamma \in \Gamma_{\psi^{[m]}} \right]\left(\frac{2q}{q^2+1}\right)^{m-1}\frac{(q-1)^2}{q^2+1}\,,
    \end{align}
    where the last line follows since the conditions $L[R[\gamma^{[s]}]] < m$ with $\gamma \in \Gamma_\psi$ and $L[\gamma] > s$ with $\gamma \in \Gamma_{\psi^{[m]}}$ both  correspond to deciding if the configuration has changed at least $m$ times within the first $s$ steps. 
    
    The quantity
    \begin{align}
        \frac{1}{(q+1)^n}\sum_{x=0}^n \binom{n}{x} \sum_{\psi \in \Psi_x^*}\Pr_{P_u}[\Gamma_\psi]\left(\frac{2q}{q^2+1}\right)^{L[\psi]}
    \end{align}
    is precisely equal to $Z_H$, as this represents the limit of infinite size where all trajectories terminate at one of the fixed points (see \autoref{sec:sanitycheckZH}). Thus, also noting that $\Pr_{P_u}[\Gamma_{\psi}] = 2^{-L[\psi]}$, we have
    \begin{align}
        Z &= Z_H + \frac{(q-1)^2}{(q+1)^n(q^2+1)}\sum_{x=0}^n \binom{n}{x} \sum_{\psi \in \Psi_x^*}\Pr_{P_u}[\Gamma_\psi] \sum_{m=1}^{L[\psi]} \Pr_{P_u}\left[L[\gamma] > s  \; \Big| \; \gamma \in \Gamma_{\psi^{[m]}} \right]\left(\frac{2q}{q^2+1}\right)^{m-1}\\
        &= Z_H + \frac{(q-1)^2}{(q+1)^n(2q)}\sum_{x=0}^n \binom{n}{x} \sum_{\psi \in \Psi_x^*}\Pr_{P_u}[\Gamma_\psi] \sum_{m=1}^{L[\psi]} \Pr_{P_u}\left[L[\gamma] > s  \; \Big| \; \gamma \in \Gamma_{\psi^{[m]}} \right]\left(\frac{2q}{q^2+1}\right)^{m} \\
        &= Z_H + \frac{(q-1)^2}{(q+1)^n(2q)}\sum_{x=0}^n \binom{n}{x} \sum_{\psi \in \Psi_x^*}2^{-L[\psi]} \sum_{m=1}^{L[\psi]} \Pr_{P_u}\left[L[\gamma] > s  \; \Big| \; \gamma \in \Gamma_{\psi^{[m]}} \right]\left(\frac{2q}{q^2+1}\right)^{m}\,.
\end{align}
Now, in the first line below, by associating $\phi = \psi^{[m]}$ we reorder and regroup the sums: instead of summing over paths $\psi$ that end at a fixed point and then all intermediate points $m = 1,\ldots, L[\psi]$ along the path, we first sum over all $m$, all $\phi$ (not necessarily ending at a fixed point) of length $m$, and then all $\psi$ for which $\phi \subset \psi$ (recall this means that the first $L[\phi]$ entries in the trajectory $\psi$ are equal to $\phi$). In the second line, we note that the sums over $m$ and $\phi$ of length $m$ is just a sum over all $\phi$ (of any length). In the third line, we note that the total probability of all the walks $\psi$ for which $\phi \subset \psi$ is just $2^{-L[\phi]}$. 
\begin{align}
       Z &= Z_H + \frac{(q-1)^2}{(q+1)^n(2q)}\sum_{x=0}^n \binom{n}{x} \sum_{m=1}^\infty \sum_{\substack{\phi \in \Psi_x\\ L[\phi]=m}}\left(\sum_{\substack{\psi \in \Psi_x^*\\ \phi \subset \psi}}2^{-L[\psi]} \right) \Pr_{P_u}\left[L[\gamma] > s  \; \Big| \; \gamma \in \Gamma_{\phi} \right]\left(\frac{2q}{q^2+1}\right)^{L[\phi]} \\
      &= Z_H + \frac{(q-1)^2}{(q+1)^n(2q)}\sum_{x=0}^n \binom{n}{x} \sum_{\phi \in \Psi_x}\left(\sum_{\substack{\psi \in \Psi_x^*\\ \phi \subset \psi}}2^{-L[\psi]} \right) \Pr_{P_u}\left[L[\gamma] > s  \; \Big| \; \gamma \in \Gamma_{\phi} \right]\left(\frac{2q}{q^2+1}\right)^{L[\phi]} \\
        &= Z_H + \frac{(q-1)^2}{(q+1)^n(2q)}\sum_{x=0}^n \binom{n}{x} \sum_{\phi \in \Psi_x}\left(2^{-L[\phi]} \right) \Pr_{P_u}\left[L[\gamma] > s  \; \Big| \; \gamma \in \Gamma_{\phi} \right]\left(\frac{2q}{q^2+1}\right)^{L[\phi]}  \\
        &= Z_H + \frac{(q-1)^2}{(q+1)^n(2q)}\sum_{x=0}^n \binom{n}{x} \sum_{\phi \in \Psi_x}  \Pr_{P_u}\left[L[\gamma] > s  \; \Big| \; \gamma \in \Gamma_{\phi} \right]\left(\frac{q}{q^2+1}\right)^{L[\phi]}\,.
    \end{align}
    Now we examine the final expression. The difference between $Z$ and $Z_H$ is a sum over $\Psi_x$, which includes all reduced paths $\phi$ that start at $x$ and may or may not terminate at $0$ or $n$. The statement $L[\gamma] > s$ is true if the number of time steps it takes to complete this reduced path is at least $s$, i.e.~the probability that \textit{the path does not finish within $s$ time steps}. As a sanity check, when $s$ becomes infinite, we expect this probability to become zero for any path as there would be enough time for any path to finish, and in this case $Z=Z_H$ as expected. This expression represents progress because we will be able to bound the probability of a certain path being completed using a Chernoff bound.
    
    For any random variable $X$ and for any constant $a>0$
    \begin{equation}
        \Pr[X > k] \leq \frac{\EV[e^{aX}]}{e^{ak}}\,.
    \end{equation}
    We use this bound with $X = L[\gamma]$, $k=s$, and yet-to-be-specified constants $a_{\phi} > 0$
    \begin{equation}
        Z-Z_H \leq \frac{(q-1)^2}{(q+1)^n(2q)}\sum_{x=0}^n \binom{n}{x} \sum_{\phi \in \Psi_x} e^{-a_{\phi}s} \EV_{P_u}\left[e^{a_{\phi} L[\gamma]}  \; \Big| \; \gamma \in \Gamma_{\phi} \right]\left(\frac{q}{q^2+1}\right)^{L[\phi]}\,.
    \end{equation}
 
    The Chernoff bound has the additional benefit that $\EV[e^{aX}]$ separates when $X$ is the sum of independent random variables. In particular, once $\phi$ is fixed, $L[\gamma]$ is the sum of exponentially distributed random variables corresponding to how many time steps the path $\gamma$ waits at each position along the reduced path $\phi$. 
    
    This is seen formally by noting that 
    \begin{align}
        \phi &= (\phi^{(0)}, \phi^{(1)}) \cdot (\phi^{(1)}, \phi^{(2)}) \cdot \ldots \cdot (\phi^{(L-1)}, \phi^{(L)}) \\
        \Gamma_\phi&= \Gamma_{(\phi^{(0)}, \phi^{(1)})} \cdot \Gamma_{(\phi^{(1)}, \phi^{(2)})} \cdot \ldots \cdot \Gamma_{(\phi^{(L-1)}, \phi^{(L)})} \,,
    \end{align}
    and meanwhile, for any collection of subsets $\Upsilon_m$
    \begin{align}
        \EV_{P_u}\big[e^{aL[\gamma]}\; | \; \gamma \in \Upsilon_1 \cdot \ldots \cdot \Upsilon_M\big] = \prod_{m=1}^M \EV_{P_u}\big[e^{aL[\gamma]}\; | \; \gamma \in \Upsilon_m \big]\,.
    \end{align}
    For any $r=0,\ldots,L[\phi]-1$, we can evaluate
    \begin{align}
        \EV_{P_u}\big[e^{aL[\gamma]}\; | \; \gamma \in \Gamma_{(\phi^{(r)},\phi^{(r+1)})} \big] &= \sum_{t=1}^\infty \left(1-\lambda_{\phi^{(r)}}^{-1}\right)^{t-1}\lambda_{\phi^{(r)}}^{-1} e^{at} \\
        &= \frac{1}{1-\lambda_{\phi^{(r)}}(1-e^{-a})}\,,
    \end{align}
    where
    \begin{equation}
        \lambda_v := \frac{n(n-1)}{2v(n-v)}
    \end{equation}
    is the expected amount of time the walk will wait at Hamming weight $v$ before moving to $v+1$ or $v-1$, and hence
    \begin{equation}
        \EV_{P_u}\left[e^{a_{\phi} L[\gamma]}  \; \Big| \; \gamma \in \Gamma_{\phi} \right] = \prod_{r=0}^{L[\phi]-1} \frac{1}{1-\lambda_{\phi^{(r)}}(1-e^{-a_\phi})}.
    \end{equation}
    
    We have made some progress at evaluating the bound on $Z$, but at this point it remains unclear how to perform the sum over paths $\phi \in \Psi_x$. To do so, first we will decompose paths $\phi$ into a series of subpaths that inch closer and closer to the fixed points at 0 and $n$.  In particular, we will decompose a path $\phi$ as a concatenation of subpaths drawn from $\Lambda_v$ for various $v$ and one final subpath drawn from $\Xi_w$, as described in the following lemma. Recall from Eqs.~\eqref{eq:Lambdav} and \eqref{eq:Xiw} that these subsets of $\Psi$ are defined by where they start, where they end, and/or some maximum or minimum point at which they ever reach. 

\begin{lemma}
    Suppose that $\phi \in \Psi_x$. Let $\tilde{x} = \min(x,n-x)$ and let $w = \min(m(\phi),n-M(\phi))$. Then there is a \textit{unique} sequence of trajectories $(\phi_v)_{v=w}^{\tilde{x}}$ with $\phi_v \in \Lambda_v$ for $v=w+1,\ldots, \tilde{x}$ and $\phi_w \in \Xi_w$ and such that
    \begin{equation}
        \phi = \alpha_{\tilde{x}} \cdot \alpha_{\tilde{x}-1}  \cdot \ldots \cdot \alpha_w\,,
    \end{equation}
    where for each $v$ either $\alpha_v = \phi_v$ or $\alpha_v = \tilde{\phi}_v$, depending on whether $\alpha_{v+1}$ terminates at $v$ or at $n-v$. 
\end{lemma}

\begin{proof}
    Let $r_v$ be the minimum $r$ such that $\phi^{(r)} = v$ or $\phi^{(r)} = n-v$. Then for each $v = w+1,\ldots, \tilde{x}$, we can define
    \begin{equation}
        \alpha_v = \phi^{[r_v,r_{v-1}]}
    \end{equation}
    and 
    \begin{equation}
        \alpha_w = \phi^{[r_w,L[\phi]]}\,.
    \end{equation}
    Then, each $\alpha_v$ begins at either $v$ or $n-v$ and terminates upon reaching either $v-1$ or $n-v+1$ for the first time. Hence it is a member of $\Lambda_v$ or $\tilde{\Lambda}_v$, but not both. Finally, $\alpha_w$ is a member of $\Xi_w$ because it begins at $w$ and never reaches either $w-1$ or $n-w+1$, since this would contradict the definition of $w$. 
\end{proof}

We will use the notation $\tilde{v}:= \min(v,n-v)$ for any integer $v$ throughout the remainder of the proof. The above lemma allows us to replace the sum over $\phi \in \Psi_x$ with sums over $w$ from $0$ to $\tilde{x}$ and sums over $\phi_v \in \Lambda_v$, $\phi_w \in \Xi_w$. The summand is a product of factors $(1-\lambda_{\phi^{(r)}}(1-e^{-a_\phi}))^{-1}$, each of which can be collected within just one of the sums. Moreover, the fact that these products are invariant under reversing the path, i.e.~
\begin{equation}
    \prod_{r=0}^{L[\psi]-1} f(\psi^{(r)}) = \prod_{r=0}^{L[\psi]-1} f(\tilde{\psi}^{(r)})
\end{equation}
for any function $f$, means that it is unimportant that $\alpha_v$ can equal $\phi_v$ or $\tilde{\phi}_v$ as both yield the same result.

We choose $a_\phi$ so that it only depends on $w = \min(m(\phi),n-M(\phi))$, denoted henceforth by $a_w$. Collecting these observations, and noting that the $L[\phi]$ factors of $q/(q^2+1)$ can each be allocated to one of the steps taken in $\phi$ we find that

\begin{align}
     &\frac{(q+1)^n(2q)}{(q-1)^2}\left(Z -Z_H\right) \\
     \leq{}& \sum_{x=0}^n \binom{n}{x} \sum_{w=0}^{\tilde{x}} e^{-a_w s} \prod_{v=w}^{\tilde{x}}\left(\sum_{\substack{\phi_v \in \Lambda_v \\\text{or}\\ \phi_w \in \Xi_w}}\prod_{r=0}^{L[\phi_v]-1}\frac{q}{q^2+1}\frac{1}{1-\lambda_{\phi_v^{(r)}}(1-e^{-a_w})}\right) \\
     \leq{}& \sum_{x=0}^n \binom{n}{x} \sum_{w=0}^{\tilde{x}} e^{-a_w s} \prod_{v=w}^{\tilde{x}}\left(\sum_{\substack{\phi_v \in \Lambda_v \\\text{or}\\ \phi_w \in \Xi_w}}\left(\frac{q}{q^2+1}\frac{1}{1-\lambda_{v}(1-e^{-a_w})}\right)^{L[\phi_v]}\right) \,,
\end{align}
where in the final line, we used the fact that by definition of $\phi_v \in \Lambda_v$ or $\phi_v \in \Xi_v$, $v \leq \phi_v^{(r)} \leq n-v$ for all $r < L[\phi]$ and also $\lambda_v \geq \lambda_{v'}$ whenever $\tilde{v} < \tilde{v}'$. 

This form is very useful because we know how to perform sums in parentheses, using the strategy we first saw in \autoref{lem:sumovertrajectories}. The values of these sums are given by the following lemma, whose proof is delayed until after the main proof. 

\begin{lemma}\label{lem:sumsinpars}
Given $v$ and a parameter $a$ that satisfies $1 \leq e^a \leq (1-\frac{(q-1)^2}{q^2+1}\frac{1}{\lambda_v})^{-1}$
\begin{align}
    \sum_{\alpha \in \Lambda_v}\left(\frac{q}{q^2+1}\frac{1}{1-\lambda_{v}(1-e^{-a})}\right)^{L[\alpha]} &= \bar{q}_{v,a}^{-1}\frac{1+\bar{q}_{v,a}^{-n+2v}}{1+\bar{q}_{v,a}^{-n+2v-2}}\\
    \sum_{\alpha \in \Xi_v}\left(\frac{q}{q^2+1}\frac{1}{1-\lambda_{v}(1-e^{-a})}\right)^{L[\alpha]} &= \frac{\bar{q}_{v,a}^2+1}{(\bar{q}_{v,a}-1)^2}\left(1-\bar{q}_{v,a}^{-1}\frac{1+\bar{q}_{v,a}^{-n+2v}}{1+\bar{q}_{v,a}^{-n+2v-2}}\right)\,,
\end{align}
where $\bar{q}_{v,a}$ is defined in \autoref{def:barq}.
\end{lemma}

\begin{definition}\label{def:barq}
    Given $x$ and parameter $a$ satisfying $1 \leq e^a \leq (1-\frac{(q-1)^2}{q^2+1}\frac{1}{\lambda_x})^{-1}$, let
    \begin{equation}
        \bar{q}_{x,a} = \left(\frac{q^2+1}{2q}\right)\left(1-\lambda_x(1-e^{-a})\right)\left(1+\sqrt{1-\frac{4q^2}{(q^2+1)^2(1-\lambda_x(1-e^{-a}))^2}}\right)\,.
    \end{equation}
    Note that $\bar{q}_{x,a}$ satisfies the equation
    \begin{equation}
        \frac{\bar{q}_{x,a}}{\bar{q}_{x,a}^2+1} = \frac{q}{q^2+1}\frac{1}{1-\lambda_x(1-e^{-a})}
    \end{equation}
    and that 
    \begin{equation}
        \bar{q}_{x,a}^{-1} = \left(\frac{q^2+1}{2q}\right)\left(1-\lambda_x(1-e^{-a})\right)\left(1-\sqrt{1-\frac{4q^2}{(q^2+1)^2(1-\lambda_x(1-e^{-a}))^2}}\right)\,.
    \end{equation}
\end{definition}

This lemma allows us to state
\begin{align}
    &\frac{(q+1)^n(2q)}{(q-1)^2}\left(Z -Z_H\right)\\
    \leq{}&\sum_{x=0}^{n}\binom{n}{x}\sum_{w=0}^{\tilde{x}} e^{-sa_w} \left(\prod_{v=w+1}^{\tilde{x}}\bar{q}_{v,a_w}^{-1}\frac{1+\bar{q}_{v,a_w}^{-n+2v}}{1+\bar{q}_{v,a_w}^{-n+2v-2}}\right)
    \left( \frac{\bar{q}_{w,a_w}^2+1}{(\bar{q}_{w,a_w}-1)^2}\left(1-\bar{q}_{w,a_w}^{-1}\frac{1+\bar{q}_{w,a_w}^{-n+2w}}{1+\bar{q}_{w,a_w}^{-n+2w-2}}\right)\right)  \\
     \leq{}&\sum_{x=0}^{n}\binom{n}{x}\sum_{w=0}^{\tilde{x}} e^{-sa_w} \left(\prod_{v=w+1}^{\tilde{x}}\tilde{q}_{v,a_w}^{-1}\right)e^{\frac{q^2}{q^2-1}}
    \left( \frac{\bar{q}_{w,a_w}^2+1}{(\bar{q}_{w,a_w}-1)^2}\left(1-\bar{q}_{w,a_w}^{-1}\frac{1+\bar{q}_{w,a_w}^{-n+2w}}{1+\bar{q}_{w,a_w}^{-n+2w-2}}\right)\right)\,, \label{eq:Z0ZHbeforeaw}
\end{align}
where the second line follows from the observation that (also noting $\bar{q}_{v,a} < q$ for all $v,a$)
\begin{align}
    &\prod_{v=w+1}^{\tilde{x}}\frac{1+\bar{q}_{v,a_w}^{-n+2v}}{1+\bar{q}_{v,a_w}^{-n+2v-2}} \leq \prod_{m=0}^\infty \frac{1+q^{-2m}}{1+q^{-2m-2}} = \prod_{m=0}^\infty \left(1+q^{-2m}\frac{1-q^{-2}}{1+q^{-2m-2}}\right) \\
    \leq{}& \prod_{m=0}^\infty \left(1+q^{-2m}\right) \leq \prod_{m=0}^\infty \exp\left(q^{-2m}\right) = \exp\left(\sum_{m=0}^\infty q^{-2m}\right) = \exp\left(\frac{q^2}{q^2-1}\right)\,.
\end{align}

To continue, we will make choices for $a_w$ and show upper bounds for the various factors in the above expression. For $w>0$, we make the specification for $a_w$ that
\begin{equation}
    \eta_w := 1-e^{-a_w} := \frac{(q-1)^2}{q^2+1}\frac{1}{2\lambda_w} = \frac{(q-1)^2}{q^2+1}\frac{w(n-w)}{n(n-1)}
\end{equation}
and that $\eta_0 = \eta_1$. This choice implies 
\begin{equation}
    a_w \geq \frac{(q-1)^2}{q^2+1}\frac{1}{2\lambda_w} = \frac{(q-1)^2}{q^2+1}\frac{w(n-w)}{n(n-1)}\,.
\end{equation}

Moreover, it implies that, so long as $w \leq x \leq n-w$
\begin{equation}
    1-\lambda_x(1-e^{-a_w}) \geq 1-\lambda_w(1-e^{-a_w}) = 1-\frac{(q-1)^2}{2(q^2+1)} = \frac{(q+1)^2}{2(q^2+1)}\,,
\end{equation}
which by \autoref{def:barq} implies that
\begin{equation}
    q \geq \bar{q}_{x,a_w} \geq \frac{(q+1)^2}{4q}\,.
\end{equation}
When this is the case, we have
\begin{align}
    \left( \frac{\bar{q}_{w,a_w}^2+1}{(\bar{q}_{w,a_w}-1)^2}\left(1-\bar{q}_{w,a_w}^{-1}\frac{1+\bar{q}_{w,a_w}^{-n+2w}}{1+\bar{q}_{w,a_w}^{-n+2w-2}}\right)\right) &\leq \frac{\bar{q}_{w,a_w}^2+1}{(\bar{q}_{w,a_w}-1)^2}\left(1-\bar{q}_{w,a_w}^{-1}\right) \\
    &= \frac{\bar{q}_{w,a_w}^2+1}{(\bar{q}_{w,a_w}-1)\bar{q}_{w,a_w}}\\
    &\leq \frac{q^2+1}{\left(\frac{(q+1)^2}{4q}-1\right)\frac{(q+1)^2}{4q}}\\
    &= \frac{(q^2+1)(4q)^2}{(q-1)^2(q+1)^2}\\
    &\leq \frac{320}{9}\,,
\end{align}
where the last line follows for all $q\geq 2$, which will be true for any physically realizable circuit. This takes care of the final factor in Eq.~\eqref{eq:Z0ZHbeforeaw}. What remains are the factors of $\bar{q}_{v,a_w}^{-1}$. To handle these we will use the following bound, whose proof is delayed to the next section.

\begin{lemma}\label{lem:qbarinvbound}
With $\bar{q}_x$ defined as in \autoref{def:barq} and as long as $1 \leq e^a \leq (1-\frac{(q-1)^2}{q^2+1}\frac{1}{\lambda_x})^{-1}$,
\begin{equation}
    \bar{q}_{x,a}^{-1} \leq q^{-1}\exp\left(a\left(\lambda_x \frac{q^2+1}{q^2-1}+\lambda_x^2(1-e^{-a})\frac{(q^2+1)^4}{(q^2-1)^3}\right)\right)\,.
\end{equation}
\end{lemma}

We also need the following observation, which holds under the assumption that $1 \leq j<k\leq \frac{n}{2}$
\begin{align}
    \sum_{r=j+1}^k \lambda_{r} &\leq \frac{n}{2}\sum_{r=j+1}^k \left(\frac{1}{r} + \frac{1}{n-r}\right) \leq \frac{n}{2}\left(\int_{j}^k d\rho\frac{1}{\rho} + \int_{n-k-1}^{n-j-1}d\rho\frac{1}{\rho}\right) \\
    &= \frac{n}{2}\left(\log(k/j) +\log\!\big((n-j-1)/(n-k-1)\big)\right)\\
    &\leq \frac{n}{2}\left(\log(k/j) +\log(2)\right) < \frac{n}{2}\left(\log(2k/j)+1\right)\,.\\
    \sum_{r=j+1}^k \lambda_{r}^2 &\leq \frac{n^2}{4}\sum_{r=j+1}^k \frac{n^2}{r^2(n-r)^2} \leq n^2\sum_{r=j+1}^k \frac{1}{r^2}
    \leq \frac{n^2}{j} < \frac{ \pi^2 n^2}{6j}\,.
\end{align}
Similarly, for the case where $j=0$, we find
\begin{equation}
    \sum_{r=1}^k \lambda_r \leq \frac{n}{2}\big(\log(2k) + 1\big) \,, \qquad \sum_{r=1}^k \lambda_r^2 \leq \frac{\pi^2n^2}{6}\,.
\end{equation}
Now we can write the following, where $\bar{w} := \max(w,1)$
\begin{align}
    &\exp\left(-\frac{q^2}{q^2-1}\right)\frac{9(q+1)^n(2q)}{320(q-1)^2}\left(Z -Z_H\right)\\
    \leq{}& \sum_{x=0}^{n}\binom{n}{x}\sum_{w=0}^{\tilde{x}} e^{-sa_w} \left(\prod_{v=w+1}^{\tilde{x}}\bar{q}_{v,a_w}^{-1}\right)  \\
    \leq{}& \sum_{x=0}^{n}\binom{n}{x}\sum_{w=0}^{\tilde{x}} e^{-sa_w}  \left(\prod_{v=w+1}^{\tilde{x}} q^{-1}\exp\left(\frac{q^2+1}{q^2-1}a_w\lambda_v + \frac{(q^2+1)^4}{(q^2-1)^3}a_w\eta_w \lambda_v^2\right)\right) \\
    ={}& \sum_{x=0}^{n}\binom{n}{x}\sum_{w=0}^{\tilde{x}} e^{-sa_w} q^{-\tilde{x}+w}\exp\left(\frac{q^2+1}{q^2-1}a_w\sum_{v=w+1}^{\tilde{x}}\lambda_v + \frac{(q^2+1)^4}{(q^2-1)^3}a_w\eta_w \sum_{v=w+1}^{\tilde{x}}\lambda_v^2\right) \\
    \leq{}& \sum_{x=0}^{n}\binom{n}{x}\sum_{w=0}^{\tilde{x}} e^{-sa_w} q^{-\tilde{x}+w}\exp\left(\frac{q^2+1}{q^2-1}a_w\frac{n}{2} \left(\log\frac{2\tilde{x}}{\bar{w}}+1\right) + \frac{(q^2+1)^4}{(q^2-1)^3}a_w\eta_w \frac{n^2\pi^2}{6\bar{w}}\right) \\
    \leq{}& \sum_{x=0}^{n}\binom{n}{x}\sum_{w=0}^{\tilde{x}} e^{-sa_w} q^{-\tilde{x}+w}\exp\left(a_w\left(\frac{q^2+1}{q^2-1}\frac{n}{2} \left(\log\frac{2\tilde{x}}{\bar{w}} +1\right) + \frac{(q^2+1)^3(q-1)^2}{(q^2-1)^3} \frac{n\pi^2(n-w)}{6(n-1)}\right) \right)\,. \label{eq:Z0ZHnearlydone}
\end{align}

Now we are in a position to nearly complete the proof. We choose
\begin{equation}
    s^* = \frac{1}{2}\frac{q^2+1}{q^2-1}n\log(n) + cn\,,
\end{equation}
where
\begin{align}
    c=\frac{q^2+1}{2(q^2-1)}+\frac{(q^2+1)^3(q-1)^2}{(q^2-1)^3} \frac{\pi^2}{6} + \frac{q^2+1}{(q-1)^2}\left(\log\left(\frac{320(q-1)(q^n+1)}{9q^n}\right)+\frac{q^2}{q^2-1}+4\log(q)\right)\,.
\end{align}
The $O(n\log(n))$ term in $s^*$ will be necessary to cancel the $O(n\log(2\tilde{x}/w))$ term in Eq.~\eqref{eq:Z0ZHnearlydone}. Meanwhile, the first two terms of $c$ will be needed to cancel the remaining terms on the right-hand-side of Eq.~\eqref{eq:Z0ZHnearlydone}. The next two terms in $c$ are used to cancel factors on the left-hand-side of Eq.~\eqref{eq:Z0ZHnearlydone}, and finally the last term is vital for canceling the $q^w$ factor, as follows
\begin{align}
    &\frac{q^n+1}{q^n}(q+1)^n(Z-Z_H) \\
    \leq{}& \frac{q-1}{q}\sum_{x=0}^{n}\binom{n}{x}\sum_{w=0}^{\tilde{x}} e^{-(s-s^*)a_w} q^{-\tilde{x}+w}\exp\left(-\frac{4w(n-w)}{n-1}\log(q) \right) \\
    \leq{}& \frac{q-1}{q}\sum_{x=0}^{n}\binom{n}{x}\sum_{w=0}^{\tilde{x}} e^{-(s-s^*)a_w} q^{-\tilde{x}-w} \\
    \leq{}& e^{-(s-s^*)a_1}\frac{q-1}{q}\sum_{x=0}^{n}\binom{n}{x}q^{-\tilde{x}}\sum_{w=0}^{\tilde{x}} q^{-w} \label{eq:pulledouta1}\\
    \leq{}& e^{-(s-s^*)a_1}\sum_{x=0}^{n}\binom{n}{x}q^{-\tilde{x}} \\
    \leq{}& e^{-(s-s^*)a_1}\sum_{x=0}^{n}\binom{n}{x}\left(q^{-x}+q^{-n+x}\right) \\
    ={}& 2e^{-(s-s^*)a_1}\left(\frac{q+1}{q}\right)^n\,.
\end{align}
It was in Eq.~\eqref{eq:pulledouta1}, where we used $a_1 \leq a_w$ to pull the exponential out from the sum, that the assumption $s \geq s^*$ was necessary. This is the only place it has been needed. As $Z_H = 2/(q^n+1)$, we then have
\begin{equation}
    Z \leq Z_H(1+e^{-(s-s^*)a_1})\,,
\end{equation}
where $a_1 = (q-1)^2/(n(q^2+1))$. Defining $a = na_1/2$, this completes the proof of the upper bound.
\end{proof}

\subsection{Lower bound proof}

\begin{theorem}[\autoref{thm:completegraphlowerboundsummary} from main text]\label{thm:completegraphLB}
    For the complete-graph architecture of size $s$ on $n$ qudits with local dimension $q$, the collision probability satisfies
    \begin{equation}\label{eq:complete-graphLB}
        Z \geq \frac{Z_H}{2} \exp\left(\frac{\log(q)}{q+1}\exp\left(\log(n)+s\log\left(1-\frac{2(q^2-1)}{n(q^2+1)}\right)\right)\right)\,.
    \end{equation}
\end{theorem}
\begin{corollary}
    For the complete-graph architecture, let $s_{AC}$ be the minimum circuit size, as a function of $n$, such that
    \begin{equation}
        Z \leq 2 Z_H\,.
    \end{equation}
    Then it must hold that
    \begin{equation}
        \left\lvert s_{AC}- \frac{q^2+1}{2(q^2-1)}n\log(n)\right\rvert = O(n)\,.
    \end{equation} 
\end{corollary}
\begin{proof}
    The upper bound on $Z$ in \autoref{thm:completegraphUB} implies
    \begin{equation}
        s_{AC} \leq \frac{q^2+1}{2(q^2-1)}n\log(n) + O(n)\,.
    \end{equation}
    Meanwhile, since $s=s_{AC}$ implies $Z \leq 2Z_H$, the bound in \autoref{thm:completegraphLB} implies that
    \begin{equation}
        \exp\left(\frac{\log(q)}{q+1}\exp\left(\log(n)+s_{AC}\log\left(1-\frac{2(q^2-1)}{n(q^2+1)}\right)\right)\right) \leq 4
    \end{equation}
    and thus
    \begin{align}
        s_{AC} &\geq \frac{\log(n)-\log\left(\frac{(q+1)\log(4)}{\log(q)}\right)}{-\log\left(1-\frac{2(q^2-1)}{n(q^2+1)}\right)}\\
        &\geq \left(\log(n)-\log\left(\frac{(q+1)\log(4)}{\log(q)}\right)\right)\left( \frac{n(q^2+1)}{2(q^2-1)} - 1\right) \\
        &= \frac{n(q^2+1)}{2(q^2-1)}\log(n) - O(n)\,,
    \end{align}
where we have used the general inequality $-1/\log(1-u) \geq 1/u-1$.
\end{proof}

\begin{proof}[Proof of \autoref{thm:completegraphLB}]
    The structure of the proof is very similar to \autoref{thm:generalLB} for general architectures. We use the framework of the biased random walk.
    
    Let $x:= \gamma^{(t)}$. The transition rule is such that
    \begin{equation}
        \gamma^{(t+1)} = 
        \begin{cases}
        x & \text{with probability } 1-\frac{2x(n-x)}{n(n-1)} \\
        x-1 & \text{with probability } \frac{2x(n-x)}{n(n-1)}\frac{q^2}{q^2+1} \\
        x+1 & \text{with probability } \frac{2x(n-x)}{n(n-1)}\frac{1}{q^2+1} \\
        \end{cases}
    \end{equation}
    and so 
    \begin{align}
        \EV_{P_b}[\gamma^{(t+1)} | \gamma^{(t)} = x] &= x-\frac{2x(n-x)}{n(n-1)}\frac{q^2-1}{q^2+1} \\
        &\geq x\left(1- \frac{2(q^2-1)}{n(q^2+1)}\right)\,.
    \end{align}
    As this is true for all $x$, when we have some probability distribution $\Lambda$ over values of $x$, it still holds that
    \begin{align}
        \EV_{P_b}[\gamma^{(t+1)}] &= \sum_{x=0}^n \Pr_{\Lambda}[\gamma^{(t)} = x] \EV_{P_b} [\gamma^{(t+1)} | \gamma^{(t)} = x] \\
        &\geq \sum_{x=0}^n \Pr_{\Lambda}[\gamma^{(t)} = x] x\left(1- \frac{2(q^2+1)}{n(q^2-1)}\right) \\
        &= \left(1- \frac{2(q^2+1)}{n(q^2-1)}\right) \sum_{x=0}^n \Pr_{\Lambda}[\gamma^{(t)} = x] x \\
        &= \left(1- \frac{2(q^2-1)}{n(q^2+1)}\right) \EV_{\Lambda}[\gamma^{(t)}]\,,
    \end{align}
    and by applying this equation recursively from the starting distribution $\Lambda_b$, we find
    \begin{align}
        \EV_{P_b,\Lambda_b}[\gamma^{(s)}] &\geq \left(1- \frac{2(q^2-1)}{n(q^2+1)}\right)^s \EV_{\Lambda_b}[\gamma^{(0)}] \\
        &= \left(1- \frac{2(q^2-1)}{n(q^2+1)}\right)^s \frac{n}{q+1}\,.
    \end{align}
    By convexity, we have $\EV[q^x] \geq q^{\EV[x]}$, and hence
    \begin{align}
        Z &= \frac{1}{q^n}\EV_{P_b,\Lambda_b}[q^{\lvert \vec{\gamma}^{(s)}\rvert}] \\
        &\geq \frac{1}{q^n}\exp\left(\log(q)\frac{n}{q+1}\left(1- \frac{2(q^2-1)}{n(q^2+1)}\right)^{s}\right) \\
        &\geq \frac{Z_H}{2} \exp\left(\frac{\log(q)}{q+1}\exp\left(\log(n)+s\log\left(1-\frac{2(q^2-1)}{n(q^2+1)}\right)\right)\right)\,.
    \end{align}
\end{proof}

\subsection{Delayed proofs of lemmas}

\begin{proof}[Proof of \autoref{lem:sumsinpars}]
The first equation will follow fairly straightforwardly from \autoref{lem:sumovertrajectories}. Note that the factor in parentheses on the left-hand-side is $\bar{q}_{v,a}/(\bar{q}_{v,a}^2+1)$ as defined in \autoref{def:barq}, as well as the fact that $\Lambda_v$ contains walks that start at $v$ and end at $v-1$ or $n-v+1$. The walks that start at $v$ and end at $v-1$ are covered by \autoref{lem:sumovertrajectories} with $x\rightarrow v$, $y\rightarrow v-1$, and $m\rightarrow n-2v+2$. The walks that start at $v$ and end at $n-v+1$ are equivalent to walks starting at $n-v$ and ending at $v-1$, and are thus covered by \autoref{lem:sumovertrajectories} with $x\rightarrow n-v$, $y\rightarrow v-1$, and $m\rightarrow n-2v+2$. Summing the results from these two substitutions yields the quantity
\begin{align}
    &\frac{1}{1-\bar{q}_{v,a}^{-2(n-2v+2)}}\left(\bar{q}_{v,a}^{-1}-\bar{q}_{v,a}^{-2n+4v-3}\right) + \frac{1}{1-\bar{q}_{v,a}^{-2(n-2v+2)}}\left(\bar{q}_{v,a}^{-n+2v-1}-\bar{q}_{v,a}^{-n+v-3}\right)= \bar{q}_{v,a}^{-1}\frac{1+\bar{q}_{v,a}^{-n+2v}}{1+\bar{q}_{v,a}^{-n+2v-2}}\\
\end{align}
which proves that the first equation in the Lemma is correct.

The second equation is not a direct application of \autoref{lem:sumovertrajectories}, but it can be shown by a similar method. Define
\begin{equation}
    \Xi_{x,v} = \{\psi \in \Psi_x: v\leq m(\psi), M(\psi) \leq n-v\}
\end{equation}
so $\Xi_v = \Xi_{v,v}$. Moreover, for fixed $v$, let
\begin{equation}
    I(x) :=     \sum_{\alpha \in \Xi_{x,v}}\left(\frac{q}{q^2+1}\frac{1}{1-\lambda_{v}(1-e^{-a})}\right)^{L[\alpha]}  = \sum_{\alpha \in \Xi_{x,v}}\left(\frac{\bar{q}_{v,a}}{\bar{q}_{v,a}^2+1}\right)^{L[\alpha]} \,.
\end{equation}
The function $I(x)$ obeys the recursion relation
\begin{equation}
    I(x) = 1+ \frac{\bar{q}_{v,a}}{\bar{q}_{v,a}^2+1}\left(I(x-1)+I(x+1)\right)\,,
\end{equation}
since there is one term in the sum corresponding to the length-0 trajectory, which contributes 1, but all other terms appear either in $I(x-1)$ or $I(x+1)$ reduced by factor $\bar{q}_{v,a}/(\bar{q}_{v,a}^2+1)$. The general solution to this recursion relation is 
\begin{equation}
    I(x) = \frac{\bar{q}_{v,a}^2+1}{(\bar{q}_{v,a}-1)^2} + A\bar{q}_{v,a}^x + B\bar{q}_{v,a}^{-x}
\end{equation}
for some constants $A$ and $B$. Here is where we rely on boundary conditions. We must have $I(v-1) = I(n-v+1)=0$ since these sums do not include any terms. This allows us to solve for $A$ and $B$ and find
\begin{align}
    A &=-\frac{\bar{q}_{v,a}^2+1}{(\bar{q}_{v,a}-1)^2}\frac{\bar{q}_{v,a}^{-n+v-1}}{1+\bar{q}_{v,a}^{-n+2v-2}}\\
    B &= -\frac{\bar{q}_{v,a}^2+1}{(\bar{q}_{v,a}-1)^2}\frac{\bar{q}_{v,a}^{v-1}}{1+\bar{q}_{v,a}^{-n+2v-2}}\\
    I(v)&= \frac{\bar{q}_{v,a}^2+1}{(\bar{q}_{v,a}-1)^2}\left(1-\bar{q}_{v,a}^{-1}\frac{1+\bar{q}_{v,a}^{-n+2v}}{1+\bar{q}_{v,a}^{-n+2v-2}}\right)\,.
\end{align}

\end{proof}

\begin{proof}[Proof of \autoref{lem:qbarinvbound}]
Define $\eta = \lambda_x(1-e^{-a})$ and let $\zeta = 1-(1-\eta)^2$. Thus, $1-\eta = \sqrt{1-\zeta}$.

We can write
\begin{align}
    \bar{q}_{x,a}^{-1} &= q^{-1}\left(\frac{q^2+1}{2}\right)\left(\sqrt{1-\zeta}\right)\left(1-\sqrt{1-\frac{4q^2}{(q^2+1)^2(1-\zeta)}}\right) \\
    &= q^{-1}\left(\frac{q^2+1}{2}\right)\left(\sqrt{1-\zeta}-\frac{q^2-1}{q^2+1}\sqrt{1-\frac{(q^2+1)^2}{(q^2-1)^2}\zeta}\right) \\
    &\leq  q^{-1}\left(\frac{q^2+1}{2}\right)\left(1-\frac{1}{2}\zeta-\frac{q^2-1}{q^2+1}\left(1-\frac{(q^2+1)^2}{2(q^2-1)^2}\zeta-\frac{(q^2+1)^4}{2(q^2-1)^4}\zeta^2\right)\right)\\
    &= q^{-1}\left(1+\frac{1}{2}\frac{q^2+1}{q^2-1}\zeta + \frac{(q^2+1)^4}{4(q^2-1)^3}\zeta^2\right) \\
    &\leq q^{-1}\exp\left(\frac{1}{2}\frac{q^2+1}{q^2-1}\zeta + \frac{(q^2+1)^4}{4(q^2-1)^3}\zeta^2\right) \\
    &\leq q^{-1}\exp\left(\frac{q^2+1}{q^2-1}a\lambda_x + \frac{(q^2+1)^4}{(q^2-1)^3}a\eta \lambda_x\right)\,,
\end{align}
which is equal to the lemma statement, where in the first inequality we utilized $1-\frac{u}{2} - \frac{u^2}{2} \leq \sqrt{1-u} \leq 1-\frac{u}{2}$, in the second inequality we used $1+u \leq \exp(u)$, and in the third inequality we used $\zeta \leq 2 \eta \leq 2a\lambda_x$. The condition on $a$ is necessary to ensure that $\bar{q}_x$ is real. 
\end{proof}

\section{Approximate 2-designs and anti-concentration}
\label{app:2design}

In this appendix we clarify the relation between approximate unitary 2-designs and anti-concentration. As we discussed in the text, forming a unitary 2-design is a sufficient condition for anti-concentration.

First we recall some definitions. The $k$-fold channel of an operator $\op$ with respect to a probability distribution $\mu$ on the unitary group $\mathcal{U}(q^n)$ is defined as
\begin{equation}
    \Phi^{(k)}_\mu(\op) := \int d\mu(U)\, U^{\otimes k} (\op) U^\dagger{}^{\otimes k}\,.
\end{equation}
We denote the channel with respect to the Haar measure on the unitary group as $\Phi^{(k)}_H$.
The diamond norm of a quantum channel $\Phi$ is defined as $\|\Phi\|_\diamond := \sup_{\psi,D} \|\Phi\otimes\mathcal{I}_D(\psi)\|_1 $, where $\mathcal{I}_D$ is the identity channel on a $D$-dimensional ancilla and $\psi$ is a state on the entire system.

\begin{definition}[Approximate designs]
    A probability distribution $\mu$ on $\mathcal{U}(q^n)$ is an $\ep$-approximate unitary $k$-design if the $k$-fold channels obey
    \begin{equation}
        \big\| \Phi^{(k)}_\mu - \Phi^{(k)}_H \big\|_\diamond \leq \ep\,.
    \end{equation}
\end{definition}
\ni For a given $k$, if $\ep=0$ we say that the distribution forms an exact $k$-design.

A weaker notion of approximate design involves the operator norm of the moment operators, sometimes referred to as the tensor product expander (TPE) condition. The vectorization isomorphism uniquely maps channels to operators, with which we can define the $k$th moment operator from the $k$-fold channel for a probability distribution $\mu$ on the unitary group $\mathcal{U}(q^n)$ as
\begin{equation}
    \widehat\Phi^{(k)}_\mu := {\rm vec}\big( \Phi^{(k)}_\mu) = \int d\mu(U)\, U^{\otimes k}\otimes U^*{}^{\otimes k}\,.
\end{equation}
For convenience we denote $U^{\otimes k,k} := U^{\otimes k}\otimes U^*{}^{\otimes k}$.
\begin{definition}[Weak approximate designs]
    A probability distribution $\mu$ on $\mathcal{U}(q^n)$ is a weak $\ep$-approximate unitary $k$-design if the $k$th moment operators obey
    \begin{equation}
        \big\| \widehat \Phi^{(k)}_\mu - \widehat\Phi^{(k)}_H \big\|_\infty \leq \ep\,.
    \end{equation}
\end{definition}

The expectation of the collision probability for completely Haar-random unitaries is $Z_H=\BE_H[Z] = 2/(q^n+1) \leq 2/q^n$, and thus anti-concentrates with $\alpha=1/2$ as defined in \autoref{def:anticoncentratedArch}. But as the collision probability is a second moment quantity, where $p_U(x)^2 = |\bra{x}U\ket{1^n} |^4$, for an exact unitary 2-design $\mu$ we find
\begin{equation}
    Z = \EV_\mu \bigg[\sum_{x} p_U(x)^2\bigg] = \EV_H \bigg[\sum_{x} p_U(x)^2\bigg] = \frac{2}{q^n+1}\,,
\end{equation}
and thus also $1/2$-anti-concentrates, where $\BE_H[\cdot]$ denotes the expectation with respect to the Haar measure on the unitary group. 

\begin{proposition}\label{prop:ACdesigns}
    An $\ep$-approximate 2-design $\mu$ with $\ep=1/q^{2n}$ has a collision probability of $Z = \BE_\mu [\sum_x p_U(x)^2] \leq 3/q^n$ and is thus a $1/3$-anti-concentrator.
    Moreover, the same holds for a weak $\ep$-approximate 2-design (TPE) $\mu$ with $\ep=1/q^{2n}$.
\end{proposition}
\begin{proof}
For an $\ep$-approximate 2-design in diamond norm, we find
\begin{align}
    \EV_\mu \big[p_U(x)^2\big] &= \EV_\mu \big[|\vev{x|U|1^n}|^4\big] - \EV_H \big[|\vev{x|U|1^n}|^4\big] + \EV_H \big[|\vev{x|U|1^n}|^4\big]\\
    &= \Tr \Big(\ketbra{x}^{\otimes 2}\Big( \EV_\mu \big[U^{\otimes 2}(\ketbra{1^n})U^\dagger{}^{\otimes 2}\big] - \EV_H \big[U^{\otimes 2}(\ketbra{1^n})U^\dagger{}^{\otimes 2}\big]\Big)\Big)
    + \frac{2}{q^n(q^n+1)}\\
    &\leq \Big\|\ketbra{x}^{\otimes 2}\big(\Phi_\mu^{(2)}(\ketbra{1^n}) - \Phi_H^{(2)}(\ketbra{1^n})\big)\Big\|_1
    + \frac{2}{q^n(q^n+1)}\\
    &\leq \big\|\ketbra{x}^{\otimes 2}\big\|_\infty \big\|(\Phi_\mu^{(2)}-\Phi_H^{(2)})(\ketbra{1^n})\big\|_1
    + \frac{2}{q^n(q^n+1)}\\
    &\leq \frac{2}{q^n(q^n+1)} + \ep\,,
\end{align}
where we wrote the difference in terms of the 2-fold channels, in the second to last line used H\"older's inequality, and in the last line used the definition of the diamond norm and the definition of an $\ep$-approximate 2-design. 

Given the definition of an approximate design in terms of the diamond norm, we must take the error to be exponentially small. Thus, for an approximate 2-design $\mu$ with $\ep=1/q^{2n}$, the collision probability is $Z \leq 3/q^n$ and thus $1/q^{2n}$-approximate unitary 2-designs in diamond norm anti-concentrate with $\alpha=1/3$. 

For a weak $\ep$-approximate 2-design in operator norm (TPE), we proceed similarly,
\begin{align}
    \EV_\mu \big[p_U(x)^2\big] &= \EV_\mu \big[|\vev{x|U|1^n}|^4\big] - \EV_H \big[|\vev{x|U|1^n}|^4\big] + \EV_H \big[|\vev{x|U|1^n}|^4\big]\\
    &= \Tr \Big(\ketbra{1^n x}^{\otimes 2}\Big( \EV_\mu \big[U^{\otimes 2,2}\big] - \EV_H \big[U^{\otimes 2,2}\big]\Big)\Big)
    + \frac{2}{q^n(q^n+1)}\\
    &\leq \Big\|\widehat\Phi_\mu^{(2)}-\widehat\Phi_H^{(2)}\Big\|_\infty
    + \frac{2}{q^n(q^n+1)}\\
    &\leq \frac{2}{q^n(q^n+1)} + \ep\,,
\end{align}
where we wrote the difference in terms of the 2-fold moment operators, in the second to last line used H\"older's inequality, and in the last line used the definition of a weak $\ep$-approximate 2-design. Again, we must take the error to be exponentially small. For $\ep=1/q^{2n}$, the collision probability is $Z \leq 3/q^n$ and thus $1/q^{2n}$-approximate unitary 2-designs in operator norm anti-concentrate with $\alpha=1/3$. 
\end{proof}

As $n$-qudit RQCs on the 1D architecture are known to form $\ep$-approximate unitary 2-designs in $O(n+\log(1/\ep))$ depth \cite{BH13,BHH2016RQCtdesign}, anti-concentration for 1D random circuits in linear depth is an immediate corollary. Moreover, an $n$-independent upper bound on the spectral gap for the 1D architecture \cite{BHH2016RQCtdesign}, implies that they form weak approximate 2-designs in $O(\log(1/\ep))$ depth. By \autoref{prop:ACdesigns}, where we must take $\ep=1/q^{2n}$, this again requires linear depth for 1D RQCs.

For non-local RQCs on the complete-graph architecture, the best known upper bounds on the 2-design circuit size are $O(n^2)$ \cite{HarrowLow2009RQC2design}. However, it has been conjectured that this can be improved to $O(n\log(n))$, in which case anti-concentration and 2-designs could occur at the same depth for the complete-graph circuit architecture. 

To argue that anti-concentration must be distinct from the 2-design property, we consider lower bounds on the 2-design depth for RQCs on the 1D architecture. The spectral gap of the second moment of a probability distribution $\nu$ on the unitary group is defined as $g(\nu):= \| \widehat \Phi^{(2)}_{\nu} - \widehat \Phi^{(2)}_{H}\|_\infty$. Ref.~\cite{BHH2016RQCtdesign} proved an $n$-independent bound on the spectral gap for 1D RQCs. This implies that the behavior of the spectral gap for 1D RQCs of depth $d$ must be $g(\nu_{\rm 1D\,RQC}) = (1-1/c)^d$, for some constant $c>1$. Further recalling that the operator norm can be written as $\|M\|_\infty = \max_{y} \vev{y|M|y} $, this implies that some states requires linear depth in order to become small. 
Specifically, there is some state $\ket y$ on the 4-fold space which requires the 1D circuit depth to be at least $d=\Omega(n)$ in order for the second moment operator for 1D RQCs $\EV_{U} \big[ \vev{y|U^{\otimes 2,2}|y}\big]$ to approach the minimal Haar value.

\bibliographystyle{utphys}
\bibliography{RQCanticon}

\end{document}